\documentclass[12pt]{amsart}
\usepackage{amsmath, amssymb, amsthm, amsfonts, enumerate, verbatim, mathtools,enumitem, bbm}
\usepackage{pgf,tikz,pgfplots}
\usepackage[justification=centering]{caption}
\pgfplotsset{compat=1.15}
\usepackage{mathrsfs}
\usepackage{graphicx}
\usetikzlibrary{arrows,positioning}
\usepackage[all]{xy}
\usepackage{tikz}
\usepackage{subfig}
\usepackage{xcolor, soul, todonotes}
\usepackage[hypertexnames=false,debug,linktocpage=true,hidelinks]{hyperref}

\hypersetup{
	colorlinks,
	linktoc=all,
	linkcolor={blue},
	citecolor={red},
	urlcolor={blue}
}
\tikzstyle{punkt}=[circle, fill=black, minimum size=1mm,inner sep=0pt, draw]

\newtheorem{Theorem}{Theorem}[section]
\newtheorem{Lemma}[Theorem]{Lemma}
\newtheorem{Corollary}[Theorem]{Corollary}
\newtheorem{Proposition}[Theorem]{Proposition}

\newtheorem{Example}[Theorem]{Example}

\newtheorem{Question}[Theorem]{Question}

\let\epsilon\varepsilon
\let\kappa=\varkappa
\def\pnt{{\raise0.5mm\hbox{\large\bf.}}}

\newif\ifistoreview
\istoreviewtrue

\begin{document}
	\title{A Fundamental Probabilistic Fuzzy Logic Framework Suitable for
		Causal Reasoning}
\author[Amir Saki]{Amir Saki}
\address{Amir Saki, School of Mathematics, Institute for Research in Fundamental Sciences (IPM), P.O. Box 19395-5746, Tehran, Iran.}
\email{amirsaki1369@ipm.ir, amir.saki.math@gmail.com}
\author[Usef Faghihi]{Usef Faghihi\,*}
\thanks{* Corresponding author}
\address{Usef Faghihi, Department of Mathematics and Computer Science, The University of Quebec at Trois-Rivieres, 3351 Bd des Forges, Trois-Rivières, QC G8Z 4M3, Canada}
\email{usef.faghihi@uqtr.ca}
	
	\begin{abstract}
    In this paper, we introduce a fundamental  framework to create a bridge between Probability Theory  and Fuzzy Logic. Indeed, our theory formulates a random experiment of  selecting crisp elements with the criterion of having a certain fuzzy attribute. To do so, we associate some specific crisp random variables to the random experiment. Then,  several formulas are presented, which make it easier to compute different conditional probabilities and expected values of these random variables. Also, we provide measure theoretical basis for our probabilistic fuzzy logic framework. Note that in our theory, the probability density functions of continuous distributions which come from the aforementioned random variables include the  Dirac delta function as a term. Further, we introduce an application of our theory in Causal Inference.    
	\end{abstract}

	
	\keywords{Probability theory, fuzzy logic, probabilistic fuzzy logic, cusal inference, fuzzy average treatment effect, qualitative research, fuzzy attributes}

\maketitle	

\tableofcontents
	
	\section{Introduction}\label{introduction}

	First, we briefly explain our motivations to create a Probabilistic Fuzzy Logic~(PFL) framework. Deep Learning Algorithms  with their incredible achievements such as very high precision results in some specific tasks are at the center of the weak Artificial Intelligence  \cite{faghihi2020cog}. Deep Learning Algorithms  fail when it comes to reasoning as pointed out by Madan et al \cite{madan2021fast}, Pearl \cite{pearl2018book}, and Faghihi et al \cite{faghihi2020association}. In order to equip the machine with reasoning (e.g., Causal Reasoning), Faghihi et al \cite{faghihi2020cog,faghihi2020association} suggest that we need to improve Non-Classical Logics'  drawbacks and then integrate them with Deep Learning Algorithms. 
	
	In this paper, we focus on an integration of Fuzzy Logic and Probability Theory in order to create a PFL framework.  Let's see a simple problem that can be solved by our PFL.  Assume that we intend to  allow only datapoints with the attribute {\it high} to be selected by a system (such as a neural network). Here, {\it high} is vague and could be considered as a {\it fuzzy attribute}. We suggest a  probabilistic fuzzy approach to deal with the aforementioned  problem (see Section \ref{Problem Setting and Some Motivations}). Also, using our PFL theory, we can answer questions such as the following:

\begin{center}
	\begin{minipage}{14cm}
			{\it"Given a fuzzy attribute (i.e., low) of $X= \{1,...,10\}$, what is the probability of a randomly selected element from $X$ being equal to $3$, and $3$ is low?"}
	\end{minipage}
\end{center}

We believe that properly answering the above question could help the field of Causal Reasoning. Roughly speaking, we associate a random variable $\xi_{X,low}$ to the above question, and the answer is $\mathbb{P}(\xi_{X,low}=3)$. We use this type of random variables to define a {\it fuzzy average treatment effect} formula in Section \ref{FATE}. 

	Apart from problems concerning Deep Learning Algorithms, our theory could also be used in Qualitative Research and Quantitative Research (which is out of the scope of this paper) \cite{leavy2017research}. For instance, assume that a market researcher has conducted a survey. The questionnaire includes different types of questions such as demographic questions,  multiple choice questions, and Likert scale questions \cite{brace2018questionnaire}. A Likert scale question is used to measure beliefs and opinions in a scale such as integers from 0 to 20 \cite{likert1932technique}. The market researcher could fuzzify the interval $[0,20]$ by using different fuzzy attributes such as {\it low}, {\it medium} and {\it high}. Then to select for instance {\it high} responses, one could apply  our PFL theory (see Section \ref{Problem Setting and Some Motivations}). 
       
     Now, we explain the organization of this paper. First, we provide the related works and their drawbacks in Section \ref{Related Works}. In this section, we divide  most known previous PFLs into three categories: 1) Zadeh's theory, 2) Meghdadi and Akbarzadeh's theory, 3) Buckley's theory. Indeed, to the best of our knowledge, other PFLs have been created following one of the above categories. In Section \ref{Probability Theory}, we provide some basic ideas of Probability Theory, especially the main philosophical interpretations of Probability Theory. We dedicate this section to the subjective nature of some probability distributions discussed in the next sections. Also, we justify why we use a measure theoretical approach to construct our PFL. Section~\ref{Fuzzy Preliminaries} is devoted to some basic concepts in the fuzzy literature required for the rest of this paper. Indeed, we define a fuzzy set and explain what we mean by a fuzzy attribute. 
     In Section \ref{cuasalinference}, some terminology of Causal Inference are discussed. Then, {\it the fundamental problem of Causal Inference} is given, especially the average treatment effect as an {\it estimand} for the treatment effect of an {\it intervention} is defined. Further, the {\it ignorability condition} is explained as an assumption to overcome  the fundamental problem of Causal Inference. Section \ref{Problem Setting and Some Motivations} is devoted to some motivations and the problem setting of this paper. Indeed, one problem is how to feed a system by numerical values that satisfy a fuzzy attribute. Another problem is the assignment mechanism in Causal Inference in the case that we intend to assign for instance {\it low}, {\it medium} or {\it high} treatments to the units. 
      In Section \ref{frame}, we build up our PFL framework in discrete case. Indeed, we define a two-steps random experiment, called Experiment~($\star$). Then,  for any element $x$ and a fuzzy attribute $A$ of the sample space, we associate a binary random variable $\xi_{x,A}$ to the second step of the experiment. Also, we associate a random variable $\xi_{X,A}$ to the whole process, where $X$ is the identity random variable defined on the sample space. Further, we provide several formulas for expected values and conditional probability distributions derived from these newly defined random variables. In Section~\ref{Some Models of our PFL Framework}, we introduce several possible models derived from our PFL framework. Since $\mathbb{P}(x\text{ is }A)$ is subjective in our theory, it could be based on the degree of "$x$ satisfying $A$". Hence, some of the introduced models are based on {\it membership degrees}, which come from Fuzzy Logic. Further, we explain how to obtain $\mathbb{P}(y\text{ is }B|x\text{ is }A)$ in the so-called {\it standard models}.  Section \ref{FATE} is devoted to the application of our PFL framework in Causal Inference. In this section, we explain an assignment mechanism using $\xi_{T,A}$, where $T$ is a treatment and $A$ is a fuzzy attribute. Further, we define a fuzzy version of the average treatment effect with respect to any pair belonging to for instance $\{low, medium, high\}$. This fuzzy version is based on the idea stating that to measure the causal effect of a treatment it suffices to measure the difference between the patients' conditions when they receive  {\it high} treatments versus {\it low } treatments. We also have an alternative for the ignorability condition to solve the fundamental problem of Causal Inference supported by providing an  experiment. In Section \ref{Measure Theoretical Approach}, we provide a measure theoretical basis for our theory, which is followed by  creating our PFL framework in continuous case. Indeed, the probability density function of $\xi_{X,A}$ includes the  Dirac delta function as a term, when $X$ is a continuous random variable. In Section \ref{Relationship Between Zadeh's PFL and our Theory}, we explore the relationship between Zadeh and our PFL frameworks  by addressing some main definitions. Next,  Section \ref{Conclusion} is devoted to the conclusion of this paper. In this section, we briefly justify the purposes of this paper and compare our constructed   PFL framework to the previous ones. Our future works and also some recommendations for other researchers are provided in this section as well.  Finally, we provide some appendices as a map to define Probability Theory as an application of Measure Theory. To do so,  we start with  Appendix \ref{1}, where we introduce Topology, which makes it possible to define Borel sets and measures. Then, in Appendix~\ref{MeasureTheory}, we very briefly  explain everything in Measure Theory required to define different basic concepts in  Probability Theory. In Appendix~\ref{Probability Theory (Measure Theoretical Point of View)}, we explain measure theoretical definitions of different basic concepts in Probability Theory such as probability measures, random variables, expected values, joint distributions, conditional distributions, discrete random variables and continuous random variables. In Appendix~\ref{gpdf}, we define {\it generalized distributions}, which include both discrete and continuous distributions by using the Dirac delta function. 
    
   \section{Related Works}\label{Related Works}
    Most of the current PFLs are based on the followings:
    \begin{enumerate}
    	\item Zadeh \cite{zadeh1968probability}, assumed that the events are fuzzy and defined the probability of an event using its membership function. However, Zadeh’s theory cannot answer the "and being equal to $3$" part of the following question: 
    	\begin{Question}
    		Given a fuzzy attribute (i.e., low) of $X= \{1,...,10\}$, what is the probability of a randomly selected element from $X$ being equal to $3$, and $3$ is low?
    	\end{Question}
     In \cite{gudder1998fuzzy}, Gudder followed Zadeh's work and constructed his PFL by postulating the imprecision in observing outcomes using some terminology in Quantum Mechanics.  Hence, he replaced events and random variables in crisp probability theory by fuzzy events (or {\it effects}) and fuzzy random variables (or {\it observables}), respectively. In this theory, a probability function on a $\sigma$-algebra of effects is interpreted as a {\it state} of the system. Finally, Gudder provides some applications of his PFL in Quantum Mechanics and  Pattern Recognition. Gudder's work suffers from the same drawbacks as Zadeh's theory.

     To compare some concepts of Zadeh's PFL to our PFL, see Section \ref{Relationship Between Zadeh's PFL and our Theory}.
    	\item  Meghdadi and Akbarzadeh \cite{meghdadi2001probabilistic} assumed that in probabilistic fuzzy logic the membership function is a random variable. For reasoning, the authors considered probabilistic fuzzy rules followed by applying a Mamdani's min-max system. However, in this theory there is no thorough explanation on how to handle membership functions in fuzzy logic as random variables in Probability Theory. In addition, there is no clear explanation of where to use the membership functions as random variables in their proposed reasoning system. 
   
    \item
    	In \cite{buckley2005fuzzy}, Buckley used crisp probabilities to calculate probabilities as fuzzy numbers. More concretely,  he determined the fuzzy  numbers corresponding to the probabilities of the events by giving the freedom to select crisp probabilities which are the probabilities of the sample points. Indeed, Buckley defined the probability as a triangular-shaped fuzzy number using confidence intervals for vagueness of the events \cite{ersel2016fuzzy}. However, in this theory, computing the probability of an event strongly depends on selecting the sample points out of that event, resulting in complex computations including some constraints. 
       \end{enumerate}

   \section{Probability Theory}\label{Probability Theory}
   In this section, we briefly explain main different interpretations in Probability Theory. Then, we justify why  we use measure theory  in this paper. There are two main types of interpretation for Probability Theory, one is $physical$ and another one is based on $evidences$. These two categories are briefly introduced in the next two following subsections.  To study more about different interpretations of Probability Theory, see \cite{gillies2012philosophical} or for a brief explanation, see  \cite{lyon2010philosophy}.
   \subsection{Physical Interpretations}
   These types of interpretations are based on random physical systems, and hence they are {\it objective}. There are different types of physical interpretations each having  its own branches. In the following, we briefly explain the two main categories of physical interpretations.
   
   \subsubsection{Relative Frequencies Interpretation} 
   Based on this interpretation, to determine the probability of an event $E$ (a set of some possible outcomes) in a random experiment, we conduct a sequence of the random experiment and compute the relative frequency of observing $E$ after each trial. This interpretation is based on that these relative frequencies tend to a unique number, which will be considered as the probability of observing $E$. 
   Let's  clarify the frequency point of view by a simple example. In tossing a coin experiment, the probabilities of appearing Head or Tail are determined by many times repeating the experiment. Indeed,  the probability of appearing Head is approximated to be  the number of times we observe Head divided by the number of times we toss the coin.
   
   \subsubsection{Propensity Interpretation}
   
   This interpretation states that the probabilities of events come from physical properties and their causal connections with the random experiment. In this interpretation, the physical properties of the events directly influence the outcome. For instance, tossing an unfair coin reveals  different probability values for Head and Tail, since it comes from its physical properties.  This interpretation is connected to the relative frequency interpretation in the way that we might compute probability values by relative frequencies. This does not happen always, since it might not be possible to repeat an experiment (e.g., the probability of a war between two certain nations). The motivation of this approach is Quantum Mechanics in problems such as the probability that a radioactive atom decays after a certain time. 
   
  \subsection{Evidential Interpretations}
  
    In contrast, the evidential probability interpretation associates probability values to an event based on how current evidences support the occurrences of the event.
  
  \subsubsection{Subjective Interpretation}\label{subjective}
    
    In this interpretation,   also called the Bayesian interpretation, probability values are determined by the degrees of beliefs. It is assumed that the degrees of beliefs are somehow $coherent$ in a process that simulates any probability problem to a lottery. Here, being coherent means trying not to lose.  Since we would like to use experiences, the Bayes rule is essential in this interpretation. Some situations are mostly fit to subjective interpretations compared to the others. For instance, the probability of someone being robbed is a kind of degree of belief and could be more and more accurate by considering more and more evidences. 
    
    \subsubsection{Logical Interpretation}
 
 Consider the following types of inferences: $Deductive$ and $Inductive$. In the deductive inference, premises could certainly determine the truth of a conclusion. In contrast, it is not possible in some situations to determine the truth of a conclusion, while in inductive inference we can find the probability of occurring such a conclusion. Logical interpretations deals with inductive logical systems. 
 
 \subsection{Kolmogorov Axioms}
 
 These axioms are the foundations of Probability Theory and make it possible to use mathematical tools in this area.  
 These axioms state that a function on the set of events of a random experiment is a probability measure, whenever it satisfies the followings:
 
 \begin{enumerate}
 	\item 
 	$P(E)\ge 0$ for any event $E$.
 	\item
 	$P(\Omega)=1$, where $\Omega$ is the set of all possible outcomes (or the sample space).
 \item
 $P\left(\bigcup_{n=1}^{\infty}E_n\right)=\sum_{n=1}^{\infty}P(E_n)$ for any family $\{E_n\}_{n=1}^{\infty}$ of events. 
 \end{enumerate}
Several fundamental results can be derived from the above axioms. For instance, the second and the third axioms yield that $P(\Omega\backslash E)=1-P(E)$, where throughout of this paper by $\Omega\backslash E$ we mean the complement of $E$ in $\Omega$. 

\subsection{Measure Theoretical Probability Theory}

Measure theory is a broad branch of mathematics mostly developed to generalize the concepts of length, area, volume, and integral. However, the Kolmogorov axioms are the properties of a measure with the additional condition that the measure of the whole space is $1$. We have provided a detailed map in the appendices, which starts with Topology, and continues with Measure Theory and Probability Theory to see how precisely all these are connected to each other. 

\subsection{Our Approach}

In this paper, we use the measure theoretical approach as a basis for our PFL. While probability interpretations are philosophical point of views, Probability measures must satisfy Kolmogorov axioms no matter which probability interpretation  has been used. Otherwise, we are not able to systematically work on Probability Theory  using suitable math tools.  Besides, in a classical point of view, probability distributions are either discrete or continuous, while in this paper we frequently deal with mixed ones. A measure theoretical approach provides a framework to deal with all possible probability distributions.
This is why we choose the general approach of working with measure theory.

\section{Fuzzy Preliminaries}\label{Fuzzy Preliminaries}

In this section, first we briefly talk about $fuzzy$ $sets$.  Then, we explain what we mean by a  $fuzzy$ $attribute$.  Finally, we  see some basic information about $t$-$norms$, which are operators on fuzzy sets and essential for the rest of this paper. 
\subsection{Fuzzy Sets}
Let $\Omega$ be a set and $A$ be a subset of $\Omega$. One could associate a function $\mathbbm{1}_A:\Omega\to\{0,1\}$, which sends any element of $A$ to $1$ and the other elements to $0$. We note that there is a one-to-one correspondence between $\mathcal{P}(\Omega)$, the set of all subsets of $\Omega$ (or the powerset of $\Omega$), and the set of all functions $f:\Omega\to\{0,1\}$. Hence, we can identify a subset $A$ of $\Omega$ with $\mathbbm{1}_A$. A generalization of the concept of a subset comes from investigating  the above identification. Indeed, we call any function $f:\Omega\to[0,1]$ a fuzzy subset of $\Omega$, where $[0,1]$ is the set of all real numbers not lower than $0$ and not greater than $1$. 

\subsection{Fuzzy Attributes}
In practice, fuzzy sets usually come from fuzzy attributes. 
Roughly speaking, we define a fuzzy attribute as a vague attribute without precise boundaries. For instance, $high$ could be considered as an attribute of the interval $[0,10]$, and it is weird if we say that $8$ is $high$ but $7.99$ is not! Now, let $\Omega$ be a set and $A$ be a fuzzy attribute of $\Omega$. Then, we can associate a fuzzy set $\mu_A:\Omega\to[0,1]$, called the $membership$ $function$ of $A$, in such a way that we interpret $\mu_A(\omega)$ as the $membership$ $degree$ of $\omega$ to be $A$. For an example of fuzzy attributes and their corresponding membership functions, see Figure \ref{fuzzyattributes}.

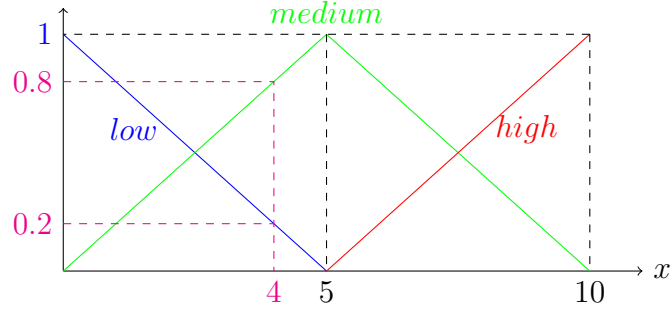
\begin{figure}
	\begin{center}
		\begin{tikzpicture}[scale=.7]
			\draw[->](0,0)--(11,0)node[right]{$ x $};
			\draw[->](0,0)--(0,5);
			\draw[blue] (0,4.5)node[left]{$1$} -- (2,2.7) node[left]{$low$}--(5,0);
			\draw[green] (0,0) --(5,4.5)node[above]{$medium$}-- (10,0);
			\draw[red] (5,0) -- (8,2.7)node[right]{$high$}--(10, 4.5);
			\draw[dashed](5,4.5) --(5,0)node[below]{$5$};
			\draw[dashed](10,4.5) --(10,0)node[below]{$10$};
			\draw[dashed] (10,4.5)-- (0, 4.5);
			\draw[dashed,magenta](0,0.9)node[left]{$0.2$}--(4,.9)--(4,0)node[below]{$4$};
			\draw[dashed, magenta](0,3.6)node[left]{$0.8$}--(4,3.6)--(4,0.9);
		\end{tikzpicture}
	\end{center}
	\caption{\small Fuzzy attributes "$low$", "$medium$" and "$high$" for $\Omega=[0,10]$. Note that $\omega=4$ is $low$, $medium$ and $high$ with the membership degrees of $0.2$, $0.8$ and $0$, respectively.}\label{fuzzyattributes}
\end{figure}

\subsection{$t$-norms}\label{tnorms}
Let $\Omega$ be a set. Assume that $A$ and $B$ are two fuzzy attributes of $\Omega$. A natural way to associate a membership function which satisfies both attributes $A$ and $B$ is using $t$-norms. Indeed, a $t$-norm is a function $T:[0,1]\times[0,1]\to[0,1]$ satisfying the following properties for any $a,b,c,d\in[0,1]$:
\begin{itemize}
	\item $T$ is commutative: 
	$T(a,b)=T(b,a)$,
	\item
	T is associative: 
	$T(a,T(b,c))=T(T(a,b),c)$,
	\item
	$1$ is an identity element: 
	$T(a,1)=a$, and
	\item
	$T$ is an increasing function in the following sense: 
	$T(a,b)\le T(c,d)$ whenever $a\le c$ and $b\le d$.
\end{itemize}
 In the following, several $t$-norms are given:

\[\begin{array}{lll}
	T_{\min}(a,b)=\min\{a,b\} && \text{minimum $t$-norm}\\
	T_{\text{prod}}(a,b)=ab && \text{product $t$-norm}\\
	T_{\text{Luk}}(a,b)=\max\{0, a+b-1\} && \text{Łukasiewicz $t$-norm}\\
	&&\\
	T_{\text{D}}(a,b)=\left\{\begin{array}{ll} b &, a=1\\ a& , b=1\\ 0 &, \text{otherwise}\end{array}\right. && \text{Drastic $t$-norm}\\
	&& \\
	T_{\text{nM}}(a,b)=\left\{\begin{array}{ll}\min\{a,b\}&, a+b>1\\ 0 & \text{otherwise}\end{array}\right. && \text{ Nilpotent Minimum $t$-norm}\\
	&& \\
	T_{\text{H}_0}(a,b)=\left\{\begin{array}{ll}0 & a=b=0\\ \dfrac{ab}{a+b-ab}&, \text{otherwise}\end{array}\right. && \text{Hamacher Product $t$-norm}
\end{array}\]
Also, there are parametric $t$-norms such as the followings:
\[
\begin{array}{ll}
	T_p^{\text{AA}}(x,y)=\left\{\begin{array}{ll}
		T_{\text{D}}(x,y)&, p=0\\
		\exp(-(|\log x|^p+|\log y|^p)^{\frac{1}{p}})&, 0<p<\infty\\
		T_{\min}(x,y)&, p=\infty
	\end{array}\right. &\text{(Aczél–Alsina $t$-norms)}
	\\ & \\
	T_p^{\text{SW}}(x,y)=\left\{\begin{array}{ll}
		T_{\text{D}}(x,y)&, p=-1\\
		\max\left\{0,\dfrac{x+y-1+pxy}{1+p}\right\}&, -1<p<\infty\\
		T_{\text{prod}}(x,y)&, p=\infty
	\end{array}\right. &\hspace*{-.35cm}\text{(Sugeno–Weber $t$-norms)}
\end{array}
\]
Now, let $T$ be a $t$-norm. Then, we can define the fuzzy set associated to the fuzzy attribute "$A\,\&\, B$" (or "$A\cap B$") as follows:
\[\mu_{A\,\&\, B}(\omega)=T(\mu_A(\omega), \mu_B(\omega)).\]

\section{Causal Inference}\label{cuasalinference}

"In Causal Inference, one reasons to the conclusion that something is, or is likely to be, the cause of something else \cite{causalinference}." For example, a researcher may want to find out if a specific pill is effective against insomnia. The {\it causal effect} of the pill against insomnia could be computed as the difference between the patient's condition in the following cases: taking the pill, and not taking the pill. Let's see some terminologies here. In the above example, we have the following definitions:

\begin{itemize}
	\item 
	The patient is called a {\it unit}, that is a member of a sample. We note that samples are selected from a target population.
	\item
	Taking the pill is called the {\it (active) treatment} or the intervention.
	\item
	Not taking the pill is called the {\it control}.
	\item
	The set of all units that receive the treatment  is called the {\it treatment group}.
	\item
	 The set of all units under control is called the {\it control group}. 
	\item
	The patient's condition for both  treatment and control states are called {\it potential outcomes}.
\end{itemize}
To have a significant result in a causal inference, we may have the following assumption:
\begin{itemize}
	\item 
	({\it Consistency Condition}) We assume that the treatment level for all units in the treatment group is the same qualitatively and quantitatively.  
\end{itemize}
We note that to measure the causal effect of the treatment for a unit, we have to do the following procedure:
\begin{itemize}
	\item 
	The unit receives the treatment at the time $t_1$ and the potential outcome associated to the treatment is measured at the time $t_2$.
	\item
	The unit receives the control at the time $t_1$ and the potential outcome associated to the control is measured at the time $t_2$.
\end{itemize}
However, it is not possible to measure both treatment and control potential outcomes for a unit simultaneously.  Hence, for a period of time, only one of the above potential outcomes could be measured for the unit, and the other one remains missing. This issue is called the  fundamental problem of Causal Inference.
 The measured potential outcome is called the {\it observed} or the {\it factual outcome}, and the missing potential outcome is called the {\it counterfactual outcome}. 

Now, consider the problem of measuring the causal effect of the pill against insomnia in a large enough population. To do so, we can measure the average treatment effect of taking the pill against insomnia. 

Let's fix some notations. We denote the random variables associated to the treatment, the potential outcome corresponding to the active treatment, the potential outcome corresponding to the control, and the observed outcome by $T$, $Y(1)$, $Y(0)$ and $Y$, respectively. 
Hence, the average treatment effect of taking the pill against insomnia is $\mathbb{E}(Y(1)-Y(0))$. To solve the issue arising from the fundamental problem of Causal Inference, one way is considering the  following assumption:
\begin{itemize}
	\item 
	(Ignorability Condition) We assume that taking the treatment for a unit does not effect the potential outcomes for that unit and the others. In other words, $Y(0)$ and $Y(1)$ are independent given $T$ (one might see some slightly different definitions for the ignorability condition in other contexts). 
\end{itemize}
Then, a {\it random assignment} mechanism could be used for which a group of people are  selected as the treatment group and the others  as the control group.  

Now, let us assume that the igorability condition is satisfied. Then,
\begin{align*}
	\mathbb{E}(Y(1)-Y(0))&=\mathbb{E}(Y(1))-\mathbb{E}(Y(0))\\
	&=\mathbb{E}(Y(1)|T=1)-\mathbb{E}(Y(0)|T=0)\\
	&=\mathbb{E}(Y|T=1)-\mathbb{E}(Y|T=0).
\end{align*}
The latter is a statistical formula without any missing data. 

While, the above concepts and formulas are provided for binary treatments, Causal Inference  deals with non-binary treatments as well. Hence, the aforementioned concepts and formulas could be generalized to include non-binary treatments.  We  recommend  a fuzzy based approach using our PFL theory in Section \ref{FATE}.

To study more about Causal Inference, see \cite{imbens2015causal} for a precise and sophisticated reference,  and see  \cite{pearl2009causal} for a graphical causal  framework.

\section{Problem Setting and Some Motivations}\label{Problem Setting and Some Motivations}

In this section, we provide some probabilistic  fuzzy nature problems and some motivations for our PFL theory. In Subsections \ref{InputswithFuzzyAttributes} and \ref{FuzzyBasedAssignment Machanism}, we discuss two different motivational problems for our framework.  To deal with the aforementioned problems, we define a certain property, called Golden Property (see Subsection~\ref{Golden Property}). To solve the above problems, we then discuss two different points of views in Subsections \ref{FuzzyBasedAssignment Machanism} and \ref{OurFramework}. Finally, we justify our solution in Subsection \ref{OurFramework}, which is the subject of the rest of this paper.

\subsection{Inputs with Fuzzy Attributes}\label{InputswithFuzzyAttributes}

Let us assume that we are given a system and a collection of data to  feed the system as input. Assume that the fuzzy nature of data is important, and we only want to feed the system with data satisfying a fuzzy attribute. For instance, a datapoint $x_0$ could have the fuzzy attribute $high$ with a degree $0.7$ of membership. Here, we have to select   $x_0$ as $high$ or not select it as $high$, and we do not have another choice. 
Clearly, $x_0$ is not completely $high$ or completely $\neg\, high$. 

\subsection{Diamond Property}\label{DiamondProperty}
Let $X$ and $Y$ be two random variables. We say that the conditional distribution of $X$ given $Y=y$ satisfies Diamond Property with respect to the $t$-norm $T$, whenever we have that
\begin{equation}\label{diamondequation}
	\mathbb{P}(X=x|Y=y)=\frac{T(\mathbb{P}(X=x), \mathbb{P}(Y=y))}{\mathbb{P}(Y=y)}.
\end{equation}
We refer to this property as Diamond Property.
Note that Diamond Property mostly fails when we check the Kolmogorov axioms and the total law of probability.  We recall that the total law of probability for $X$ given $Y$ is the following:
\[\mathbb{P}(X=x)=\sum_{y}\mathbb{P}(X=x|Y=y)\mathbb{P}(Y=y)\]
for any value $X=x$.
For example, let $X$ and $Y$ be two Bernoulli random variables with the success chance of $0.6$ and $0.7$, respectively. Then, by considering the minimum $T$-norm, we have that
\begin{align*}
&\frac{\min\{\mathbb{P}(X=1), \mathbb{P}(Y=1)\}}{\mathbb{P}(Y=1)}=\frac{\min\{0.6, 0.7\}}{0.7}=\frac{0.6}{0.7},\\
&\frac{\min\{\mathbb{P}(X=0), \mathbb{P}(Y=1)\}}{\mathbb{P}(Y=1)}=\frac{\min\{0.4, 0.7\}}{0.7}=\frac{0.4}{0.7},\\
&\frac{\min\{\mathbb{P}(X=1), \mathbb{P}(Y=0)\}}{\mathbb{P}(Y=0)}=\frac{\min\{0.6, 0.3\}}{0.3}=1,\\
&\frac{\min\{\mathbb{P}(X=0), \mathbb{P}(Y=0)\}}{\mathbb{P}(Y=0)}=\frac{\min\{0.4, 0.3\}}{0.3}=1.
\end{align*}
If the conditional distribution of $X$ given $Y=1$ (or $Y=0$) satisfies Diamond Property with respect to the minimum $t$-norm, then
\begin{align*}
&\mathbb{P}(X=1|Y=1)+	\mathbb{P}(X=0|Y=1)=\frac{1}{0.7}>1,\\
&\mathbb{P}(X=1|Y=0)+	\mathbb{P}(X=0|Y=0)=2>1,
\end{align*}
a contradiction! Further, if we normalize the above conditional distributions  dividing each of the probability values by the sum of probability values, then we obtain the followings:
\begin{align*}
	&\mathbb{P}(X=1|Y=1)=0.6, &\mathbb{P}(X=0|Y=1)=0.4,\\
	&\mathbb{P}(X=1|Y=0)=0.5,&	\mathbb{P}(X=0|Y=0)=0.5.
\end{align*}
Now, we note that
the total law of probability for $X$ given $Y$ fails. For instance, we have that
\[0.6=\mathbb{P}(X=1)\neq \underbrace{\mathbb{P}(X=1|Y=1)}_{0.6}\underbrace{\mathbb{P}(Y=1)}_{0.7}+\underbrace{\mathbb{P}(X=1|Y=0)}_{0.5}\underbrace{\mathbb{P}(Y=0)}_{0.3}=0.57.\] 
\subsection{Golden Property}\label{Golden Property}

Let's say  Equality~(\ref{diamondequation}) is satisfied for most but not all values of $X$. 
In the case that we are interested in satisfying Equality~(\ref{diamondequation}) for all except for one value of $X$, we define the following property. Let $X$ and $Y$ be two random variables. We say that the conditional distribution of $X$ given $Y=y$ satisfies Golden Property with respect to $X=x_0$ and a $t$-norm $T$, whenever Equality ($\ref{diamondequation}$) holds for all values of $X$ except for $x_0$. In other words, we must have
\[\mathbb{P}(X=x|Y=y)=\frac{T(\mathbb{P}(X=x), \mathbb{P}(Y=y))}{\mathbb{P}(Y=y)},\quad x\neq x_0.\]
One could see that the example in Subsection \ref{DiamondProperty} does not satisfy Golden Property with respect to any value of $X$. See Section \ref{Some Models of our PFL Framework} to construct many examples satisfying Golden Property but not Diamond Property.

Now, we briefly explain the reasons that make us interested in Diamond Property and Golden Property. Also, we justify why we assume that $T$ in Equation (\ref{diamondequation}) is a $t$-norm. Let $X$, $Y$ and $Z$ be three random variables whose distributions come from subjective random experiments. In general, the conditional distribution of $X$ given $Y$ could be complicated. However, based on the subjective nature of the random selections, the subject might consider the aforementioned conditional distribution as a function of the distributions of $X$ and $Y$.  To look for such a  setting that satisfies the Bayes rule, one could use the  expression appeared in Equation (\ref{diamondequation})  for $\mathbb{P}(X=x|Y=y)$, where $T:[0,1]\times[0,1]\to[0,1]$ is a function. It follows that $\mathbb{P}(X=x,Y=y)=T(\mathbb{P}(X=x),\mathbb{P}(Y=y))$ for any values $X=x$ and $Y=y$. Obviously, $T$ is commutative on the Cartesian product of the probability values of $X$ and $Y$.  Further,   for instance if $\mathrm{Supp}(Y)=\{y\}$, then $\mathbb{P}(Y=y)=1$, and hence $1$ is the identity element for the probability values of $X$  with  respect to  $T$. Furthermore, if the conditional distribution of any pair from $\{X,Y,Z\}$ satisfies Equation~(\ref{diamondequation}), then $T$ is associative on the probability values of $X$, $Y$ and $Z$. We note that the above conditions are satisfied for $T$ on a subset of $[0,1]$ (the probability values of the above distributions). 
If  the above conditions are satisfied for the whole interval $[0,1]$, then $T$   differs with the definition of a $t$-norm only on the increasing condition. Now, to integrate Fuzzy Logic and Probability Theory and due to the importance  of $t$-norms in Fuzzy Logic, we might  consider that $T$ is a $t$-norm. 

\subsection{Fuzzy Based Assignment Machanism}\label{FuzzyBasedAssignment Machanism}

Assume that we wish to randomly assign the values of a random variable $T$ to a target population in such a way that a fuzzy attribute $A$ of $T$ is the criterion for selecting $T=t$ in the assignment mechanism . To do so, one could consider a probability function $q_A$ such that $q_A(t)$ is the probability of randomly selecting $t$ as $A$ among all possible values of $T$.  Intuitively, this makes sense but Diamond Property mostly fails. Now, let's work on Golden Property. Note that in Golden Property, one value of $T$ must be distinct comparing to the other values (i.e, $x_0$ in Subsection \ref{Golden Property}). To have such a distinct value, one could   extend the sample space of the above random experiment to include "selecting nothing" as $A$. More precisely, we extend the values of $T$ to include a value $t_A$ representing "nothing", and then we associate a random variable $\xi$ to the above experiment in such a way that $\xi=t$ with a probability of $q_A(t)$. Note that

\[\mathbb{P}(T=t_A|Y=y)=1-\sum_{\substack{T=t,\\ t\neq t_A}}\frac{T(\mathbb{P}(T=t), P(Y=y))}{\mathbb{P}(Y=y)}.\]

\subsection{Our Framework}\label{OurFramework}

To deal with the problem explained in Subsection \ref{InputswithFuzzyAttributes}, we suggest a probabilistic approach to conduct a random experiment with the probability $0.7$ of selecting $x_0$ as {\it high}. If $x_0$ is not selected, then we select nothing! This procedure could be applied to any input $x$ to feed the system with only  $high$ values. In the aforementioned example, we may use a different probabilistic approach to select $x_0$ as $high$. Thus, in general we have a probability value $\mathbb{P}(x_0\text{ is } high)$ for selecting $x_0$ as high and selecting nothing with a probability of $1-\mathbb{P}(x_0\text{ is } high)$  (see Figure~\ref{selectashigh}).
\begin{figure}
	\begin{tikzpicture}[
		node distance = 4mm and 22mm
		]
		\node (high) [draw,minimum size=24mm] {$high$};
		\node (system) [draw,minimum size=24mm, right = of high.east] {System};
		\coordinate[above left = of high.west]   (a1);
		\coordinate[below left = of high.west]      (a2);

		\coordinate[above = of high.east]   (b1);
		\coordinate[below  = of high.east]      (b2);
		%
		\foreach \i [count=\xi from 1] in {$x_0$, $x_1$ }
		\draw[-latex']  (a\xi) node[left] {\i} -- (a\xi-| high.west);
		
		\draw[-latex']  (b1)--(b1-|system.west);
	\end{tikzpicture}
	\caption{ After tossing two coins for which the probabilities of appearing Head is $\mathbb{P}(x_0\text{ is }high)$ and $\mathbb{P}(x_1\text{ is }high)$, respectively, $x_0$ is selected as $high$, while $x_1$ is not selected as high. }\label{selectashigh}
\end{figure}
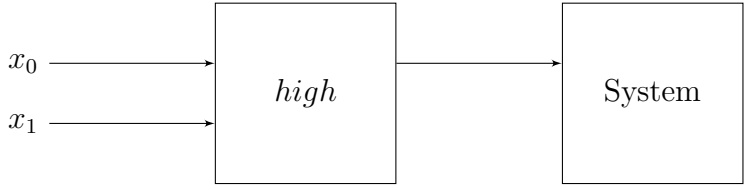
 Further, we naturally associate a random variable $\xi_{x_0,high}$ to the random experiment of selecting $x_0$ or nothing with the fuzzy attribute $A$. Indeed, as we discuss in Section \ref{Some Models of our PFL Framework}, the  conditional distribution of $\xi_{x_0,A}$ given $\xi_{y,B}$ in a so-called standard model satisfies Golden Property (explained in Subsection \ref{Golden Property}). Then, we extend our framework to include an alternative to the assignment mechanism in Causal Inference. To do so,  the model explained in Subsection \ref{FuzzyBasedAssignment Machanism} for the random mechanism could be built up by our framework  using a random variable $\xi_{T,A}$ (see Section \ref{frame} for the definition of $\xi_{T,A}$ and Section \ref{FATE} for our assignment mechanism). Indeed, in  contrast to the probability measure $q_A$ (defined in Subsection \ref{FuzzyBasedAssignment Machanism}), the probability measure induced by $\xi_{T,A}$ randomly selects elements based on Experiment~($\star$) explained in Section \ref{frame}.

 \section{Our PFL Framework in Discrete Case}\label{frame}

Let $\Omega$ be a set and $A$ be a fuzzy attribute of $\Omega$. For instance, $A$ could be an attribute like {\it large}, {\it medium}, {\it small}, {\it not large}, {\it low} or {\it high}. Each element $\omega$ of $\Omega$ has the attribute $A$ with a membership degree $\mu_A(\omega)$ in the interval $[0,1]$. We use the membership notation "$\in$" only for crisp sets, while for the fuzzy attribute $A$ of  $\Omega$, we write "$\omega$ is $A$".  In the sequel of this section, we assume that $\Omega$ is a discrete subset of $\mathbb{R}$. Also, we assume that $X:\Omega\to\mathbb{R}$ is the inclusion random variable.

A general measure theoretical framework including the continuous case is defined in Section \ref{Measure Theoretical Approach}.

For any $\omega\in \Omega$, we define $\mathbb{P}(\omega \text{ is } A)$ to be the chance that $\omega$ is randomly selected as an element with the attribute $A$. This kind of selection is subjective, and hence it is based on the knowledge or the assumptions of the subject. However, the subject could make her selection reasonable by paying attention to the facts and information about $A$ and random selection of $\omega$ as $A$ (see the subjective interpretation of Probability Theory in Section \ref{subjective}, and for more detail see \cite{gillies2012philosophical}). 
 Note that for any $\omega\in \Omega$, there is a random experiment for selecting $\omega$ or nothing as $A$ rather than selecting a random element from $\Omega$ (see Figure \ref{fig1}).
\begin{figure}[!h]
\begin{center}
\includegraphics[scale=.5]{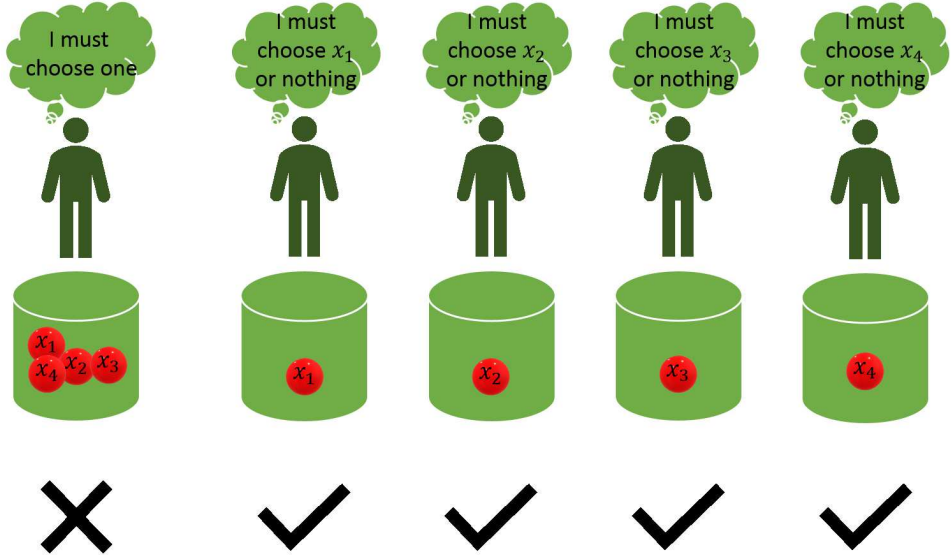}
\end{center}
\caption{\Small{This figure shows what a random experiment means when we intend to choose an element of a set $\Omega$ with the fuzzy attribute $A$. So, we have one experiment for each element of $\Omega$.}}\label{fig1}
\end{figure}
Indeed, a similar point of view but in the context of possibility theories has been considered in \cite{kovalerchuk2017relationships}.

This interpretation implies that for an event $E$, in general we have the following:
\[ \mathbb{P}(E)\neq\sum_{\omega\in \Omega}\mathbb{P}(E|\omega\text { is }A)\mathbb{P}(\omega\text { is }A).\] 
Note that $\mathbb{P}(\neg(\omega\text{ is } A))$ and $\mathbb{P}(\omega\text{ is }not\,A)$ are different. Indeed, $\mathbb{P}(\neg(\omega\text{ is } A))=1-\mathbb{P}(\omega\text{ is } A)$, while $\mathbb{P}(\omega\text{ is }not\,A)$ is the chance for $\omega$ to be chosen with the attribute "$not\;\text{A}$". Hence, $\mathbb{P}(\neg(\omega\text{ is } A))$ and $\mathbb{P}(\omega\text{ is }not\,A)$ come from different random experiments.

We also define $\mathbb{P}(\Omega\text{ is } A)$ to be the chance that a randomly chosen element of $\Omega$ has the attribute $A$. Indeed, $\mathbb{P}(\Omega \text{ is } A)$ is obtained in two steps:
\begin{enumerate}\label{experiment}
\item
Randomly selecting an element $\omega$ from $\Omega$,
\item
Doing a new random experiment for selected $\omega$ in Step (1) with the following outcomes:
\[\{\text{$\omega$ is selected as $A$}, \text{Nothing is selected as $A$}\}.\]
\end{enumerate}

In the sequel of this paper, we refer to this experiment as Experiment ($\star$). Now, we associate the inclusion random variable $X:\Omega\to\mathbb{R}$ to the first step of Experiment~($\star$). Thus, we have the following equality:
\[ \mathbb{P}(\Omega \text{ is } A|X=x)=\mathbb{P}(x \text{ is } A).\]
So, by the total law of probability, we have that:
\[\mathbb{P}(\Omega \text{ is } A)=\sum_{x\in \Omega}\mathbb{P}(\Omega \text{ is } A|X=x)\mathbb{P}(X=x)=\sum_{x\in \Omega}\mathbb{P}(x \text{ is } A)\mathbb{P}(X=x).\]
We might use the notation $\mathbb{P}(X\text{ is } A)$ rather than $\mathbb{P}(\Omega\text{ is } A)$.

In Section \ref{Relationship Between Zadeh's PFL and our Theory}, we will see that how  $\mathbb{P}(\Omega\text{ is }A)$ helps us to define a relationship between Zadeh and our PFL frameworks.

Let $\phi_A:\Omega\to\mathbb{R}$ be the map sending any $x\in \Omega$ to $\mathbb{P}(x\text{ is }A)$. We define
\[\mathrm{Supp}(\phi_A)=\{x\in \Omega:\phi_A(x)\neq 0\}=\{x\in X:\mathbb{P}(x\text{ is }A)\neq 0\}.\]
We call $A$ a \textit{proper} attribute of $\Omega$ with respect to $\phi_A$, whenever $\mathrm{Supp}(\phi_A)\subsetneqq \Omega$. In the sequel, all considered fuzzy attributes are proper.
Fix $x_A\in \Omega\backslash\mathrm{Supp}(\phi_A)$ as a base element for $A$. 

In the sequel, we define some random variables which form some basis of our framework. Then, we reformulate probability formulas coming from these random variables using their underlying random experiments in order to obtain more friendly formulas. By more friendly, we mean for instance,  $\mathbb{P}(x\text{ is }A)$, $\mathbb{P}(x\text{ is }A|Y=y)$, $\mathbb{P}(x\text{ is } A|y\text{ is }B)$ and $\mathbb{P}(x\text{ is } A| y\text{ is } B, Z=z)$.  In the random experiment for selecting $x$ with the attribute $A$ or selecting nothing, we imagine that selecting "nothing" is the same  action as selecting $x_A$. This makes sense, since $\mathbb{P}(x_A\text{ is }A)=0$, and hence $x_A$ is impossible to be selected with respect to $A$. For any $x\in \Omega\backslash\{x_A\}$, we define a random variable $\xi_{x,A}$ to reflect the aforementioned experiment with two outcomes $x$ and "nothing" as follows:
\[ \xi_{x, A}=\left\{\begin{array}{ll} x  & \text{,with the probability } \mathbb{P}(x \text{ is } A),\\ x_A & \text{,with the probability } 1-\mathbb{P}(x \text{ is } A).\end{array}\right.\]
As a convention,  let $\xi_{x_A,A}$ be the constant variable  equals 1. For any $x\in \Omega$, we interpret $\mathbb{E}(\xi_{x,A})$ as a normalized amount of $x$  relative to $A$ and $x_A$. Indeed, we have that:
\[\mathbb{E}(\xi_{x,A})=x\mathbb{P}(x \text{ is } A)+x_A(1-\mathbb{P}(x \text{ is } A))=x_A+(x-x_A)\mathbb{P}(x \text{ is } A).\] 
In other words, $\mathbb{E}(\xi_{x,A})$ is a point between $x_A$ and $x$. Hence, the greater $\mathbb{P}(x\text{ is }A)$, the more $\mathbb{E}(\xi_{x,A})$ is closer to $x$.

Now, assume that we want to have a specific element $x$ as the result of Experiment~($\star$). Then, 
 we consider a random variable $\xi_{X,A}$ with the following property:
\[\mathbb{P}(\xi_{X,A}=\alpha|X=x)=\mathbb{P}(\xi_{x,A}=\alpha)\] 
for any $x\in \Omega$.
It follows that:
\begin{align*}
\mathbb{P}(\xi_{X,A}=\alpha)&=\sum_{x\in \Omega}\mathbb{P}(\xi_{X,A}=\alpha|X=x)\mathbb{P}(X=x)=\sum_{x\in \Omega}\mathbb{P}(\xi_{x,A}=\alpha)\mathbb{P}(X=x).
\end{align*}
Thus, if $\alpha\notin \Omega$, then $\mathbb{P}(\xi_{X,A}=\alpha)=0$. Also, if $\alpha=x$ for some $x\in \Omega\backslash\{x_A\}$, then we have that $\mathbb{P}(\xi_{X,A}=\alpha)=\mathbb{P}(x\text{ is } A)\mathbb{P}(X=x)$.  We also have that
\begin{align*}
\mathbb{P}(\xi_{X,A}=x_A)&=\sum_{x\in \Omega}\mathbb{P}(\xi_{x,A}=x_A)\mathbb{P}(X=x)\\
&=\sum_{x\in \Omega}(1-\mathbb{P}(x \text{ is } A))\mathbb{P}(X=x)\\
&=1-\mathbb{P}(\Omega \text{ is } A).
\end{align*}
We note that $\xi_{X,A}=x_A$ means to select nothing in Experiment ($\star$). 
It follows that:
\[\mathbb{E}(\xi_{X,A})=x_A+\sum_{x\in \Omega}(x-x_A)\mathbb{P}(x\text{ is } A)\mathbb{P}(X=x).\]
We interpret $\mathbb{E}(\xi_{X,A})$ as the mean value of a randomly selected element of $\Omega$ to be $A$, relative to $x_A$. 

Now, let $\Omega$ and $\Omega'$ be two subsets of $\mathbb{R}$, and $A$ and $B$ be two fuzzy attributes of $\Omega$ and $\Omega'$, respectively. Also, we assume that $X:\Omega\to\mathbb{R}$ and $Y:\Omega'\to\mathbb{R}$ are the inclusion random variables. To obtain more friendly formulas, We investigate  different conditional distributions formed of $\xi_{X,A}$, $\xi_{Y,B}$, $X$ and $Y$. To do so, first note that we have the following simple lemma:
\begin{Lemma}\label{help}
Let $Z$ be a real valued random variable and $z_0\in\mathbb{R}$. Then, we have that 
\[\mathbb{E}(Z|E)=z_0+\sum_{Z=z}(z-z_0)\mathbb{P}(Z=z|E).\]
\end{Lemma}
Now, we investigate the aforementioned conditional distributions one by one as it appears in the following bullet points.
\\[3pt]

\noindent$\bullet\;\mathbb{P}(\xi_{y,B}=\beta|X=\alpha)$:\\[3pt]

Naturally, we have the following:
\[\mathbb{P}(\xi_{y,B}=\beta|X=\alpha)=\left\{\begin{array}{ll}
	\mathbb{P}(y\text{ is }B|X=\alpha)&,\beta=y\neq y_B\\
	\mathbb{P}(\neg(y\text{ is }B)|X=\alpha)&,\beta=y_B\\
	0&,\text{otherwise}
\end{array}\right..
\]
Now, it follows from Lemma \ref{help} that
\[\mathbb{E}(\xi_{y,B}|X=\alpha)=y_B+(y-y_B)\mathbb{P}(y\text{ is }B|X=\alpha).\]

\noindent$\bullet\;\mathbb{P}(X=\alpha|\xi_{y,B}=\beta)$:\\[3pt]

We  simply have a friendly formula as below:
\[\mathbb{P}(X=\alpha|\xi_{y,B}=\beta)=\left\{\begin{array}{ll}
	\mathbb{P}(X=\alpha|y\text{ is }B)&,\beta=y\neq y_B\\
		\mathbb{P}(X=\alpha|\neg(y\text{ is }B))&,\beta=y_B\\
		0&,\text{otherwise}
\end{array}\right..
\]
Note that by Bayes rule, we have that
\begin{align*}
	\mathbb{P}(X=\alpha|\neg(y\text{ is }B))&=\frac{\mathbb{P}(\neg(y\text{ is }B)|X=\alpha)\mathbb{P}(X=\alpha)}{\mathbb{P}(\neg(y\text{ is }B))}\\
	&=\frac{\left(1-\mathbb{P}(y\text{ is }B|X=\alpha)\right)\mathbb{P}(X=\alpha)}{1-\mathbb{P}(y\text{ is }B)}\\
	&=\frac{\mathbb{P}(X=\alpha)-\mathbb{P}(X=\alpha|y\text{ is }B)}{1-\mathbb{P}(y\text{ is }B)}.
\end{align*}
Thus, for $\beta=y\neq y_B$, we have that
\[\mathbb{E}(X|\xi_{y,B}=y)=\mathbb{E}(X|y\text{ is }B)=\sum_{x\in \Omega}x\mathbb{P}(X=x|y\text{ is }B).\]
Further, for $\beta=y_B$, we have that
\begin{align*}
	\mathbb{E}(X|\xi_{y,B}=y_B)=\mathbb{E}(X|\neg(y\text{ is }B))&=\sum_{x\in \Omega}x\frac{\mathbb{P}(X=x)-\mathbb{P}(X=x|y\text{ is }B)}{1-\mathbb{P}(y\text{ is }B)}\\
	&=\frac{\mathbb{E}(X)-\mathbb{E}(X|y\text{ is }B)}{1-\mathbb{P}(y\text{ is }B)}.
\end{align*}

\noindent$\bullet\;\mathbb{P}(\xi_{X,A}=\alpha|Y=\beta)$:\\

Let $\alpha\in \Omega\backslash\{x_A\}$. We use the total law of probability as follows:
\begin{align*}
\mathbb{P}(\xi_{X,A}=\alpha|Y=\beta)&=\sum_{x\in \Omega}\mathbb{P}(\xi_{X,A}=\alpha|X=x,Y=\beta)\mathbb{P}(X=x|Y=\beta)\\
&=\sum_{x\in \Omega}\mathbb{P}(\xi_{x,A}=\alpha|X=x,Y=\beta)\mathbb{P}(X=x|Y=\beta)\\
&=\mathbb{P}(\xi_{\alpha,A}=\alpha|X=\alpha,Y=\beta)\mathbb{P}(X=\alpha|Y=\beta).
\end{align*}
Now, it follows from Lemma \ref{help} that
\begin{align*}
	\mathbb{E}(\xi_{X,A}|Y=\beta)&=\sum_{x\in \Omega}x\mathbb{P}(\xi_{X,A}=x|Y=\beta)\\
	&=x_A+\sum_{x\in \Omega}(x-x_A)\mathbb{P}(x\text{ is }A|X=x,Y=\beta)\mathbb{P}(X=x|Y=\beta).
\end{align*}
Thus, if $\xi_{X,A}$ is independent of $Y$ given $X$, then we have that
\[\mathbb{P}(\xi_{X,A}=\alpha|Y=\beta)=\mathbb{P}(\alpha\text{ is }A)\mathbb{P}(X=\alpha|Y=\beta),\]
and hence we have that
\[\mathbb{E}(\xi_{X,A}|Y=y)=x_A+\sum_{x\in \Omega}(x-x_A)\mathbb{P}(x\text{ is }A)\mathbb{P}(X=x|Y=y).\]

\noindent$\bullet\; \mathbb{P}(\xi_{y, B}=\beta|\xi_{x,A}=\alpha)$:\\[3pt]

 Assume that $x\in \Omega$ and $y\in \Omega'$.  We reformulate the probability mass function of the conditional distribution $\xi_{y,B}$ given $\xi_{x,A}$ as follows:
\[ \mathbb{P}(\xi_{y,B}=\beta|\xi_{x,A}=\alpha)=\left\{\begin{array}{ll}\mathbb{P}(y \text{ is } B|x \text{ is } A)&,\beta=y\neq y_B, \alpha=x\neq x_A\\

\mathbb{P}(y \text{ is } B|\neg(x \text{ is } A))&,\beta=y\neq y_B, \alpha=x_A\\

\mathbb{P}(\neg(y \text{ is } B)|x \text{ is } A)&,\beta=y_B, \alpha=x\neq x_A\\

\mathbb{P}(\neg(y \text{ is } B)|\neg(x \text{ is } A))&,\beta=y_B, \alpha=x_A\\

0&,\text{otherwise}\end{array}\right.\]

Recall that in the above reformulation, by $\mathbb{P}(\neg(y\text{ is }‌B))$ and $\mathbb{P}(\neg(y\text{ is } B)|x\text{ is }A)$, we mean $1-\mathbb{P}(y\text{ is } B)$ and $1-\mathbb{P}(y\text{ is } B|x\text{ is }A)$, respectively. Also, it follows from Bayes rule that
\begin{align*}
\mathbb{P}(y\text{ is } B|\neg(x\text{ is }A))&=\frac{\mathbb{P}(\neg(x\text{ is } A)|y\text{ is }B)\mathbb{P}(y\text{ is }B)}{\mathbb{P}(\neg(x\text{ is }A))}\\
&=\frac{(1-\mathbb{P}(x\text{ is } A|y\text{ is }B))\mathbb{P}(y\text{ is }B)}{1-\mathbb{P}(x\text{ is }A)}\\
&= \frac{\mathbb{P}(y\text{ is }B)-\mathbb{P}(y\text{ is } B)|x\text{ is }A)\mathbb{P}(x\text{ is }A)}{1-\mathbb{P}(x\text{ is }A)}.
\end{align*}
We obtain the expected value in this case as follows:
\[\mathbb{E}(\xi_{y,B}|\xi_{x,A}=\alpha)=\left\{\begin{array}{ll}
	y_A+(y-y_A)\mathbb{P}(y\text{ is }B|x\text{ is }A)&,\alpha=x\neq x_A\\
	y_A+(y-y_A)\mathbb{P}(y\text{ is }B|\neg(x\text{ is }A))&, \alpha=x_A\\
	0&,\text{otherwise}
\end{array}\right..\]

\noindent$\bullet\;\mathbb{P}(Y=\beta|\xi_{X,A}=\alpha)$:\\[3pt]

For $\alpha\in \Omega\backslash\{x_A\}$, we obviously have that
\[\mathbb{P}(Y=\beta|\xi_{X,A}=\alpha)=\mathbb{P}(Y=\beta|(\alpha\text{ is }A)\&(X=\alpha)).\]
Now, assume that $\alpha=x_A$. It follows from the Bayes rule that
\[\mathbb{P}(Y=\beta|\xi_{X,A}=x_A)=\frac{\mathbb{P}(\xi_{X,A}=x_A|Y=\beta)\mathbb{P}(Y=\beta)}{\mathbb{P}(\xi_{X,A}=x_A)}.\]
Since we have  $\mathbb{P}(\xi_{X,A}=x_A|Y=\beta)$, then \\[-6pt]
\begin{center}
\scalebox{.85}{$\mathbb{P}(Y=\beta|\xi_{X,A}=x_A)=
	\dfrac{\left(1-\sum_{x\in \Omega}\mathbb{P}(x\text{ is }A|X=x, Y=\beta)\mathbb{P}(X=x|Y=\beta)\right)\mathbb{P}(Y=\beta)}{1-\mathbb{P}(\Omega\text{ is }A)}$.}
\end{center}

Now, by using the Bayes rule, we have that\\[-7pt]
\begin{center}
\scalebox{.85}{$\mathbb{P}(Y=\beta|\xi_{X,A}=x_A)=
	\dfrac{\mathbb{P}(Y=\beta)-\sum_{x\in \Omega}\mathbb{P}(x\text{ is }A|X=x, Y=\beta)\mathbb{P}(Y=\beta|X=x)\mathbb{P}(X=x)}{1-\mathbb{P}(\Omega\text{ is }A)}
$.}
\end{center}
If we once again use the Bayes rule, then \\[-7pt]
\begin{center}
\scalebox{.85}{$\mathbb{P}(Y=\beta|\xi_{X,A}=x_A)=
	\dfrac{\mathbb{P}(Y=\beta)-\sum_{x\in \Omega}\mathbb{P}(Y=\beta|(x\text{ is }A)\&( X=x))\mathbb{P}(x\text{ is }A)\mathbb{P}(X=x)}{1-\mathbb{P}(\Omega\text{ is }A)}$.}
\end{center}

Now, the expected value in this case for $\alpha\neq x_A$ is
\begin{align*}
	\mathbb{E}(Y|\xi_{X,A}=\alpha)&=\sum_{\beta\in \Omega'}\beta\mathbb{P}(Y=\beta|\xi_{X,A}=\alpha)\\
	&=\sum_{\beta\in \Omega'}\beta\mathbb{P}(Y=\beta|(\alpha\text{ is }A)\& (X=\alpha)).
\end{align*}
Also, for $\alpha=x_A$, we have that
\begin{align*}
	\mathbb{E}(Y|\xi_{X,A}=x_A)&=\sum_{\beta\in \Omega'}\beta\mathbb{P}(Y=\beta|\xi_{X,A}=x_A)\\
	&=\frac{\sum_{\beta\in \Omega'}\beta\mathbb{P}(Y=\beta)}{1-\mathbb{P}(\Omega\text{ is }A)}\\
	&-\frac{\sum_{\beta\in \Omega'}\sum_{x\in \Omega}\beta\mathbb{P}(Y=\beta|(x\text{ is }A)\&( X=x))\mathbb{P}(x\text{ is }A)\mathbb{P}(X=x)}{1-\mathbb{P}(\Omega\text{ is }A)}.
\end{align*}
Hence,\\[-7pt]
\begin{center}
\scalebox{.85}{$	\mathbb{E}(Y|\xi_{X,A}=x_A)=\dfrac{\mathbb{E}(Y)-\sum_{\beta\in \Omega'}\sum_{x\in \Omega}\beta\mathbb{P}(Y=\beta|(x\text{ is }A)\&( X=x))\mathbb{P}(x\text{ is }A)\mathbb{P}(X=x)}{1-\mathbb{P}(\Omega\text{ is }A)}$.}
\end{center}
\noindent$\bullet\;\mathbb{P}(X=x|\xi_{X,A}=\alpha)$:\\[3pt]

This could be obtained from the previous case, although it has a special character, and hence we independently investigate it.
It follows from the Bayes rule that
\[\mathbb{P}(X=x|\xi_{X,A}=\alpha)=\frac{\mathbb{P}(\xi_{X,A}=\alpha|X=x)\mathbb{P}(X=x)}{\mathbb{P}(\xi_{X,A}=\alpha)}=\frac{\mathbb{P}(\xi_{x,A}=\alpha)\mathbb{P}(X=x)}{\mathbb{P}(\xi_{X,A}=\alpha)}.\]
Thus, if $\alpha\neq x_A$, then
\begin{equation}\label{X|xi1}
	\mathbb{P}(X=x|\xi_{X,A}=\alpha)=\left\{\begin{array}{ll}
		1 &, \alpha=x\\
		0 &, \alpha\neq x
	\end{array}
	\right.,
\end{equation}
and if $\alpha=x_A$, then
\begin{equation*}
	\mathbb{P}(X=x|\xi_{X,A}=\alpha)=\left\{\begin{array}{ll}
		\displaystyle\frac{\mathbb{P}(X=x_A)}{1-\mathbb{P}(\Omega\text{ is } A)} &, x=x_A\\[12pt]
		\displaystyle	\frac{\mathbb{P}(\neg(x\text{ is } A))\mathbb{P}(X=x)}{1-\mathbb{P}(\Omega\text{ is } A)} &, x\neq x_A
	\end{array}
	\right..
\end{equation*}
We note that $\mathbb{P}(\neg(x_A\text{ is } A))=1-\mathbb{P}(x_A\text{ is } A)=1$. Hence, we can remove the first expression ($x=x_A$), which results only one expression as follows:
\begin{equation}\label{X|xi2}
	\mathbb{P}(X=x|\xi_{X,A}=x_A)=
	\frac{\mathbb{P}(\neg(x\text{ is }A))\mathbb{P}(X=x)}{1-\mathbb{P}(\Omega\text{ is } A)}.
\end{equation}
Therefore, for $\alpha\in \Omega\backslash\{x_A\}$, we have that
\[\mathbb{E}(X|\xi_{X,A}=\alpha)=\sum_{x\in \Omega}x\mathbb{P}(X=x|\xi_{X,A}=\alpha)=\alpha,\]
and for $\alpha=x_A$, we have that
\begin{align*}
	\mathbb{E}(X|\xi_{X,A}=x_A)&=\sum_{x\in \Omega}x\mathbb{P}(X=x|\xi_{X,A}=x_A)=\sum_{x\in \Omega}x	\frac{\mathbb{P}(\neg(x\text{ is }A))\mathbb{P}(X=x)}{1-\mathbb{P}(\Omega\text{ is } A)}\\
	&=\frac{1}{1-\mathbb{P}(\Omega\text{ is }A)}\sum_{x\in \Omega}x(1-\mathbb{P}(x\text{ is }A))\mathbb{P}(X=x)\\
	&=\frac{1}{1-\mathbb{P}(\Omega\text{ is }A)}\left(\sum_{x\in \Omega}x\mathbb{P}(X=x)-\sum_{x\in \Omega}x\mathbb{P}(x\text{ is }A)\mathbb{P}(X=x)\right)\\
	&=\frac{\mathbb{E}(X)-\mathbb{E}(\xi_{X,A})+x_A(1-\mathbb{P}(\Omega\text{ is }A))}{1-\mathbb{P}(\Omega\text{ is }A)}=x_A+\frac{\mathbb{E}(X)-\mathbb{E}(\xi_{X,A})}{1-\mathbb{P}(\Omega\text{ is }A)}.
\end{align*}

\noindent$\bullet\;\mathbb{P}(\xi_{Y,B}=\beta|\xi_{x,A}=\alpha)$:\\[3pt]

For $\beta\in \Omega'\backslash\{y_B\}$, by applying the total law of probability over the values of $Y$, we have that
\begin{align*}
	\mathbb{P}(\xi_{Y,B}=\beta|\xi_{x,A}=\alpha)&=\sum_{y\in \Omega'}\mathbb{P}(\xi_{Y,B}=\beta|\xi_{x,A}=\alpha, Y=y)\mathbb{P}(Y=y|\xi_{x,A}=\alpha)\\
	&=\sum_{y\in \Omega'}\mathbb{P}(\xi_{y,B}=\beta|\xi_{x,A}=\alpha, Y=y)\mathbb{P}(Y=y|\xi_{x,A}=\alpha)\\
	&=\mathbb{P}(\xi_{\beta,B}=\beta|\xi_{x,A}=\alpha, Y=\beta)\mathbb{P}(Y=\beta|\xi_{x,A}=\alpha).
\end{align*}
It follows that $\mathbb{P}(\xi_{Y,B}=\beta|\xi_{x,A}=\alpha)$ equals
\[	\left\{\begin{array}{ll}
	\mathbb{P}(\beta\text{ is }B|(\alpha\text{ is }A)\&(Y=\beta))\mathbb{P}(Y=\beta|\alpha\text{ is }A)&, x=\alpha\neq x_A\\
	\mathbb{P}(\beta\text{ is }B|(\neg(\alpha\text{ is }A))\&(Y=\beta))\mathbb{P}(Y=\beta|\neg(\alpha\text{ is }A))&, \alpha= x_A\\
	0&,\text{otherwise}
	\end{array}\right..
	\]
Hence, $\mathbb{E}(\xi_{Y,B}|\xi_{x,A}=\alpha)$ equals
\[\Small\left\{\begin{array}{ll}
	y_B+\sum_{\beta\in \Omega'}(\beta-y_B)\beta\mathbb{P}(\beta\text{ is }B|(\alpha\text{ is }A)\&(Y=\beta))\mathbb{P}(Y=\beta|\alpha\text{ is }A)&, x=\alpha\neq x_A\\
	y_B+\sum_{\beta\in \Omega'}(\beta-y_B)\beta\mathbb{P}(\beta\text{ is }B|(\neg(\alpha\text{ is }A))\&(Y=\beta))\mathbb{P}(Y=\beta|\neg(\alpha\text{ is }A))&, \alpha= x_A\\
	0&,\text{otherwise}
\end{array}\right..\]

Thus, if selecting $\beta$ as $B$ and $Y=\beta$ are  independent given $\alpha$ is $A$, and  also selecting $Y=\beta$ and selecting $\alpha$ as $A$ are independent, then 
$\mathbb{P}(\xi_{Y,B}=\beta|\xi_{x,A}=\alpha)$ equals
\[	\left\{\begin{array}{ll}
	\mathbb{P}(\beta\text{ is }B|\alpha\text{ is }A)\mathbb{P}(Y=\beta)&, x=\alpha\neq x_A\\
	\mathbb{P}(\beta\text{ is }B|\neg(\alpha\text{ is }A))\mathbb{P}(Y=\beta)&, \alpha= x_A\\
	0&,\text{otherwise}
\end{array}\right..
\]
Hence, $\mathbb{E}(\xi_{Y,B}|\xi_{x,A}=\alpha)$ equals
\[\left\{\begin{array}{ll}
	y_B+\sum_{\beta\in \Omega'}(\beta-y_B)\beta\mathbb{P}(\beta\text{ is }B|\alpha\text{ is }A)\mathbb{P}(Y=\beta)&, x=\alpha\neq x_A\\
	y_B+\sum_{\beta\in \Omega'}(\beta-y_B)\beta\mathbb{P}(\beta\text{ is }B|\neg(\alpha\text{ is }A))\mathbb{P}(Y=\beta)&, \alpha= x_A\\
	0&,\text{otherwise}
\end{array}\right..\]

\noindent$\bullet\;\mathbb{P}(\xi_{y,B}=\beta|\xi_{X,A}=\alpha)$:\\[3pt]

We have that:
\begin{align*}
	\mathbb{P}(\xi_{y,B}=\beta|\xi_{X,A}=\alpha)&=\sum_{x\in \Omega}\mathbb{P}(\xi_{y,B}=\beta|\xi_{X,A}=\alpha, X=x)\mathbb{P}(X=x|\xi_{X,A}=\alpha)\\
	&=\sum_{x\in \Omega}\mathbb{P}(\xi_{y,B}=\beta|\xi_{x,A}=\alpha, X=x)\mathbb{P}(X=x|\xi_{X,A}=\alpha).
\end{align*}
This implies that $\mathbb{P}(\xi_{y,B}=\beta|\xi_{X,A}=\alpha)$ equals
\[\left\{\begin{array}{ll}\mathbb{P}(\xi_{y,B}=\beta|\xi_{\alpha,A}=\alpha, X=\alpha)&, \alpha\neq x_A\\[4pt]
	
 \displaystyle\frac{\sum_{x\in \Omega}\mathbb{P}(\xi_{y,B}=\beta|(\neg(x\text { is }A))\& (X=x))\mathbb{P}(\neg(x\text{ is }A))\mathbb{P}(X=x)}{1-\mathbb{P}(\Omega\text{ is }A)}&, \alpha= x_A\end{array}\right.	.\]

 Hence, for $y\neq y_B$, $\mathbb{P}(\xi_{y,B}=y|\xi_{X,A}=\alpha)$ equals
\[\left\{\begin{array}{ll}
\mathbb{P}(y \text{ is } B| (\alpha \text{ is } A)\& (X=\alpha))&, \alpha\neq x_A\\[3pt]

\frac{\sum_{x\in \Omega}f_y(x)}{1-\mathbb{P}(\Omega\text{ is }A)}&, \alpha=x_A\\

0&,\text{otherwise}\end{array}\right.,\]
where
\[f_y(x)=
	\mathbb{P}(y\text{ is }B|(\neg(x\text{ is } A))\& (X=x))\mathbb{P}(\neg(x\text{ is }A))\mathbb{P}(X=x).\]
Thus, for $\alpha\neq x_A$, we have that
\[\mathbb{E}(\xi_{y,B}|\xi_{X,A}=\alpha)=y_B+(y-y_B)\mathbb{P}(y \text{ is } B| (\alpha \text{ is } A)\& (X=\alpha)).\]
Also, for $\alpha=x_A$, we have that 
\begin{align*}
\mathbb{E}(\xi_{y,B}|\xi_{X,A}=\alpha)=y_A+(y-y_A)\frac{\sum_{x\in \Omega}f_y(x)}{1-\mathbb{P}(\Omega\text{ is }A)}.
\end{align*}

\noindent$\bullet\;\mathbb{P}(\xi_{Y,B}=\beta|\xi_{X,A}=\alpha)$:\\[3pt]

By the total law of probability, we have that
\begin{align*}
\mathbb{P}(\xi_{Y,B}=\beta|\xi_{X,A}=\alpha)&=\sum_{y\in \Omega'}\mathbb{P}(\xi_{Y,B}=\beta|Y=y,\xi_{X,A}=\alpha)\mathbb{P}(Y=y)\\
&=\sum_{y\in \Omega'}\mathbb{P}(\xi_{y,B}=\beta|Y=y,\xi_{X,A}=\alpha)\mathbb{P}(Y=y).
\end{align*}
It follows for $\beta\neq y_B$ that 
\[\mathbb{P}(\xi_{Y,B}=\beta|\xi_{X,A}=\alpha)=\mathbb{P}(\xi_{\beta,B}=\beta|Y=\beta,\xi_{X,A}=\alpha)\mathbb{P}(Y=\beta).\]
It follows from the previous part that for $\alpha\neq x_A$, we have that
\begin{align*}
	\mathbb{P}(\xi_{\beta,B}=\beta|Y=\beta,\xi_{X,A}=\alpha)&=\mathbb{P}(\beta\text{ is }B|(\alpha\text{ is }A)\&(Y=\beta)\&(X=\alpha))
\end{align*}
and for $\alpha=x_A$, we have that
{\Small \begin{align*}
	&\mathbb{P}(\xi_{\beta,B}=\beta|Y=\beta,\xi_{X,A}=x_A)=\frac{\sum_{x\in \Omega}f_{\beta}(x)}{1-\mathbb{P}(\Omega\text{ is }A|Y=\beta)},\\
		&\scriptsize  f_{\beta}(x)=
			\mathbb{P}(\beta\text{ is }B|(\neg(x\text{ is } A))\& (X=x)\&(Y=\beta))\mathbb{P}(\neg(x\text{ is }A)|X=x,Y=\beta)\mathbb{P}(X=x|Y=\beta).
\end{align*}}
Therefore, by applying the Bayes rule, $	\mathbb{P}(\xi_{Y,B}=\beta|\xi_{X,A}=\alpha)$ equals
\begin{align*}
	&\small \left\{\begin{array}{ll}
		\mathbb{P}(\beta\text{ is }B|(\alpha\text{ is } A)\& (X=\alpha)\&(Y=\beta))\mathbb{P}(Y=\beta)&, \alpha\neq x_A\\[7pt]
		\frac{\sum_{x\in \Omega}g_{\beta}(x)}{1-\mathbb{P}(\Omega\text{ is }A|Y=\beta)}&,\alpha=x_A
	\end{array}\right.,
\end{align*}
where, we have that
{\tiny\[ g_{\beta}(x)=
	\mathbb{P}(\beta\text{ is }B|(\neg(x\text{ is } A))\& (X=x)\&(Y=\beta))\mathbb{P}(\neg(x\text{ is }A)|X=x,Y=\beta)\mathbb{P}(Y=\beta|X=x)\mathbb{P}(X=x).\]}
 Thus, for $\alpha\neq x_A$, we have that:
\begin{align*}
	\mathbb{E}(\xi_{Y,B}|\xi_{X,A}=\alpha)&=\sum_{y\in \Omega'}y \mathbb{P}(\xi_{Y,B}=y|\xi_{X,A}=\alpha)\\
	&=y_B+\sum_{y\in \Omega'}(y-y_B)\mathbb{P}(\xi_{Y,B}=\beta|\xi_{X,A}=\alpha)\\
	&=y_B+\sum_{y\in \Omega'}(y-y_B)	\mathbb{P}(\beta\text{ is }B|(\alpha\text{ is } A)\& (X=\alpha)\&(Y=\beta))\mathbb{P}(Y=\beta).
\end{align*}
Also, for $\alpha=x_A$, we have that
\begin{align*}
	\mathbb{E}(\xi_{Y,B}|\xi_{X,A}=x_A)&=y_B+\sum_{y\in \Omega'}(y-y_B)\mathbb{P}(\xi_{Y,B}=\beta|\xi_{X,A}=x_A)\\
	&=y_B+\sum_{y\in \Omega'}\left(\frac{(y-y_B)\sum_{x\in \Omega}g_{y}(x)}{1-\mathbb{P}(\Omega\text{ is }A|Y=y)}\right).
\end{align*}

\section{Some Models of our PFL Framework}\label{Some Models of our PFL Framework}
In this section, we introduce several different models of our framework. Let $\Omega$ and $\Omega'$ be two arbitrary subsets of $\mathbb{R}$, and $X$ and $Y$ be their corresponding inclusion random variables.  Also, let $A$ and $B$ be two proper fuzzy attributes of $X$ and $Y$, respectively. Let's fix the base elements $x_A$ and $y_B$ with respect to $A$ and $B$, respectively. Recall that $\mathbb{P}(x_A\text{ is }A)=\mathbb{P}(y_B\text{ is }B)=0$. We can construct the main part of a model of our PFL framework by defining the following subjective probabilities:
\[ \mathbb{P}(x \text{ is } A),\quad \mathbb{P}(y \text{ is } B| x \text{ is } A),\quad\forall\;x\in \Omega,\,y\in \Omega'.\]
However, probabilities such as $\mathbb{P}(X=x|y\text{ is }B)$ remain undefined. Because, these probabilities are not subjective, and they depend on their underlying random experiments. 

Note that $\mathbb{P}(y \text{ is } B| x \text{ is } A)$ could come from any distribution. In our case, to calculate $\mathbb{P}(y \text{ is } B| x \text{ is } A)$, we focus on the fuzzy point of view and introduce \textit{standard and non-standard models} (see below). Now,
we define standard models as follows. Let $T$ be a $t$-norm (for several examples of $t$-norms, see Section \ref{tnorms}). 
Then, for any $x_0\in \Omega$ and $y_0\in \Omega'$, we define:
\begin{align*}
	\mathbb{P}_T(y_0\text{ is } B|x_0\text{ is } A)&=\frac{T\left(\mathbb{P}(y_0\text{ is }B), \mathbb{P}(x_0\text{ is }A)\right)}{\mathbb{P}(x_0 \text{ is }A)}.
\end{align*}
We note that 
\begin{align*}
	\mathbb{P}_T(\neg(y_0\text{ is } B)|x_0\text{ is } A)&=1-\mathbb{P}_T(y_0\text{ is } B|x_0\text{ is } A),
\end{align*}
and $\mathbb{P}_T(y_0\text{ is } B|\neg(x_0\text{ is } A))$ and $\mathbb{P}_T(\neg(y_0\text{ is } B)|\neg(x_0\text{ is } A))$ are determined by the Bayes rule and the aforementioned two probability formulas. 

Note that a $t$-norm $T$ is commutative and associative, and hence the order of applying the $t$-norm on values is not important. Hence, $T(\alpha_1,\ldots,\alpha_l)$ has only one meaning.  In general, we define:
\[\mathbb{P}_T\left(\bigwedge_{i=1}^n(y_i\text{ is }B_i)\left|\bigwedge_{j=1}^m(x_i\text{ is }A_i)\right)\right.=\frac{T\left(\mathbb{P}(z_1\text{ is }C_1),\ldots, \mathbb{P}(z_{m+n}\text{ is } C_{m+n})\right)}{T(\mathbb{P}(z_{n+1}\text{ is }C_{n+1}),\ldots, \mathbb{P}(z_{m+n}\text{ is } C_{m+n}))},\]
where we have that $z_i=y_i$ and $C_i=B_i$ for $1\le i\le n$, and  $z_j=x_j$ and $C_j=A_j$ for $n+1\le j\le n+m$. In the above expression, by the symbol "$\wedge$" we mean the logical "and".
When there is no ambiguity, we write $\mathbb{P}$ instead of $\mathbb{P}_T$. Therefore, to obtain the subjective part of standard models, we only need to define $\mathbb{P}(x\text{ is }A)$ and fix a $t$-norm. It is worth mentioning that the conditional distribution $\xi_{y,B}$ given $\xi_{x,A}=\alpha$ satisfies Golden Property  with respect to $\xi_{y,B}=y_B$ (see Section \ref{Golden Property} for the definition of this property). We note that the following definition does not determine a joint distribution:
\[\mathbb{P}(\xi_{x,A}=\alpha, \xi_{y,B}=\beta)=T(\mathbb{P}(\alpha \text{ is }A), \mathbb{P}(\beta \text{ is }B)).\]
Indeed, by the conditional probability formulas we defined before, we have that 
\begin{align*}
	\mathbb{P}(\xi_{x,A}=\alpha, \xi_{y,B}=\beta)&=\mathbb{P}(\xi_{x,A}=\alpha| \xi_{y,B}=\beta)\mathbb{P}(\xi_{y,B}=\beta)\\
	&=\mathbb{P}(\xi_{y,B}=\beta| \xi_{x,A}=\alpha)\mathbb{P}(\xi_{x,A}=\alpha).
\end{align*}
It follows that
\begin{align*}
	\mathbb{P}(\xi_{x,A}=x, \xi_{y,B}=y)&=T(\mathbb{P}(\alpha \text{ is }A), \mathbb{P}(\beta \text{ is }B)),\\
	\mathbb{P}(\xi_{x,A}=x_A, \xi_{y,B}=y)&=\mathbb{P}(\xi_{y,B}=y)- \mathbb{P}(\xi_{x,A}=x, \xi_{y,B}=y),\\
	\mathbb{P}(\xi_{x,A}=x, \xi_{y,B}=y_B)&=\mathbb{P}(\xi_{x,A}=x)- \mathbb{P}(\xi_{x,A}=x, \xi_{y,B}=y),\\
	\mathbb{P}(\xi_{x,A}=x_A, \xi_{y,B}=y_B)&=\mathbb{P}(\xi_{y,B}=y_B)- \mathbb{P}(\xi_{x,A}=x, \xi_{y,B}=y_B).
\end{align*}
We provide two generalizations of standard models in Subsection \ref{generlizedmodel} and Subsection~\ref{randomgeneralizedmodel}.
\subsection{Classic Model}
Let $\widetilde{A}$ be the fuzzy subset of $\Omega$ associated to $A$. We define the {\it cardinality} of $\widetilde{A}$ as follows:
\[ \lVert \widetilde{A}\rVert=\sum_{x\in \Omega}\mu_A(x).\]
In the sequel, we assume that any considered fuzzy attribute has a non-empty support (i.e., there exists $x\in \Omega$ for which $\mu_A(x)>0$ for any $\Omega\subseteq\mathbb{R}$ and a considered fuzzy attribute $A$ of $\Omega$). Now, we define:
\[ \mathbb{P}(x_0 \text{ is } A)=\frac{\mu_A(x_0)}{\lVert \widetilde{A}\rVert}=\frac{\mu_A(x_0)}{\sum_{x\in \Omega}\mu_A(x)}\]
for any $x_0\in \Omega$. Therefore, we have that:

\begin{align*}
\mathbb{P}(\Omega\text{ is } A)&=\frac{1}{\lVert \widetilde{A}\rVert}\sum_{x\in \Omega}\mu_A(x)\mathbb{P}(X=x),\\
\mathbb{E}(\xi_{X,A})&=x_A+\frac{1}{\lVert \widetilde{A}\rVert}\sum_{x\in \Omega}(x-x_A)\mu_A(x)\mathbb{P}(X=x).
\end{align*}
In this model, we see that  the following equality does not necessarily hold:
\[\mathbb{P}(x_0 \text{ is } A)+\mathbb{P}(x_0 \text{ is } not\,A)=1.\]
To deal with the above equality, we provide the following proposition.
\begin{Proposition}
In the classic model, for any $x\in \Omega$, we have that
\[\mathbb{P}(x \text{ is } A)+\mathbb{P}(x \text{ is } not\,A)=1\]
if and only if $|\Omega|=2$ (i.e., $\Omega$ has two elements) and $\widetilde{A}$ is a crisp singleton.
\end{Proposition}
\begin{proof}
Let $x\in \Omega$. then, we have that
\[\mathbb{P}(x \text{ is } A)+\mathbb{P}(x \text{ is } not\,A)=\frac{\mu_A(x)}{\lVert \widetilde{A}\rVert}+\frac{\mu_{not\,A}(x)}{\lVert \widetilde{not\,A}\rVert},\]
while $\mu_{not\,A}(x)=1-\mu_A(x)$ and
\[\lVert \widetilde{not\,A}\rVert=\sum_{x'\in \Omega}\mu_{not\,A}(x')=\sum_{x'\in \Omega}(1-\mu_{A}(x'))=|\Omega|-\lVert \widetilde{A}\rVert.\]
One could easily see that if   $ \widetilde{A}$ is a crisp singleton subset of $\Omega$, then the desired equality in the proposition is satisfied. Conversely, for any $x\in \Omega$, we have that
\[ \frac{\mu_A(x)}{\lVert \widetilde{A}\rVert}+\frac{1-\mu_A(x)}{|\Omega|-\lVert \widetilde{A}\rVert}=1,\]
which implies that
\begin{equation}\label{classic}
    \mu_A(x)(|\Omega|-2\lVert \widetilde{A}\rVert))=\lVert \widetilde{A}\rVert(|\Omega|-\lVert \widetilde{A}\rVert-1).
\end{equation}
Now, we have that
\[\sum_{x\in \Omega}\mu_A(x)(|\Omega|-2\lVert \widetilde{A}\rVert)=\sum_{x\in \Omega}\lVert \widetilde{A}\rVert(|\Omega|-\lVert \widetilde{A}\rVert-1),\]
and hence, 
\[\lVert \widetilde{A}\rVert(|\Omega|-2\lVert \widetilde{A}\rVert)=|\Omega|\lVert \widetilde{A}\rVert(|\Omega|-\lVert \widetilde{A}\rVert-1).\]
Thus, $|\Omega|^2-(2+\lVert \widetilde{A}\rVert)|\Omega|+2\lVert \widetilde{A}\rVert=0$, and hence $|\Omega|=\frac{2+\lVert \widetilde{A}\rVert\pm\left|\lVert \widetilde{A}\rVert-2\right|}{2}$. Therefore, $|\Omega|=\lVert \widetilde{A}\rVert$ or $|\Omega|=2$. In the first case,  for any $x\in \Omega$, we have $\mu_A(x)=1$, which implies that $\mathbb{P}(x_A\text{ is }A)=\frac{1}{|\Omega|}$, a contradiction with the fact that $A$ is proper (as we assumed in the previous section). If $|\Omega|=2$, then it follows from Equation (\ref{classic}) that 
\[2\mu_A(x_i)(1-\lVert \widetilde{A}\rVert))=\lVert \widetilde{A}\rVert(1-\lVert \widetilde{A}\rVert),\]
for $i=1,2$, where $\Omega=\{x_1,x_2\}$.
This implies that $\lVert\widetilde{A}\rVert=1$ or $\lVert\widetilde{A}\rVert=2\mu_A(x_i)$ for $i=1,2$. Consequently, $\lVert\widetilde{A}\rVert=1$ or $\mu_A(x_1)=\mu_A(x_2)$. First, consider the case with $\lVert\widetilde{A}\rVert=1$. Then, it follows from $\mu_A(x_A)=0$ that $\widetilde{A}$ is a crisp singleton. To investigate the second case,  we note that $\mu_A(x_1)=\mu_A(x_2)$ implies that $\lVert\widetilde{A}\rVert=0$, since $x_A=x_1$ or $x_A=x_2$ with $\mu_A(x_A)=0$. However, $\lVert\widetilde{A}\rVert=0$ is impossible, since we have assumed in this model that $\lVert\widetilde{A}\rVert\neq 0$. Hence, the only acceptable case is that $\widetilde{A}$ is a crisp singleton.
\end{proof}

Now, to have a standard model, we provide the formula for the conditional probability $\mathbb{P}(y\text{ is }B|x\text{ is }A)$. Let $T$ be a $t$-norm.
For any $x_0\in \Omega$ and $y_0\in \Omega'$, we have that:
\begin{align*}
 \mathbb{P}(y_0\text{ is } B|x_0\text{ is } A)&=\frac{T\left(\mathbb{P}(y_0\text{ is }B), \mathbb{P}(x_0\text{ is }A)\right)}{\mathbb{P}(x_0 \text{ is }A)}=\frac{T\left(\frac{\mu_B(y_0)}{\lVert \widetilde{B}\rVert},\frac{\mu_A(x_0)}{\lVert \widetilde{A}\rVert}\right)}{\frac{\mu_A(x_0)}{\lVert \widetilde{A}\rVert}}.
 \end{align*}
A similar model is the  Zadeh's PFL introduced in  \cite{zadeh1996fuzzy} and \cite{zadeh1968probability}. However, on the contrary to  Zadeh who associated probabilities to fuzzy events of a set $\Omega$,  we associate probabilities to crisp events of $\Omega$ with respect to a fixed fuzzy attribute of $\Omega$. In Section \ref{Relationship Between Zadeh's PFL and our Theory}, we clarify the relationship between our PFL theory and Zadeh's PFL.

\subsection{Classic Probability Based Model}

In this model, the chance for an element $x_0$ to be selected as $A$ depends on both $\mu_A(x_0)$ and $\mathbb{P}(X=x_0)$. Indeed, we define:
\[\mathbb{P}(x_0\text{ is } A)=\mu_A(x_0)\mathbb{P}(X=x_0).\]
It follows that:
\[\mathbb{P}(\Omega\text{ is } A)=\sum_{x\in \Omega}\mu_A(x)\mathbb{P}(X=x)^2.\]
Let $T$ be a $t$-norm. Then, to obtain a standard model,  we have that
\[ \mathbb{P}(y_0\text{ is } B|x_0\text{ is } A)=\frac{T(\mu_B(y_0)\mathbb{P}(Y=y_0), \mu_A(x_0)\mathbb{P}(X=x_0))}{\mu_A(x_0)\mathbb{P}(X=x_0)}.\]
However, a non-standard conditional probability could be defined as follows:
\[\mathbb{P}(y_0\text{ is } B|x_0\text{ is } A)=\frac{T(\mu_B(y_0), \mu_A(x_0))}{\mu_A(x_0)}\mathbb{P}(Y=y_0|X=x_0).\]
One could see that for the above non-standard conditional probability, the Bayes rule is satisfied.
\subsection{Simple Fuzzy Model}
In this model, the chance for an element $x_0$ to be selected as $A$ is only depended on $\mu_A(x_0)$. Indeed, we define:
\[\mathbb{P}(x_0\text{ is } A)=\mu_A(x_0).\]
It follows that:
\[\mathbb{P}(\Omega\text{ is } A)=\sum_{x\in \Omega}\mu_A(x)\mathbb{P}(X=x).\]
Let $T$ be a $t$-norm. Then we define
\[ \mathbb{P}(y_0\text{ is }B|x_0\text{ is } A)=\frac{T(\mu_B(y_0), \mu_A(x_0))}{\mu_A(x_0)}.\]
\begin{Example}\label{exreproductive}
	The probabilities of reaching reproductive capability of a certain species in different days after the birth is given in Table \ref{speciestable}. The bar chart of these probability values is provided in Figure \ref{barchart}. For simplicity, the small probability values for more than 199 days are ignored and assumed to be zero. We assume that $\Omega$ is the set of all integers from 0 to 199 reflecting the days order. Let us assume that $X$ is the inclusion random variable associated to $\Omega$. Consider the fuzzy attributes "early", "normal" and "late" for $\Omega$ according to Figure \ref{fuzzyreproductive}. Denote "early", "normal" and "late" by $A$, $B$ and $C$, respectively. So, it follows from Figure \ref{fuzzyreproductive} that
	\begin{align*}
		\mu_A(x)&=\left\{\begin{array}{ll}
			1 & ,0\le x\le 40\\
			\dfrac{100-x}{60}& , 40\le x\le 100\\
			 0&,\text{otherwise}\end{array} \right.,\\
		\mu_B(x)&=\left\{\begin{array}{ll}
			\dfrac{x}{80}& , 0\le x\le 80
			\\
			1 &, 80\le x\le 120\\
			\dfrac{200-x}{80} &, 120\le x\le 200\\ 
			0&,\text{otherwise}\end{array} \right.,\\
		\mu_C(x)&=\left\{\begin{array}{ll}
			\dfrac{x-100}{60}&, 100\le x\le 160\\
			1 & ,160\le x\le 200\\
			 0&,\text{otherwise}\end{array} \right..\\
	\end{align*}
	Thus, for instance, we have that:
	\[\mathbb{P}(57\text{ is } A)=\mu_A(57)=\frac{43}{60},\qquad \mathbb{P}(57\text{ is } B)=\mu_B(57)=\frac{57}{80}.\]
	Now, we have that
	\begin{align*}
		\mathbb{P}(\Omega\text{ is } A)&=\sum_{x\in \Omega}\mu_A(x)\mathbb{P}(X=x)\approx 0.2057917.\\
		\mathbb{P}(\Omega\text{ is } B)&=\sum_{x\in X}\mu_B(x)\mathbb{P}(X=x)
		=0.878925\\
		\mathbb{P}(\Omega\text{ is } C)&=\sum_{x\in X}\mu_C(x)\mathbb{P}(X=x)\approx 0.1967083
	\end{align*}
	Now, we consider the  base points $x_A=x_C=100$ and $x_B=200$. It follows that
	\begin{align*}
		\mathbb{E}(\xi_{X,A})&=x_A+\sum_{x\in \Omega}(x-x_A)\mu_A(x)\mathbb{P}(X=x)\approx 93.0572083,\\
		\mathbb{E}(\xi_{X,B})&=x_B+\sum_{x\in \Omega}(x-x_B)\mu_B(x)\mathbb{P}(X=x)= 110.649525, \\
		\mathbb{E}(\xi_{X,C})&=x_C+\sum_{x\in X}(x-x_C)\mu_C(x)\mathbb{P}(X=x)\approx 108.1719583.
	\end{align*}
It follows that occurring a day in $\Omega$  (other than base elements) as $normal$ in Experiment~($\star$) is more probable than occurring a day as $early$ or $high$. Because, $\mathbb{E}(\xi_{X,B})$ is far enough from $x_B=200$, while $\mathbb{E}(\xi_{X,A})$ and $\mathbb{E}(\xi_{X,C})$ are close to $x_A=x_C=100$.

	By considering the minimum $t$-norm, we have that:
	\begin{align*}
		\mathbb{P}(57\text{ is }B| 57\text{ is }A)&=\frac{\min\{\mu_B(57),\mu_A(57)\}}{\mu_A(57)}=\frac{\min\{\frac{57}{80},\frac{43}{60}\}}{\frac{43}{60}}= \frac{\frac{57}{80}}{\frac{43}{60}}=\frac{171}{172},\\
		\mathbb{P}(57\text{ is }A| \neg(57\text{ is }B))&=\frac{\mathbb{P}(57\text{ is }A)-\mathbb{P}(57\text{ is }A|57\text{ is }B)\mathbb{P}(57\text{ is }B)}{1-\mathbb{P}(57\text{ is }B)}\\
		&=\frac{\mu_A(57)-\min\{\mu_A(57), \mu_B(57)\}}{1-\mu_B(57)}=\frac{\frac{43}{60}-\frac{57}{80}}{1-\frac{57}{80}}=\frac{1}{69}.
	\end{align*}
We note that $\mathbb{P}(57\text{ is }A|57\text{ is }B)=1$.
	The probability value 1 here, means that when we certainly select $57$ as early, then we certainly select it as normal. Note that by changing the model  and the $t$-norm, these probability values might change as well (for instance, see Generalized Models and Random Generalized Models in Subsections \ref{generlizedmodel} and \ref{randomgeneralizedmodel}). 
	
	We also have that
	\begin{align*}
		\mathbb{P}(\xi_{X,B}=57|57\text{ is }A)&=\mathbb{P}(57\text{ is }B|(57\text{ is }A)\&(X=57))\mathbb{P}(X=57|57\text{ is }A).
	\end{align*}
	Now, let suppose selecting $57$ as $B$ and $X=57$ are  independent given "57 is $A$", and  also selecting $X=57$ and selecting 57 as $A$ are independent. Then, we have that
	\[\mathbb{P}(\xi_{X,B}=57|57\text{ is }A)=\mathbb{P}(57\text{ is }B|57\text{ is }A)\mathbb{P}(X=57)=\frac{171}{172}\times 0.005=0.00497093.\]

The Python code for this example is available  in \href{https://github.com/joseffaghihi/A-Fundamental-Probabilistic-Fuzzy-Logic-Framework-Suitable-for-Causal-Reasoning.git}{this}  Github repository.

	\begin{table}
		\Small
		\begin{center}
			\begin{tabular}{|l|c|c|c|c|c|c|c|c|c|c|c|c|c|c|c|}
				\hline
				Center of Interval & 5 & 15 & 25 & 35 & 45 & 55 &65 & 75 &85 & 95\\
				\hline
				Probability &‌0 & 0.0015 & 0.002 & 0.005 & 0.02 & 0.05 &0.085 & 0.105 & 0.13 &0.135 \\
				\hline
				Center of Interval & 105 & 115 & 125 & 135 & 145 & 155 &165 & 175 &185 & 195\\
				\hline
				Probability &‌0.12 & 0.1 & 0.075 & 0.06 & 0.04 & 0.03 &0.02 & 0.014 & 0.0065 &0.001 \\
				\hline
			\end{tabular}
			\caption{\Small Reproductive capability of a specific species.\\ Each probability value is associated to the interval with the given center and the length 10 days. For instance, the first interval is $[0,9]$. For the sake of simplicity, it is assumed that each interval is uniformly distributed.}\label{speciestable}
		\end{center}
	\end{table}
	
	\pgfplotstableread[row sep=\\,col sep=&]{
		interval & prob \\
		5    & 0  \\
		15   & 0.0015 \\
		25   & 0.002 \\
		35   & 0.005  \\
		45   & 0.02 \\
		55   & 0.05\\
		65   & 0.085\\
		75   & 0.105 \\
		85   & 0.13\\
		95  & 0.135\\
		105 & 0.12\\
		115 & 0.1\\
		125 & 0.075\\
		135 & 0.06\\
		145 & 0.04\\
		155 & 0.03\\
		165 & 0.02\\
		175 & 0.014\\
		185 & 0.0065\\
		195 & 0.001\\
	}\mydata
	
	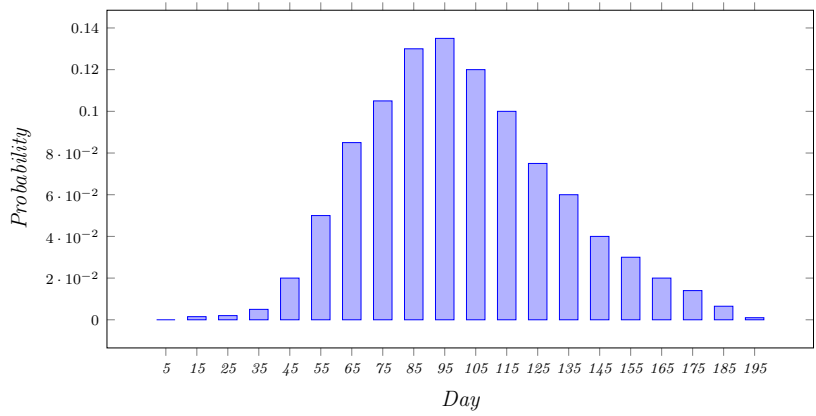
\begin{figure}
		\begin{center}
			\begin{tikzpicture}[scale=.7]
				\begin{axis}[
					ybar,
					symbolic x coords={5,15,25,35,45, 55, 65, 75, 85, 95, 105, 115, 125,135, 145,155, 165, 175, 185, 195},
					xtick=data, width=15cm, height=8cm, font=\tiny, xlabel ={\normalsize{Day}}, ylabel={\normalsize{Probability}}
					]
					\addplot table[x=interval,y=prob]{\mydata};
				\end{axis}
			\end{tikzpicture}
			\caption{\small The bar chart of Table \ref{speciestable}. Obviously, the probability values are not symmetrically distributed.}\label{barchart}
		\end{center}
	\end{figure}
	
	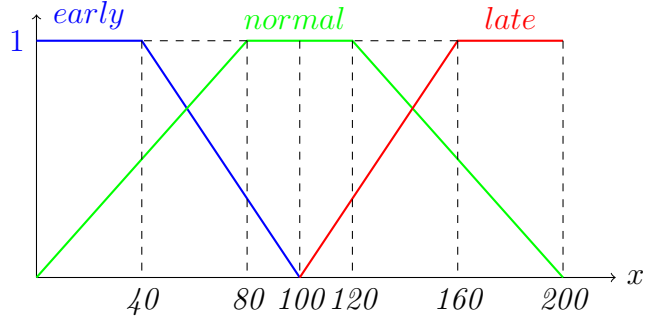
\begin{figure}
		\begin{center}
			\begin{tikzpicture}[scale=.7]
				\draw[->](0,0)--(11,0)node[right]{$ x $};
				\draw[->](0,0)--(0,5);
				\draw[blue, thick] (0,4.5)node[left]{$1$} --node[above]{$early$}(2,4.5)--(5,0);
				\draw[green, thick] (0,0) --(4,4.5)node[above]{$\qquad\quad normal$}-- (6,4.5)--(10,0);
				\draw[red, thick] (5,0) --(8,4.5)-- node[above]{$late$}(10, 4.5);
				\draw[dashed](5,4.5) --(5,0)node[below]{100};
				\draw[dashed](10,4.5) --(10,0)node[below]{200};
				\draw[dashed] (8,4.5)-- (6, 4.5);
				\draw[dashed] (4,4.5)-- (2, 4.5);
				\draw[dashed](2,4.5) --(2,0)node[below]{40};
				\draw[dashed](4,4.5) --(4,0)node[below]{80};
				\draw[dashed](6,4.5) --(6,0)node[below]{120};
				\draw[dashed](8,4.5)  --(8,0)node[below]{160};
			\end{tikzpicture}
		\end{center}
		\caption{\small Fuzzy attributes "$early$", "$normal$" and "$late$" for $\Omega$ described in Example \ref{exreproductive}.}\label{fuzzyreproductive}
	\end{figure}
\end{Example}

\subsection{Relative Fuzzy Model}\label{RFM}
Let $A$, $B$ and $C$ be three fuzzy subsets of $\Omega$. In this model, for any $x_0\in\Omega$, we define:
\[\mathbb{P}(x_0\text{ is }A)=\frac{\mu_A(x_0)}{\lVert x_0\rVert},\quad \lVert x_0\rVert:=\mu_A(x_0)+\mu_B(x_0)+\mu_C(x_0).\]
This model is useful, when we intend to select an element of $\Omega$ as $A$, $B$ or $C$. The standard conditional distributions are obtained as follows:
\[\mathbb{P}(y_0\text{ is }B|x_0\text{ is }A)=\frac{T\left(\frac{\mu_B(y_0)}{\lVert y_0\rVert},\frac{\mu_A(x_0)}{\lVert x_0\rVert} \right)}{\frac{\mu_A(x_0)}{\lVert x_0\rVert}}.\]
It follows that the above conditional probability coincides with the corresponding formula in Simple Fuzzy Model, when we use the minimum $t$-norm. 

\subsection{Membership Degree Scaled Model}
In this model, we scale values of $X$ with respect to their membership degrees. Indeed, we define:
\[\mathbb{P}(x_0\text{ is }A)=\mathbb{P}(X=\mu_A(x_0)x_0).\]
Thus, the standard conditional probabilities are defined as follows:
\[\mathbb{P}(y_0\text{ is }B|x_0\text{ is }A)=\frac{T\left(\mathbb{P}(Y=\mu_B(y_0)y_0),\mathbb{P}(X=\mu_A(x_0)x_0)\right)}{\mathbb{P}(X=\mu_A(x_0)x_0)}.\]
We can also define a non-standard conditional probability as follows:
\[\mathbb{P}(y_0\text{ is }B|x_0\text{ is }A)=\mathbb{P}(Y=y_0\mu_B(y_0)|X=x_0\mu_A(x_0)).\]

\subsection{Generalized Membership Degree Based Models}
In these models we replace $\mu_A(x_0)$ in the previous models by $\mu_A^{r_{x_0,A}}(x_0)$,
where $r_{x_0,A}$ is a non-negative real number. A particular case is to consider
$\mu_A^r(x_0)$
for any $x_0\in \Omega$ and a non-negative number $r$. Then, by the approaches explained in other models, we are able to define all subjective joint and conditional distributions by replacing any membership degree $\mu_D(x_0)$ with $\mu_D(x_0)^{r_{x_0,D}}$ for any fuzzy attribute $D$ of $\Omega$ and $x_0\in\Omega$.

 Note that for any $x\in \Omega$ with $r_{x,A}>1$ and $\mu_A(x)\notin\{0,1\}$, we have that $\mu_A(x)>\mu_A^{r_{x,A}}(x)$. Further, for any $x\in \Omega$ with $r_{x,A}<1$ and $\mu_A(x)\notin\{0,1\}$, we have that $\mu_A(x)<\mu_A^{r_{x,A}}(x)$. Therefore, when the membership degrees are large, by selecting $r_{x,A}>1$, we can have smaller values. 
Similarly, when we want to have larger values of membership degrees, we select $r_{x,A}<1$.

\subsection{Generalized Standard Models}\label{generlizedmodel}

When we use $t$-norms (e.g., in standard models), we might need to normalize the probability of selecting an element  with respect to a fuzzy attribute. Hence, we call a model of our framework  a {\it generalized standard model} if it satisfies the following equality:
\[\mathbb{P}(y_0\text{ is }B|x_0\text{ is }A)=\frac{T(r\mathbb{P}(y_0\text{ is }B), \mathbb{P}(x_0\text{ is }A))}{\mathbb{P}(x_0\text{ is }A)},\]
where $r$ is a positive real number, which could be fixed or depended on $A$ and $B$ or depended on $x_0$, $y_0$, $A$ and $B$. 

\subsection{Random Generalized Standard Models}
\label{randomgeneralizedmodel}
These models are obtained when we use a random variable instead of the real number $r$ in generalized standard models. For instance, assume that we would like to double the probabilities of selecting  elements with respect to a fuzzy attribute $B$  in $70$ percent of times. Hence, instead of $r$, we can use a Bernoulli random variable  that  only depends on $A$ and $B$  with the success probability of $0.7$. More precisely, we have that

\[\mathbb{P}(y_0\text{ is }B|x_0\text{ is }A)=\frac{T((Z+1)\mathbb{P}(y_0\text{ is }B), \mathbb{P}(x_0\text{ is }A))}{\mathbb{P}(x_0\text{ is }A)},\]
where $Z$ is a Bernoulli random variable with the success probability of $0.7$.
\section{Fuzzy Average Treatment Effect}\label{FATE}

Assume that we intend to measure the causal effect of a discrete (but not necessarily binary) intervention $T$  on a Bernoulli outcome $Y$. For instance, $T$ could be a type of medical treatment and $Y$ be the outcome of $T$. Indeed, $Y=1$ when the treatment is successful and otherwise $Y=0$. We also denote the potential outcome corresponding to $T=t$ by $Y(t)$. A significant way to measure the causal effect of $T$ on $Y$ is to see the difference between the outcomes when $T$  is $high$ and when $T$ is $low$. We might also be  interested in $medium$ treatments.  Thus, we assume that $low$, $medium$ and $high$ are fuzzy attributes of $T$, and then we use our PFL theory as follows. Given a statistical population, we define:
\[Y(A):=\frac{1}{\mathbb{P}(T\text{ is }A)}\sum_{t\neq t_A}Y(t)\mathbb{P}(\xi_{T,A}=t),\]
where $A\in\{low, medium, high\}$. 

Now, we define the \textit{fuzzy average treatment effects (FATE)} as follows:
\begin{align*}
	\mathrm{FATE}_{l}^{h}&:=\mathbb{E}(Y(high))-\mathbb{E}(Y(low)),\\
	\mathrm{FATE}_{l}^{m}&:=\mathbb{E}(Y(medium))-\mathbb{E}(Y(low)),\\
	\mathrm{FATE}_{m}^{h}&:=\mathbb{E}(Y(high))-\mathbb{E}(Y(medium)).
\end{align*}
Obviously, we have that
\[\mathrm{FATE}_{l}^{h}=\mathrm{FATE}_{l}^{m}+\mathrm{FATE}_{m}^{h}.\]
To deal with the fundamental problem of Causal Inference, there are some assumptions in classic causal inference such as  the ignorability condition (see Section \ref{cuasalinference}).  One could create a fuzzy version of  the ignorability condition using similar assumptions. Our fuzzy ignorability condition is the following assumption:

\begin{itemize}
	\item
	$Y(t)$ and $\xi_{A,T}$ are independent for any $t\in T$ and $A\in\{low, high\}$.
\end{itemize}  
Thus, to compute $\mathrm{FATE}_l^h$, we have that
\begin{align*}
	\mathbb{E}(Y(high))&=\frac{1}{\mathbb{P}(T\text{ is }high)}\sum_{t\neq t_{high}}\mathbb{E}(Y(t))\mathbb{P}(\xi_{T,high}=t)\\
	&=\frac{1}{\mathbb{P}(T\text{ is }high)}\sum_{t\neq t_{high}}\mathbb{E}(Y(t)|\xi_{T,high}=t)\mathbb{P}(\xi_{T,high}=t)\\
	&=\frac{1}{\mathbb{P}(T\text{ is }high)}\sum_{t\neq t_{high}}\mathbb{P}(Y(t)=1|\xi_{T,high}=t)\mathbb{P}(\xi_{T,high}=t)\\
	&=\frac{1}{\mathbb{P}(T\text{ is }high)}\sum_{t\neq t_{high}}\mathbb{P}(Y=1|\xi_{T,high}=t)\mathbb{P}(\xi_{T,high}=t)\\
	&=\frac{1}{\mathbb{P}(\xi_{T,high}\neq t_{high})}\sum_{t\neq t_{high}}\mathbb{P}(\xi_{T,high}=t|Y=1)\mathbb{P}(Y=1)\\
	&=\frac{\mathbb{P}(\xi_{T,high}\neq t_{high}|Y=1)\mathbb{P}(Y=1)}{\mathbb{P}(\xi_{T,high}\neq t_{high})}\\
	&=\mathbb{P}(Y=1|\xi_{T,high}\neq t_{high})\\
	&=\mathbb{E}(Y|\xi_{T,high}\neq t_{high}),
\end{align*}
which is a statistical formula. Similarly, we can obtain $\mathbb{E}(Y(low))$ as follows:
\[\mathbb{E}(Y(low))=\mathbb{E}(Y|\xi_{T,low}\neq t_{low}),\]
which implies that
\[\mathrm{FATE}_l^h=\mathbb{E}(Y|\xi_{T,high}\neq t_{high})-\mathbb{E}(Y|\xi_{T,low}\neq t_{low}).\]
In practice, we suggest to fuzzify the treatment into three fuzzy subsets $low$, $medium$ and $high$ in such a way that 
\[\mathbb{P}(t\text{ is }low)+\mathbb{P}(t\text{ is }medium)+\mathbb{P}(t\text{ is }high)=1\]
for any $t\in \mathrm{Supp}(T)\backslash\{t_{high}\}$ (for instance, see Relative Fuzzy Model described in Subsection \ref{RFM}). This follows that
\[\mathbb{P}(T\text{ is }low)+\mathbb{P}(T\text{ is }medium)+\mathbb{P}(T\text{ is }high)=1.\]
Then, for any treatment $t\in \mathrm{Supp}(T)\backslash\{t_{high}\}$, we randomly assign $t$ to a proportion $\mathbb{P}(\xi_{T,high}=t)$ of the entire sample. Thus, a proportion  $\mathbb{P}(T\text{ is }low)+\mathbb{P}(T\text{ is }medium)$ of the entire sample is untreated. Next, we randomly assign  $\xi_{T,medium}$ to a proportion  $\mathbb{P}(T\text{ is }medium)$ of the entire sample, but only to untreated units. Thus, a proportion $\mathbb{P}(T\text{ is }low)$ of the entire sample remains untreated. This proportion of the entire sample could be treated using $\xi_{T,low}$.

In the following, we provide an example of FATE.

\begin{Example}
	Assume that we are given a uniformly distributed treatment $T$ whose set of values is $\{0,1,\ldots,9\}$. We randomly generate the binary potential outcomes of a sample of 10000 people in such a way that the probability of happening $Y(t)=1$ is more probable when $t$ increases. The head and the tail of the sample are shown in Figure~\ref{sample}. Let the fuzzy attributes "$low$", "$medium$" and "$high$" be defined like the ones in Figure~\ref{lmh}. Then, by considering Simple Fuzzy Model (i.e., $\mathbb{P}(x\text{ is }A)=\mu_A(x)$) we have that
	\[\mathrm{FATE}_l^m=0.227241,\qquad \mathrm{FATE}_l^h=0.598788,\]
	while by considering the ignorability condition and applying the procedure explained just before this example, we get the followings:
	\begin{align*}
		&\mathbb{E}(Y|\xi_{T,med}\neq t_{med})-\mathbb{E}(Y|\xi_{T,low}\neq t_{low})\approx 0.227027,\\
		&\mathbb{E}(Y|\xi_{T,high}\neq t_{high})-\mathbb{E}(Y|\xi_{T,low}\neq t_{low})\approx 0.598198,
	\end{align*}
	which are approximately equal to $\mathrm{FATE}_l^m$ and $\mathrm{FATE}_l^h$, respectively. These results show that the treatment $T$ in our context could be considered as a  probable cause of the outcome $Y$. 
	
The Python code for this example is available  in \href{https://github.com/joseffaghihi/A-Fundamental-Probabilistic-Fuzzy-Logic-Framework-Suitable-for-Causal-Reasoning.git}{this}  Github repository.
	\begin{figure}
		\begin{center}
			\includegraphics[scale=.5]{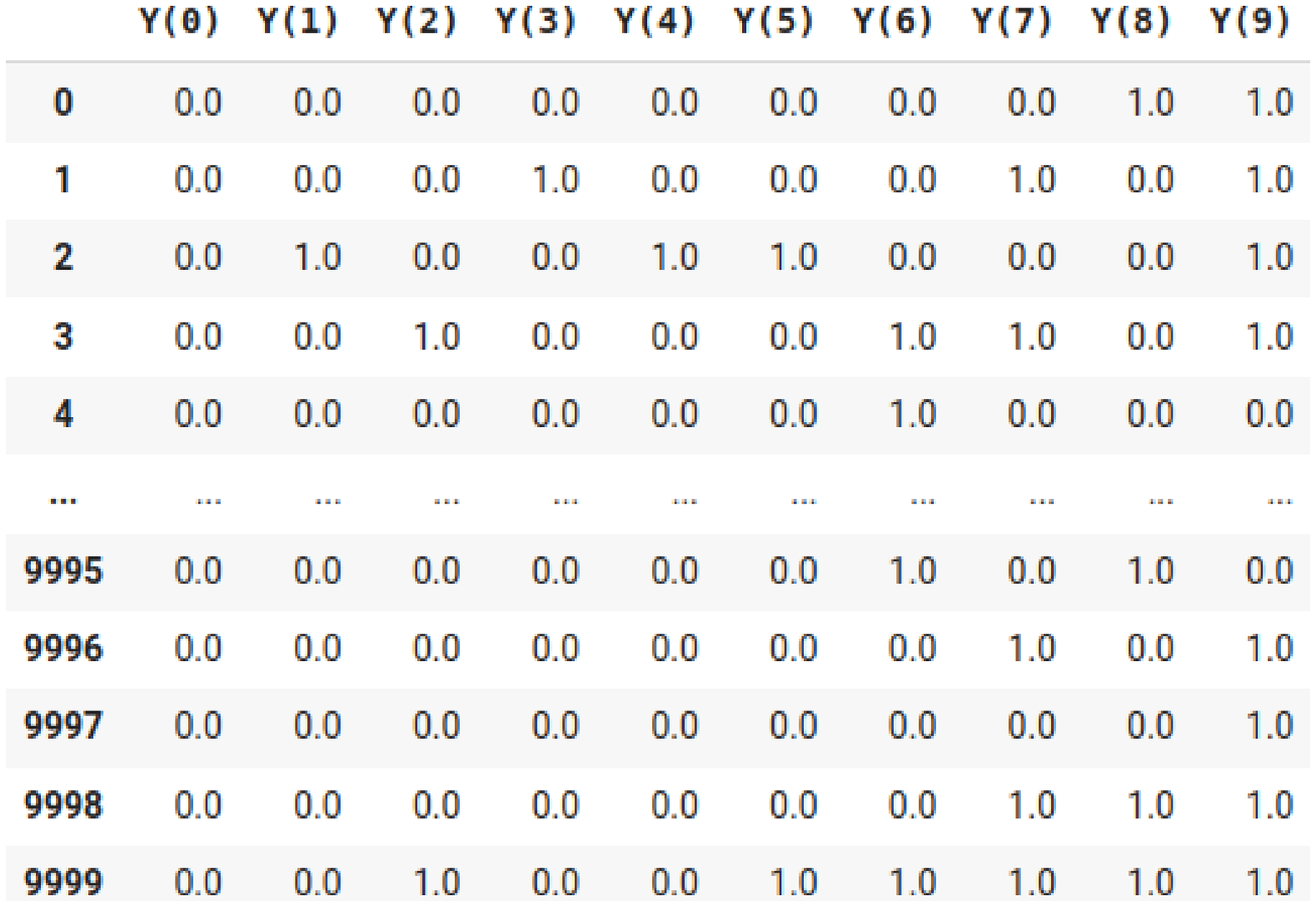}
			
			\caption{\small The head and the tail of the random sample of binary potential outcomes in such a way that the probability of the $Y(t)=1$ gets higher when $t$ increases. }\label{sample}
		\end{center}
	\end{figure}
	
	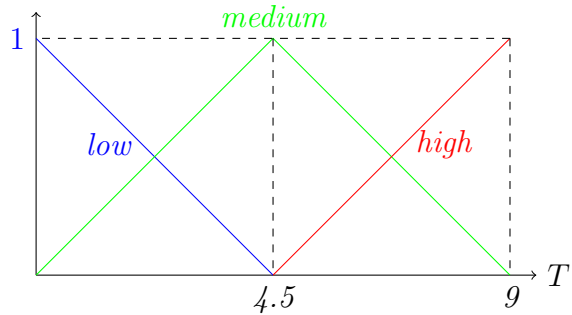
\begin{figure}
		\begin{center}
			\begin{tikzpicture}[scale=.7]
				\draw[->](0,0)--(9.5,0)node[right]{$ T $};
				\draw[->](0,0)--(0,5);
				\draw[blue] (0,4.5)node[left]{$1$} -- (2,2.5) node[left]{low}--(4.5,0);
				\draw[green] (0,0) --(4.5,4.5)node[above]{medium}-- (9,0);
				\draw[red] (4.5,0) -- (7,2.5)node[right]{high}--(9, 4.5);
				\draw[dashed](4.5,4.5) --(4.5,0)node[below]{4.5};
				\draw[dashed](9,4.5) --(9,0)node[below]{9};
				\draw[dashed] (9,4.5)-- (0, 4.5);
			\end{tikzpicture}
		\end{center}
		\caption{\small Fuzzy attributes "$low$", "$medium$" and "$high$" for $T=\{0,1,\ldots,9\}$.}\label{lmh}
	\end{figure}
\end{Example}

\section{Measure Theoretical Approach}\label{Measure Theoretical Approach}
Let $(\Omega, \mathscr{F}, \mathbb{P})$ be a probability space, where $\Omega\subseteq\mathbb{R}$ is the set of possible outcomes, $\mathscr{F}$ is a $\sigma$-algebra over $\Omega$ consisting of all events, and $\mathbb{P}$ is a measure on $(\Omega,\mathscr{F})$ with $\mathbb{P}(\Omega)=1$. We assume that $\mathscr{F}$ is the Borel $\sigma$-algebra on $\Omega$. Now, let $X:\Omega\to\mathbb{R}$ be the identity random variable, and $A$ be a fuzzy attribute of $\Omega$. Hence, the probability measure induced by $X$, denoted by $\mathbb{P}_X$, coincides with $P$. For the sake of simplicity, in the sequel, we denote all probability measures and their derived notations (e.g., $\mathbb{P}_X$)  by $\mathbb{P}$, when the underlying measurable space $(\Omega,\mathscr{F})$ is known. For any $x\in \Omega$, assume that $\mathbb{P}(x\text{ is } A)$ is the chance of selecting $x$ as $A$ in the random experiment whose sample space is $\{\text{Selecting $x$ as $A$},\text{Not selecting $x$ as $A$}\}$, and it is measurable as a function of $x$. Considering Experiment ($\star$), We have the following assumption for any $x\in \Omega$:
\[\mathbb{P}(\Omega\text{ is }A|X=x)=\mathbb{P}(x\text{ is } A).\]
It follows that
\[\mathbb{P}(\Omega\text{ is } A)=\int_{\mathbb{R}}\mathbb{P}(\Omega\text{ is }A|X=x)\,\mathrm{d}\mathbb{P}_X(x)=\int_{\mathbb{R}}\mathbb{P}(x\text{ is } A)\,\mathrm{d}\mathbb{P}_X(x).\]
where, $\mathbb{P}_X$ is the probability measure on $\mathbb{R}$ induced by $X$. 
Let $\xi_{X,A}$ be the random variable associated to Experiment ($\star$). Similar to the above, we consider the following assumption for any $x\in \Omega$ and $E\in\mathscr{F}$:
\[ \mathbb{P}(\xi_{X,A}\in E|X=x)=\mathbb{P}(\xi_{x,A}\in E).\]
Let $E\in\mathscr{F}$. Then, by the total law of probability, we have that:
\begin{align*}
    \mathbb{P}(\xi_{X,A}\in E)&=\int_{\mathbb{R}}\mathbb{P}(\xi_{X,A}\in E|X=x)\,\mathrm{d}\mathbb{P}_X(x)=\int_{\mathbb{R}}\mathbb{P}(\xi_{x,A}\in E)\,\mathrm{d}\mathbb{P}_X(x)\\
    &=\int_E\mathbb{P}(\xi_{x,A}\in E)\,\mathrm{d}\mathbb{P}_X(x)+\int_{{\mathbb{R}}\backslash E}\mathbb{P}(\xi_{x,A}\in E)\,\mathrm{d}\mathbb{P}_X(x).
\end{align*}
If $x_A\in E$, then:
\begin{align*}
    \int_E\mathbb{P}(\xi_{x,A}\in E)\,\mathrm{d}\mathbb{P}_X(x)&=\int_E\left(\mathbb{P}(\xi_{x,A}=x)+\mathbb{P}(\xi_{x,A}=x_A)\right) \,\mathrm{d}\mathbb{P}_X(x)=\int_E\,\mathrm{d}\mathbb{P}_X(x),\\
   \int_{\mathbb{R}\backslash E}\mathbb{P}(\xi_{x,A}\in E)\,\mathrm{d}\mathbb{P}_X(x)&=\int_{\mathbb{R}\backslash E}\mathbb{P}(\xi_{x,A}=x_A)\,\mathrm{d}\mathbb{P}_X(x)\\
   &=\int_{\mathbb{R}\backslash E}\left(1-\mathbb{P}(x\text{ is } A)\right)\,\mathrm{d}\mathbb{P}_X(x).
\end{align*}
Further, if $x_A\notin E$, then:
\begin{align*}
    &\int_E\mathbb{P}(\xi_{x,A}\in E)\,\mathrm{d}\mathbb{P}_X(x)=\int_E\mathbb{P}(\xi_{x,A}=x) \,\mathrm{d}\mathbb{P}_X(x)=\int_{ E}\mathbb{P}(x\text{ is } A)\,\mathrm{d}\mathbb{P}_X(x),\\
   & \int_{\mathbb{R}\backslash E}\mathbb{P}(\xi_{x,A}\in E)\,\mathrm{d}\mathbb{P}_X(x)=0.
\end{align*}
It follows from $\int_{\mathbb{R}}\,\mathrm{d}\mathbb{P}_X(x)=1$ that:
\[\mathbb{P}(\xi_{X,A}\in E)=\left\{\begin{array}{ll}
 1-\displaystyle \int_{\mathbb{R}\backslash E}\mathbb{P}(x\text{ is } A)\,\mathrm{d}\mathbb{P}_X(x)   &,x_A\in E  \\
   \displaystyle \int_{ E}\mathbb{P}(x\text{ is } A)\,\mathrm{d}\mathbb{P}_X(x)  & , x_A\notin E
\end{array}\right..\]
We can rewrite the above formula as follows:
\begin{equation*}
	\mathbb{P}(\xi_{X,A}\in E)=\int_{ E}\mathbb{P}(x\text{ is }A)\,\mathrm{d}\mathbb{P}_X(x)+\alpha\delta_{x_A}(E),\quad \alpha=1-\int_{ \mathbb{R}}\mathbb{P}(x\text{ is }A)\,\mathrm{d}\mathbb{P}_X(x),
\end{equation*}
where $\delta_{x_A}$ is the Dirac measure with respect to $x_A$.
The above formula equals the following as well:
\begin{equation*}
    \mathbb{P}(\xi_{X,A}\in E)=\int_{ E}\mathbb{P}(\xi_{x,A}=x)\,\mathrm{d}\mathbb{P}_X(x)+\beta\delta_{x_A}(E),\quad \beta=1-\int_{ \mathbb{R}}\mathbb{P}(\xi_{x,A}=x)\,\mathrm{d}\mathbb{P}_X(x),
\end{equation*}

Now, we determine the expected value of $\xi_{X,A} $ in the following theorem.
\begin{Theorem}\label{expectedvalueingeneral}
\[\mathbb{E}(\xi_{X,A})=x_A+\int_{\mathbb{R}}(x-x_A)\mathbb{P}(x\text{ is }A)\,\mathrm{d}\mathbb{P}_X(x).\]
\end{Theorem}
\begin{proof}
Let $E\in\mathscr{F}$. Then, we have that
\begin{align*}
\int_{ E}\mathbb{P}(\xi_{x,A}=x)\,\mathrm{d}(\mathbb{P}_X+\beta\delta_{x_A})&=\int_{ E}\mathbb{P}(\xi_{x,A}=x)\,\mathrm{d}\mathbb{P}_X+\beta\int_{ E}\mathbb{P}(\xi_{x,A}=x)\,\mathrm{d}\delta_{x_A}\\
&=\int_E\mathbb{P}(\xi_{x,A}=x)\,\mathrm{d}\mathbb{P}_X+\beta\underbrace{\mathbb{P}(\xi_{x_A,A}=x_A)}_{=1}\delta_{x_A}(E)\\
&=\mathbb{P}(\xi_{X,A}\in E).
\end{align*}
It follows that 
\begin{align*}
    \mathbb{E}(\xi_{X,A})&=\int_{\mathbb{R}}x\,\mathrm{d}\mathbb{P}_{\xi_{X,A}}=\int_{\mathbb{R}}x\mathbb{P}(\xi_{x,A}=x)\,\mathrm{d}(\mathbb{P}_X+\beta\delta_{x_A})\\
    &=\int_{\mathbb{R}}x\mathbb{P}(\xi_{x,A}=x)\,\mathrm{d}\mathbb{P}_X+\beta\int_{\mathbb{R}}x\mathbb{P}(\xi_{x,A}=x)\,\mathrm{d}\delta_{x_A}\\
    &=\int_{\mathbb{R}}x\mathbb{P}(\xi_{x,A}=x)\,\mathrm{d}\mathbb{P}_X+\beta x_A\mathbb{P}(\xi_{x_A,A}=x_A)\\
    &=\int_{\mathbb{R}}x\mathbb{P}(\xi_{x,A}=x)\,\mathrm{d}\mathbb{P}_X+ x_A\left(1-\int_{\mathbb{R}}\mathbb{P}(\xi_{x,A}=x)\,\mathrm{d}\mathbb{P}_X\right)\\
    &=x_A+\int_{\mathbb{R}}(x-x_A)\mathbb{P}(\xi_{x,A}=x)\,\mathrm{d}\mathbb{P}_X\\
     &=x_A+\int_{\mathbb{R}}(x-x_A)\mathbb{P}(x\text{ is }A)\,\mathrm{d}\mathbb{P}_X.
\end{align*}
\end{proof}
In the following theorem, we provide the  classic continuous distributions of $\xi_{X,A}$. To do so, we initially assume that $\mathbb{P}(x\text{ is }A)$ and the probability density function (pdf) of $X$ are continuous functions of the variable $x$, and hence they are Riemann integrable.

\begin{Theorem}\label{theoremcont}
Let $X$ be a classic continuous distribution whose  pdf is $f_X(x)$. Then, $\xi_{X,A}$ has a generalized pdf, denoted by $f_{\xi_{X,A}}$, and it is given by the following:
\[ f_{\xi_{X,A}}(t)=\mathbb{P}(t\text{ is } A)f_X(t)+\alpha\delta(t-x_A),\]
where $\delta(x-x_A)$ is the Dirac delta function, and we have that
\[\alpha=1-\int_{ -\infty}^{\infty}\mathbb{P}(x\text{ is }A)f_X(x)\,\mathrm{d}x.\]
\end{Theorem}
\begin{proof}
We have that $\mathrm{d}\mathbb{P}_X=f(x)\,\mathrm{d}x$. First, we note that 
\[\alpha=1-\int_{ \mathbb{R}}\mathbb{P}(x\text{ is }A)\,\mathrm{d}\mathbb{P}_X(x)=1-\int_{ -\infty}^{\infty}\mathbb{P}(x\text{ is }A)f_X(x)\,\mathrm{d}x.\]
Now, we show that $\xi_{X,A}$ has a generalized pdf (see Appendix \ref{gpdf}). To do so, it follows from the proof of Theorem \ref{expectedvalueingeneral} that 
\[\mathbb{P}(\xi_{X,A}\in E)=\int_{ E}\mathbb{P}(\xi_{x,A}=x)\,\mathrm{d}(\mathbb{P}_X+\beta\delta_{x_A}),\]
while $\mathbb{P}_X$ is absolutely continuous (see Appendix \ref{continuousrandomvariables}). Hence, $\mathbb{P}_X+\beta\delta_{x_A}$ is absolutely continuous with respect to $\lambda+\beta\delta_{x_A}$, where $\lambda$ is the Lebesgue measure on $\mathbb{R}$. It follows that $\xi_{X,A}$ is absolutely continuous with respect to $\lambda+\beta\delta_{x_A}$, and hence it has a generalized pdf. 

Denote the distribution function of $\xi_{X,A}$ by $F$. Then, we have that
\begin{align*}
     F(t)=\mathbb{P}(\xi_{X,A}\le t)&=\int_{-\infty}^t\mathbb{P}(x\text{ is } A)f_X(x)\,\mathrm{d}x+\alpha\delta_{x_A}(-\infty,t]\\
     &=\int_{-\infty}^t\mathbb{P}(x\text{ is } A)f_X(x)\,\mathrm{d}x+\alpha U_{x_A}(t),
\end{align*}

where $U_{x_A}$ is the Heaviside step function.
It follows that
\[ f_{\xi_{X,A}}(t)=\frac{\mathrm{d}F}{\mathrm{d}x}(t)=\mathbb{P}(t\text{ is } A)f_X(t)+\alpha\delta(t-x_A).\]
\end{proof}

 By a similar process (i.e., taking the total law of probability over the values of $X$ and using the Bayes rule), we can define the conditional probabilities $\mathbb{P}(\xi_{X,A}\in E'|E)$ and $\mathbb{P}(E|\xi_{X,A}\in E')$ that we described in the discrete case in Section \ref{frame}. For instance, we have that
\begin{align*}
	&\mathbb{P}(\xi_{X,A}\in E'|E)=\left\{\begin{array}{ll}
		1-\displaystyle \int_{\mathbb{R}\backslash E'}\mathbb{P}(x\text{ is } A|X=x,E)\,\mathrm{d}\mathbb{P}^{E}_X   &,x_A\in E'  \\
		\displaystyle \int_{ E'}\mathbb{P}(x\text{ is } A|X=x,E)\,\mathrm{d}\mathbb{P}^{E}_X  & , x_A\notin E'
	\end{array}\right.,\\
&\mathbb{E}(\xi_{X,A}|E)=x_A+\int_{\mathbb{R}}(x-x_A)\mathbb{P}(x\text{ is }A|X=x, E)\,\mathrm{d}\mathbb{P}^{E}_X,
\end{align*}
where $\mathbb{P}^E$ is the probability measure on $(\Omega,\mathscr{F})$ given $E$, and $\mathbb{P}_X^E$ is the probability measure induced by $X$ on $\mathbb{R}$ with respect to $\mathbb{P}^E$.

Now, we have the following theorem, which is similar to Theorem \ref{theoremcont} but in conditional case. In this theorem, we assume that $f_{X|E}(t|E)$ and $\mathbb{P}(t\text{ is } A|X=t,E)$ are continuous functions of $t$. 

\begin{Theorem}\label{cond}
	Let $X$ be a classic continuous distribution whose probability density function (pdf) is $f_X(x)$. Then, the generalized pdf of $\xi_{X,A}$ given $E$, denoted by $f_{\xi_{X,A}|E}$, is as follows:
	\[ f_{\xi_{X,A}|E}(t|E)=\mathbb{P}(t\text{ is } A|X=t,E)f_{X|E}(t|E)+\alpha^E\delta(t-x_A),\] 
	where, $\alpha^E=1-\int_{ -\infty}^{\infty}\mathbb{P}(x\text{ is }A|X=x,E)f_{X|E}(x|E)\,\mathrm{d}x$.
\end{Theorem}
The above theorem follows the following  corollaries:

\begin{Corollary}
	Considering the assumption of Theorem \ref{cond}, if $\xi_{x,A}$ and $X$ are independent given $E$, then we have that
	\[ f_{\xi_{X,A}|E}(t|E)=\mathbb{P}(t\text{ is } A|E)f_{X|E}(t|E)+\alpha^E\delta(t-x_A),\] where, $\alpha^E=1-\int_{ -\infty}^{\infty}\mathbb{P}(x\text{ is }A|E)f_{X|E}(x|E)\,\mathrm{d}x$.
\end{Corollary}

\begin{Corollary}
	Considering the assumption of Theorem \ref{cond}, if $E$ is the event of "$\xi_{X,B}=x_B$", and $B$ is a fuzzy attribute of $X$, then  we have that
		\[ f_{\xi_{X,A}|E}(t|E)=
			\mathbb{P}(t\text{ is }A|\neg(t\text{ is }B))f_{X|E}(t|E)+\alpha\delta(t-x_A),\]
		where, $\alpha^E=1-\int_{ -\infty}^{\infty}\mathbb{P}(x\text{ is }A|\neg(x\text{ is }B))f_{X|E}(x|E)\,\mathrm{d}x$.
\end{Corollary}

Let's finish this section by the following example:

\begin{Example}
	Let $\Omega$ denotes the rainfall in a specific region and $\Omega'$ be collected water on the ground at a certain place located in the region. Consider the fuzzy attributes $"high"$ and $"wet"$  for  $\Omega$  and   $\Omega'$, respectively. Also, assume that $X$ and $Y$ are the inclusion random variables corresponding to $\Omega$ and $\Omega'$, respectively.  One could think of being $"wet"$ for a specific $y_0$  not depending on the joint distribution of the following events:
	\begin{itemize}
		\item $X=x$,
		\item $Y=y$, and
		\item  selecting $x'$ as $"high"$
		\end{itemize}  
	 for any $x,x'\in\mathrm{Supp}(X)$ and $y_0,y\in\mathrm{Supp}(Y)$. Here, since all conditions are  removed, the above formulas for conditional probabilities are not applicable. Hence, consider the case that we drop the third one (i.e. selecting $x'$ as "high").
	  However, for greater values of $X=x$ and $Y=y$, since grater values of $Y$ are more likely available, we might select $y_0$ with a lower probability as $"wet"$. Also, being $"high"$ depends on $X$, since for instance what is considered high rainfall in London is significantly different from what is considered high rainfall in Lut desert (located in Iran). Hence, dependencies of the above random variables come from the subjective decisions that might be made by criteria such as the required precision. Therefore, by considering $x_{high}=y_{wet}=0$ and the aforementioned independence in selecting $y_0$ as $"wet"$ (without the third one), we have the followings:
	\begin{align*}
		&\mathbb{P}(y_0 \text{ is } wet|Y=y,X=x)=\mathbb{P}(y_0\text{ is }wet),\\
		&f_{\xi_{Y,wet}}(y_0)=\mathbb{P}(y_0\text{ is }wet)f_Y(y_0),\\
		&\mathbb{E}(\xi_{Y,wet})=\int_{-\infty}^{\infty}y\mathbb{P}(y\text{ is }wet)f_Y(y)\,\mathrm{d}y,\\
		&f_{\xi_{Y,wet}|X}(y_0|x)=\mathbb{P}(y_0\text{ is }wet)f_{Y|X}(y_0|x),\\
		&f_{Y|\xi_{X,high}}(y_0|x)=f_{Y|E}(y_0|E),\quad E=(x\text{ is }high)\& (X=x),
	\end{align*}
where $y_0\in\mathrm{Supp}(Y)\backslash\{0\}$ and $x\in\mathrm{Supp}(X)\backslash\{0\}$.
Now, if we assume that selecting $x$ as "high" is independent from the joint distribution of $X$ and $Y$ as well, then the last  formula can be written as follows:
\begin{align*}
	&f_{Y|\xi_{X,high}}(y_0|x)=f_{Y|X}(y_0|x).
\end{align*}
\end{Example}

\section{Relationship Between Zadeh's PFL and our Theory}\label{Relationship Between Zadeh's PFL and our Theory}
Zadeh's PFL is different from our theory in nature. Indeed, Zadeh's theory is a probability theory for fuzzy events, while our theory is about the probability of selecting crisp elements with a certain fuzzy attribute. Anyway, several aspects of Zadeh's theory could be interpreted in our theory. Consider a probability  space $(\Omega,\mathscr{F}, P)$, where $\Omega\subseteq\mathbb{R}$ and $\mathscr{F}$ is the Borel $\sigma$-algebra on $\Omega$. Let $X:\Omega\to\mathbb{R}$ be the inclusion random variable, and let $A$ be a fuzzy attribute of $\Omega$. Zadeh defined the probability of $\widetilde{A}$ (i.e. the fuzzy subset of $\Omega$ associated to $A$) as follows:
\[ P(\widetilde{A})=\int_{\mathbb{R}}\mu_A(x)\,\mathrm{d}P_X,\]
where $P_X$ is the probability function induced by $X$.
 The above probability in our Simple Fuzzy Model is equal to $\mathbb{P}(\Omega\text{ is } A)$ (see Section \ref{Measure Theoretical Approach}). Zadeh also defined a mean for $P(\widetilde{A})$ as follows:
\[m(\widetilde{A})=\frac{1}{P(\widetilde{A})}\int_{\mathbb{R}}x\mu_A(x)\,\mathrm{d}P_X.\]
Note that $m(\widetilde{A})$ in our Simple Fuzzy Model could be obtained as follows:
\[ m(\widetilde{A})=\frac{\mathbb{E}\left(\xi_{X,A}\right)}{\mathbb{P}(\Omega\text{ is } A)},\]
where $x_A=0$.

\section{Conclusion}\label{Conclusion}

One major problem with Deep Learning Algorithms is that they fail when it comes to reasoning (see \cite{faghihi2020association}, \cite{madan2021fast}, and \cite{pearl2018book}). Hence, it has been suggested in \cite{faghihi2020cog} and \cite{faghihi2020association}  to use Non-Classical Logics to equip the machine with reasoning. In this paper, we constructed a PFL framework that can solve basic problems but fundamental in some different types of reasoning such as Causal Reasoning. Our PFL suggests a probabilistic screening criterion to feed a machine with datapoints satisfying a specific fuzzy attribute at a time (see Section \ref{Problem Setting and Some Motivations}). Further, it provides an assignment mechanism to calculate the fuzzy average treatment effect for Causal Inference. Our Fuzzy Causal Inference framework is flexible enough to integrate the vagueness of the selection criteria into different Causal Inference concepts such as ignorability condition. Note that as we mentioned in Section \ref{Related Works}, the previous PFL frameworks have at least one of the following drawbacks:
 \begin{itemize}
 	\item The lack of a precise mathematical setting.
 	\item
 	Providing a framework for which events are just fuzzy, and hence it is not possible to have crisp events while the criteria of the random selection are fuzzy.
 	\item
 Computationally expensive.
 \end{itemize}
Our PFL framework solves all of the above drawbacks. 

In the near future, we intend to create an integration theory of Fuzzy Logic and Causal Inference by using our PFL framework. Although we very briefly touched the concepts of Causal Inference in this paper,  we prepared the ground for the fully integration of PFL and Causal Inference in general. To demonstrate this, we showed an  application of our PFL theory in Causal Inference. 

We recommend other researchers to take part in developing the PFL framework suggested in this paper. For instance, the followings are recommended:
\begin{itemize}
	\item
	Working on the probabilistic aspects of our PFL theory such as studying the independence, dependence and correlation of random variables, and studying  famous inequalities and theorems  in Probability Theory.
	\item Constructing new models of our PFL capable of processing specific needs for specific applications in the areas  such as Medicine, Market Research, Social Studies and Control Systems. 
	\item Working on the philosophical aspects of our PFL theory. A notable work in our research  was to interpret "selecting {\it nothing }" as "selecting a specific element". This could be controversial for some researchers and philosophers.
\end{itemize}

\section*{Acknowledgments} 

The research of Amir Saki has been supported by the Institute for Research in Fundamental Sciences (IPM).


\bibliographystyle{amsplain}
\bibliography{References}

\appendix
\section{Topological Spaces}\label{1}
Let $f:\mathbb{R}\to\mathbb{R}$ be a function. Then, $f$ is continuous if and only if $f^{-1}(I)$ is a union of open intervals for any open interval $I\subseteq\mathbb{R}$. Continuous functions are important for many mathematical purposes such as integration. To generalize continuity on arbitrary spaces (not only $\mathbb{R}$), one common and popular way is to define the so called open sets. 

\subsection{Definition of a Topology}

Let $\Omega$ be a set. A collection $\tau$ of subsets of $\Omega$ is called a topology on $\Omega$, whenever 
\begin{itemize}
	\item 
	$\emptyset, \Omega\in\tau$, 
	\item
	$U,V\in\tau$ implies that $U\cap V\in\tau$, and
	\item
	for any family $\{U_i\}_{i\in I}$ of elements $\tau$, we have that $\bigcup_{i\in I}U_i\in\tau$.
\end{itemize}
Then,  $(\Omega, \tau)$ (or simply $\Omega$) is called a topological space. Moreover, each element of $\tau$ is called an open subset of $\Omega$. The complement of each open subset of $\Omega$ is called a closed subset of $\Omega$. 

\subsection{Base for a Topology}

Let $(\Omega,\tau)$ be a topological space. A subcollection $\mathcal{B}=\{B_i\}_{i\in I}$ of $\tau$ is called a base for $\tau$, if 
\begin{itemize}
	\item 
	$\bigcup_{i\in I}B_i=X$, and
	\item
	for any $\omega\in B_i\cap B_j$, there exists $l\in I$ with $\omega\in B_l\subseteq B_i\cap B_j$ for any $i,j\in I$. 
\end{itemize}
A topology could be uniquely determined by its base. Indeed, if $\mathcal{B}$ is a base for the topology $\tau$, then $\tau$ is the collection of arbitrary unions of the elements of $\mathcal{B}$. Conversely, any collection $\mathcal{B}$ of subsets of a set $\Omega$ satisfying the above conditions uniquely defines a topology on $\Omega$. 

Let $\Omega$ be a set. Then, $\tau_1=\{\emptyset, \Omega\}$ and $\tau_2=\mathcal{P}(\Omega)$ are the smallest and the largest topologies defined on $\Omega$, respectively. The standard topology on $\mathbb{R}$ is the collection of arbitrary unions of open intervals. 
\subsection{Continuous Functions}
Let $(\Omega,\tau)$ and $(\Omega',\tau')$ be two topological spaces. A function $f:(\Omega,\tau)\to(\Omega',\tau')$ is called continuous, whenever $f^{-1}(U)\in\tau$ for any $U\in\tau'$. One could see that $f:(\Omega,\tau)\to(\Omega',\tau')$ is continuous if and only if $f^{-1}(B')\in\tau$ for any $B'\in\mathcal{B}'$, where $\mathcal{B}'$ is a base for $\tau'$. 

\subsection{Product Topology}

Let $(\Omega, \tau)$ and $(\Omega',\tau')$ be two topological spaces. Then, the product topology on the Cartesian product $\Omega\times \Omega'$ is the topology whose base is $\{U\times U':U\in\tau, U'\in\tau'\}$. Note that if $\mathcal{B}$ and $\mathcal{B}'$ are bases of $\tau$ and $\tau'$, respectively, then $\{B\times B':B\in\mathcal{B}, B'\in\mathcal{B}'\}$ is the base for the product topology on $\Omega\times \Omega'$. Now, inductively, one could define the product topology on the Cartesian product $\Omega_1\times\cdots\times\Omega_n$ of topological spaces $(\Omega_i,\tau_i)$ for any $i=1\ldots,n$. For instance, the product topology of the standard topology on $\mathbb{R}$ with itself on $\mathbb{R}^2$ has the following base:
\[ \{I\times J: I\text{ and }J\text{ are  open intervals in }\mathbb{R}\}.\]
In general, the following set (i.e., the set of all hypercubes in $\mathbb{R}^n$) is a base for the standard topology on $\mathbb{R}^n$:
\[\{I_1\times\cdots\times I_n:I_i\text{'s  are  open intervals in }\mathbb{R}^n\}.\]
In the sequel, we consider the standard topology on $\mathbb{R}^n$.
Note that it could be shown that any open set in $\mathbb{R}^n$ could be written as a countable union of disjoint hypercubes. 
\subsection{Extended Real Line}\label{extendedrealline}

Let $\overline{\mathbb{R}}=\mathbb{R}\cup\{\pm\infty\}$. We define the following arithmetic:
\begin{align*}
	&a\times \infty=\infty\times a=-((-a)\times\infty)=-(\infty\times(-a))=\infty,\\
	& b+\infty=\infty+\infty=\infty\times\infty=\infty,\qquad \frac{b}{\infty}=0.
	\end{align*}
for any $a,b\in\mathbb{R}$ with $a>0$. Indeed, there are similar identities for $-\infty$. Note that we leave $\pm(\infty-\infty)$, $0\times (\pm\infty)$ and $(\pm\infty)\times 0$ undefined. One could see that the following set is a base for a topology on $\overline{\mathbb{R}}$ called standard base:
\[\mathcal{B}=\{(a,\infty]:a\in\mathbb{R}\}\cup\{[-\infty,b):a\in\mathbb{R}\}\cup\{(a,b):a,b\in\mathbb{R},\;a<b\},\]
which is called the standard topology on $\overline{\mathbb{R}}$. We define $\overline{\mathbb{R}^n}:=\overline{\mathbb{R}}^n$ for any positive integer $n$, and we equip it by the product topology. It is worth mentioning that each open set in $\overline{\mathbb{R}}^n$ could be written as a countable union of disjoint elements of the form of $E_1\times\cdots\times E_n$, where $E_i\in\mathcal{B}$ for any $1\le i\le n$.

\vspace*{.4cm}
Note that several important concepts such as connected spaces, compact spaces, local properties, quotient spaces and several other concepts are discussed in Topology as a branch of mathematics. To study more in General Topology, readers are referred to the following textbooks \cite{armstrong2013basic}, \cite{munkres2000topology} and \cite{prasolov1998intuitive}.

\section{Measure Theory}\label{MeasureTheory}
In this appendix, we provide some definitions and theorems required for explaining the measure theoretical point of view of Probability Theory (see Appendix \ref{Probability Theory (Measure Theoretical Point of View)}). To see more details, we refer readers to study a standard textbook such as \cite{bogachev2007measure} and \cite{folland1999real}.
\subsection{$\sigma$-Algebras}

Let $\Omega$ be a set. A $\sigma$-algebra on $\Omega$ is a collection $\mathscr{F}$ of subsets of $\Omega$ which satisfies the following properties:
\begin{enumerate}
	\item 
	$\Omega\in\mathscr{F}$, 
	\item
	$E\in\mathscr{F}$ implies that $\Omega\backslash\{E\}\in\mathscr{F}$ for any $E\in\mathscr{F}$, and 
	\item
	if $\{E_i\}_{i=1}^{\infty}$ is a family of the elements of $\mathscr{F}$ then $\bigcup_{i=1 }^{\infty}E_i\in\mathscr{F}$.
\end{enumerate}
Let $\Omega$ be a set and $\mathscr{F}$ be a $\sigma$-algebra on $\Omega$. It follows from these properties that $\emptyset$ and the intersection of any countable family of elements of a $\sigma$-algebra on $\Omega$ belong to the $\sigma$-algebra as well. We sometimes refer to elements of a $\sigma$-algebra on $\Omega$ as measurable subsets of $\Omega$.

For any set $\Omega$, the collection $\{\emptyset, \Omega\}$ and the powerset of $\Omega$ are the smallest and largest $\sigma$-algebras defined on $\Omega$. The $\sigma$-algebra generated by a collection $\mathcal{S}$ of subsets of $\Omega$ is defined to be the intersection of all $\sigma$-algebras on $\Omega$ containing $\mathcal{S}$. Note that the $\sigma$-algebra generated by all open subsets of $\mathbb{R}^n$ is called the Borel $\sigma$-algebra and is denoted by $\mathscr{B}(\mathbb{R}^n)$. The $\sigma$-algebra generated by each of the following sets is $\mathscr{B}(\mathbb{R})$:
\begin{align*}
	\mathcal{B}_1&=\{(a,b):a,b\in\mathbb{R},\;a<b\},&\mathcal{B}_2&=\{(a,b]:a,b\in\mathbb{R},\;a<b\},\\
	\mathcal{B}_3&=\{[a,b):a,b\in\mathbb{R},\;a<b\},&
	\mathcal{B}_4&=\{[a,b]:a,b\in\mathbb{R},\;a\le b\},\\
	\mathcal{B}_5&=\{(a,\infty):a\in\mathbb{R}\}\cup\{(-\infty,b]:b\in\mathbb{R}\},&&\\
	\mathcal{B}_6&=\{[a,\infty):a\in\mathbb{R}\}\cup\{(-\infty,b]:b\in\mathbb{R}\},&&\\
	\mathcal{B}_7&=\{(a,\infty):a\in\mathbb{R}\}\cup\{(-\infty,b):b\in\mathbb{R}\},&&\\
	\mathcal{B}_8&=\{[a,\infty):a\in\mathbb{R}\}\cup\{(-\infty,b):b\in\mathbb{R}\}.&&
\end{align*}

 Indeed, the Borel $\sigma$-algebra is defined for any topological space. Especially, for any integer $n$, we have the Borel $\sigma$-algebra on $\overline{\mathbb{R}}^n$, denoted by $\mathscr{B}(\overline{\mathbb{R}}^n)$(see Appendix~\ref{1} for the definition of $\overline{\mathbb{R}}^n$ and the the standard topology on it).  

\subsection{Measurable Spaces and Measurable Functions}

For a set $\Omega$  and a $\sigma$-algebra $\mathscr{F}$  on $\Omega$, the pair $(\Omega,\mathscr{F})$ is called a measurable space. Let $(\Omega,\mathscr{F})$ and $(\Omega',\mathscr{F}')$ be two measurable spaces. A function $f:(\Omega,\mathscr{F})\to(\Omega',\mathscr{F}')$ is called measurable if $f^{-1}(E')\in\mathscr{F}$ for any $E'\in\mathscr{F}'$. If $\Omega$ and $\Omega'$ are both topological spaces and $\mathscr{F}$ and $\mathscr{F}'$ are the Borel $\sigma$-algebras on $\Omega$ and $\Omega'$, respectively, then any continuous function from $(\Omega,\mathscr{F})$ to $(\Omega',\mathscr{F}')$ is measurable (but the converse is wrong). One could see that the composition of two measurable functions is also measurable. As a convention, by a real valued measurable function $f:(\Omega,\mathscr{F})\to \overline{\mathbb{R}}$, we mean $f:(\Omega,\mathscr{F})\to(\overline{\mathbb{R}},\mathscr{B}(\overline{\mathbb{R}}))$ is measurable. We also consider the similar convention for multi-valued functions into $\overline{\mathbb{R}}^n$.

\subsection{Indicator Functions and Simple Functions}
Let $(\Omega,\mathscr{F})$ be a measurable space and $E\in\mathscr{F}$. Then, the indicator function of $E$, denoted by $\mathbbm{1}_E$, is defined as follows:
\[\mathbbm{1}_E:\Omega\to\mathbb{R},\qquad \mathbbm{1}_E(\omega)=\left\{\begin{array}{ll}1&,\omega\in E\\0&,\omega\notin E\end{array}\right..\]
Let $E_1,\ldots, E_n\in\mathscr{F}$ be pairwise disjoint and $r_1,\ldots,r_n\in\mathbb{R}$. Then, $\sum_{i=1}^nr_i\mathbbm{1}_{E_i}$ is called a simple function. One could see that the sum, product and composition of two simple functions is a simple function as well. Note that a simple function (and hence an indicator function) is measurable.

A famous result states that any non-negative measurable function $g:\Omega\to\overline{\mathbb{R}}$ is the point-wise limit of a sequence of simple functions defined on $\Omega$. This result plays a key role in the sequel. 
\subsection{Measures}

Let $(\Omega,\mathscr{F})$ be a measurable space. A function $\mu:\mathscr{F}\to[0,\infty]$ is called a measure, whenever it satisfies the following conditions:
\begin{enumerate}
	\item 
	$\mu(\emptyset)=0$, and
	\item
	$\mu\left(\bigcup_{i=1}^{\infty}E_i\right)=\sum_{i=1}^{\infty}\mu(E_i)$ for any countable family $\{E_i\}_{i=1}^{\infty}$ of pairwise disjoint elements $\mathscr{F}$. 
\end{enumerate}
One could see that a measure $\mu$ defined on a measurable space $(\Omega,\mathscr{F})$ is increasing (i.e., $\mu(E)\le \mu(E')$ if $E\subseteq E'$ for any $E,E'\in\mathscr{F}$).
If $\mu$ is a measure defined on the measurable space $(\Omega,\mathscr{F})$, then $(\Omega,\mathscr{F},\mu)$ is called a measure space. 

As an example, assume that $\mathscr{F}$ is the powerset of $\Omega$ and $\omega$ is an arbitrary element of $\Omega$. Then,  the Dirac measure with respect to $\omega$ is defined as follows:
\[\delta_{\omega}(E)=\left\{\begin{array}{ll}1 &,\omega\in E\\ 0&,\omega\notin E\end{array}\right.\]
for any subset $E$ of $\Omega$. 

\subsection{Continuity Property of Measures}

Let $(\Omega,\mathscr{F},\mu)$ be a measure space and $\{E_n\}_{n=1}^{\infty}$ be a family of elements of $\mathscr{F}$ with $E_1\subseteq E_2\subseteq\cdots$. Let us assume that $E=\bigcup_{n=1}^{\infty}E_n$ (i.e., $E_n\uparrow E$). Then, obviously $E$ is measurable. Now, we show that $\mu(E)=\lim_{N\to\infty}\mu(E_N)$. To see this,   assume that
\[B_1=E_1,\qquad B_n=E_n\backslash E_{n-1},\quad n=2,3,\ldots.\]
Then, $\{B_n\}_{n=1}^{\infty}$ is a family of pairwise disjoint measurable sets with $E=\bigcup_{n=1}^{\infty}B_n$. Thus, we have that
\begin{align*}
	\mu(E)=\mu\left(\bigcup_{n=1}^{\infty}B_n\right)=\sum_{n=1}^{\infty}\mu(B_n)=\lim_{N\to\infty}\sum_{n=1}^{N}\mu(B_n)=\lim_{N\to\infty}\mu\left(\bigcup_{n=1}^NB_n\right)=\lim_{N\to\infty}\mu(E_N).
\end{align*}
Similarly, one could see that if $E_1\supseteq E_2\supseteq\cdots $,  $E=\bigcap_{n=1}^{\infty} E_n$ (i.e., $E_n\downarrow E$), and $\mu(E_1)<\infty$, then $\mu(E)=\lim_{N\to\infty}\mu(E_N)$. 
\subsection{Almost Everywhere Properties}

Let $(\Omega,\mathscr{F},\mu)$ be a measure space and $\mathcal{P}$ be a property of elements in $\Omega$. It is said that $\mathcal{P}$ holds  almost everywhere in $\Omega$ if there exists $E\in\mathscr{F}$ with $\mu(E)=0$, and $\mathcal{P}$ holds for each $\omega\in\Omega\backslash\{E\}$. For instance, assume that the measure of any singleton is zero. If a property is satisfied for all elements of $\Omega$ except for a countable subset of $\Omega$, since the measure of any countable set is zero, then that property holds almost everywhere in $\Omega$. In the sequel, any equality or property involving measures holds almost everywhere but we might no longer mention it. For instance, when we say that "there is a unique function satisfying the property $\mathcal{P}$", we mean that if two functions satisfy $\mathcal{P}$, then they are almost everywhere equal.

\subsection{Lebesgue Measure}

Any measure defined on a Borel $\sigma$-algebra is called a Borel measure. A theorem guaranties that there is a unique Borel measure $\lambda$ on $\overline{\mathbb{R}}$ in such a way that $\lambda([a,b])=b-a$ for any real numbers $a$ and $b$ with $a\le b$. Similarly, there is a unique Borel measure $\lambda$ on $\overline{\mathbb{R}}^n$ such that $\lambda(H)=\mathrm{Vol}(H)$, where $H$ is a hypercube (i.e., a Cartesian product of $n$ intervals in $\overline{\mathbb{R}}$) and $\mathrm{Vol}(H)$ is the volume of $H$.  The aforementioned Borel measure on $\overline{\mathbb{R}}^n$ is called the Lebesgue measure.  For instance,  hypercubes in $\overline{\mathbb{R}}^2$ and $\overline{\mathbb{R}}^3$ are squares and cubes, respectively. Hence, intuitively the Lebesgue measures in $2$ and $3$ dimensions  represent  the area and the ordinary volume, respectively.

\subsection{Integrals}

Let $(\Omega,\mathscr{F},\mu)$ be a measure space and $s=\sum_{i=1}^nr_i\mathbbm{1}_{E_i}$ be a non-negative simple function.  The integral of $s$ on $\Omega$ is defined as follows:
\[\int_{\Omega}s\,\mathrm{d}\mu=\sum_{i=1}^nr_i\mu(E_i).\]
Now, assume that $f:\Omega\to[0,\infty]$ is a measurable function. Then, the integral of $f$ on $\Omega$ is defined as below:
\[\int_{\Omega}f\,\mathrm{d}\mu=\sup\left\{\int_{\Omega}s\,\mathrm{d}\mu: s\text{ is a simple function with } 0\le s\le f\right\}.\]
If $E\in\mathscr{F}$, then the integral of $f$ on $E$ is defined as follows: 
\[\int_Ef\,\mathrm{d}\mu=\int_{\Omega}f\mathbbm{1}_E\,\mathrm{d}\mu.\]
Note that if $f$ is only defined on $E\in\mathscr{F}$, then we can extend it to a measurable function $\tilde{f}:\Omega\to\mathbb{R}$ by the following setting:
\[\tilde{f}(x)=\left\{\begin{array}{ll}f(x)&,x\in E\\0&,x\notin E\end{array}\right.,\]
and we define
\[\int_Ef\,\mathrm{d}\mu=\int_E\tilde{f}\,\mathrm{d}\mu.\]
Now, let $f:\Omega\to\mathbb{R}$ be a measurable function. Assume that $f^+=\max\{f,0\}$ and $f^-=\max\{-f,0\}$. Both $f^+$ and $f^-$ are non-negative and measurable, and we have that $f=f^+-f^-$ and $|f|=f^++f^-$. We say that $f$ is integrable, if $\int_{\Omega}|f|\,\mathrm{d}\mu<\infty$. Now, for an integrable function $f:\Omega\to\mathbb{R}$, we define
\[\int_{E}f\,\mathrm{d}\mu=\int_Ef^+\,\mathrm{d}\mu-\int_Ef^-\,\mathrm{d}\mu.\]
To be more precise, sometimes we use the notation $\int_Ef(\omega)\,\mathrm{d}\mu(\omega)$ to clarify the variable we are taking the integral with respect to.

\subsection{Integrals of Multi-Values Functions}

Let $(\Omega,\mathscr{F},\mu)$ be a measure space and $f:\Omega\to\overline{\mathbb{R}}^n$ be non-negative (in each component) and measurable. Assume that $f(\omega)=(f_1(\omega),\ldots,f_n(\omega))$ for any $\omega\in\Omega$. For any $1\le i\le n$, note that $\pi_i:\overline{\mathbb{R}}^n\to\overline{\mathbb{R}}$ sending each $\omega\in\overline{\mathbb{R}}^n$ to its $i^{\text{th}}$ component is measurable, which implies that $f_i=\pi_i\circ f$ is measurable as well. In general,  $f$ is measurable if and only if $f_i$ is measurable for any $1\le i\le n$. Now, we define the integral of $f$ as follows:
\[\int_Ef\,\mathrm{d}\mu:=\left(\int_Ef_1\,\mathrm{d}\mu,\ldots,\int_Ef_n\,\mathrm{d}\mu\right).\]
Similarly,  an arbitrary function $f:\Omega\to\overline{\mathbb{R}}^n$ (not necessarily non-negative) is integrable if and only if $f_i$ is intgrable for any $1\le i\le n$. Thus, the above definition is considered for the integrable function $f$ as well.  
 Hence, many properties of the integration of the single valued functions (provided in the sequel) hold for multi-values functions as well.
\subsection{Some Properties of Integrals}
Let $(\Omega,\mathscr{F},\mu)$ be a measure space, and let $f,g:\Omega\to\overline{\mathbb{R}}$ be non-negative and measurable (or in general case, one could consider the integrable functions).
If $\{E_i\}_{i=1}^{\infty}$ is a family of pairwise disjoint measurable subsets of $\Omega$, then we have that
\[\int_{\bigcup_{i=1}^{\infty}E_i}f\,\mathrm{d}\mu=\sum_{i=1}^{\infty}\int_{E_i}f\,\mathrm{d}\mu.\]
Further, if $\lambda,\eta\in\mathbb{R}$, then for any measurable subset $E$ of $\Omega$, we have that 
\[\int_E(\lambda f+\eta g)\,\mathrm{d}\mu=\lambda \int_Ef\,\mathrm{d}\mu+\eta\int_Eg\,\mathrm{d}\mu.\]
Furthermore, if $\{f_n\}_{n=1}^{\infty}$ is a sequence of non-negative measurable real valued functions defined on $(\Omega,\mathscr{F})$, then we have that
\[\sum_{n=1}^{\infty}\int_Ef_n(\omega)\,\mathrm{d}\mu(\omega)=\int_{E}\sum_{n=1}^{\infty}f_n(\omega)\,\mathrm{d}\mu(\omega)\]
for any $E\in\mathscr{F}$.

\subsection{Measures Defined via Measures}

Let $(\Omega,\mathscr{F},\mu)$ be a measure space and $f:(\Omega,\mathscr{F})\to[0,\infty]$ be a non-negative measurable function. Define $\nu:\mathscr{F}\to[0,\infty]$  as follows:
\[\nu(E)=\int_Ef\,\mathrm{d}\mu\]
for any $E\in\mathscr{F}$. Then, $\nu$ is a measure on $(\Omega,\mathscr{F})$. This method to define new measures is common to define probability measures.

\subsection{Monotone Convergence Theorem}
Let $(\Omega,\mathscr{F}, \mu)$ be a measure space and $f:(\Omega,\mathscr{F})\to\overline{\mathbb{R}}$ be measurable and non-negative. If $\{s_n\}_{n=1}^{\infty}$ is a sequence of non-negative simple functions defined on $(\Omega,\mathscr{F})$ which converges to $f$ almost everywhere, then 
\[\lim_{n\to\infty}\int_{E}s_n(\omega)\,\mathrm{d}\mu(\omega)=\int_Ef(\omega)\,\mathrm{d}\mu(\omega)\]
for any $E\in\mathscr{F}$.

\subsection{Dominating Convergence Theorem}

Let $(\Omega,\mathscr{F},\mu)$ be a measure space, and $\{f_n\}_{n=1}^{\infty}$ be a sequence of non-negative measurable functions from $(\Omega,\mathscr{F})$ to $\overline{\mathbb{R}}^n$. If $\{f_n\}_{n=1}^{\infty}$ converges point-wise to $f$, and there exists an integrable function $g:(\Omega,\mathscr{F})\to\overline{\mathbb{R}}^n$ in such a way that 
$|f_n(\omega)|\le g(\omega)$ for any $\omega\in\Omega$, then we have that
\[\lim_{n\to\infty}\int_{E}f_n(\omega)\,\mathrm{d}\mu(\omega)=\int_E f(\omega)\,\mathrm{d}\mu(\omega)\]
for any $E\in\mathscr{F}$.

\subsection{Equality of Finite Borel Measures}

Let $\Omega$ be a topological space, and let $\mu_1$ and $\mu_2$ be two finite Borel measures on $\Omega$. Then, by \cite[Lemma 7.1.2]{bogachev2007measure}, if $\mu_1$ and $\mu_2$ are equal on all open sets in $\Omega$, then they are equal on all Borel sets in $\Omega$. Now, let's consider the case $\Omega=\mathbb{R}^n$. We know that any open set in $\mathbb{R}^n$ is a countable union of disjoint hypercubes, which implies that if two Borel measures on $\mathbb{R}^n$ are equal on all hypercubes, then they are equal on all open sets and hence equal on all Borel sets in $\Omega$. Similarly, if two Borel measures are equal on a base  $\mathcal{B}$ for the standard topology on $\overline{\mathbb{R}}^n$ (discussed in \ref{extendedrealline}), then they are equal on all Borel sets in $\overline{\mathbb{R}}^n$ .

\subsection{Computing Lebesgue Integrals via Riemann Integrals}
Consider the measure space $(\mathbb{R},\mathscr{B}(\mathbb{R}),\lambda)$, where $\lambda$ is the Lebesgue measure. Then, one could see that any Riemann integrable function $f:[a,b]\to\mathbb{R}$ is Lebesgue integrable, and we have that
\[\int_{[a,b]}f\,\mathrm{d}\mu=\int_a^bf\,\mathrm{d}x.\]
Also, if $f:\mathbb{R}\to\mathbb{R}$ is a non-negative Riemann integrable function (in the sense of improper integrability), then $f$ is Lebesgue integrable, and we have that 
\[\int_{\mathbb{R}}f\,\mathrm{d}\mu=\int_{-\infty}^{\infty}f\,\mathrm{d}x.\]
A similar statement holds when we consider any interval like $[a,\infty)$ and $(-\infty,b]$ for any real numbers $a$ and $b$. 

\subsection{Radon-Nikodym Theorem}

Let $(\Omega,\mathscr{F},\mu)$ be a measure space. Then, $\mu$ is called a finite measure, if $\mu(\Omega)<\infty$ (and hence $\mu(E)<\infty$ for any $E\in\mathscr{F}$). Further, $\mu$ is called a $\sigma$-finite measure if there is a countable family $\{E_i\}_{i=1}^{\infty}$ of measurable subsets of $\Omega$ with finite measure (i.e., $\mu(E_i)<\infty$ for any $i$) such that $\Omega=\bigcup_{i=1}^{\infty}E_i$. Obviously, a finite measure is a $\sigma$-finite measure. 

Now, assume that $(\Omega,\mathscr{F})$ is a measurable space, and $\mu$ and $\nu$ are two measures defined on this space. Then, we say that $\nu$ is absolutely continuous with respect to $\mu$, if $\mu(E)=0$ implies that $\nu(E)=0$ for any $E\in\mathscr{F}$. 

Now, we are ready to state the Radon-Nikodym theorem. Let $(\Omega,\mathscr{F})$ be a measurable space, and $\mu$ and $\nu$ be two $\sigma$-finite measures on this space. If $\nu$ is absolutely continuous with respect to $\mu$, then there exists a non-negative measurable function $f:\Omega\to[0,\infty]$ in such a way that
\[\nu(E)=\int_Ef\,\mathrm{d}\mu\]
for any $E\in\mathscr{F}$.

\subsection{Changing the Measure of Integration}

Let $(\Omega,\mathscr{F})$ be a measurable space. Also, let $\mu$ and $\nu$  be two measures on $(\Omega,\mathscr{F})$ such that $\nu$ is absolutely continuous with respect to $\mu$. Then, by Radon-Nikodym theorem, there exists a  measurable function $f:\Omega\to[0,\infty]$ satisfying the following:
\[\nu(E)=\int_Ef\,\mathrm{d}\mu\]
for any $E\in\mathscr{F}$. Now, if $g:\Omega\to[0,\infty]$ is a  measurable function, then for any $E\in\mathscr{F}$, we have that
\[\int_E g\,\mathrm{d}\nu=\int_Egf\,\mathrm{d}\mu.\]
To prove it, assume that $E,E'\in\mathscr{F}$. If $g=\mathbbm{1}_{E'}$, then 
\[\int_E\mathbbm{1}_{E'}\,\mathrm{d}\nu=\nu(E\cap E')=\int_{E\cap E'}f\,\mathrm{d}\mu=\int_E\mathbbm{1}_{E'}f\,\mathrm{d}\mu.\]
Hence, the desired identity holds for all indicator functions $g$. By the linearity of integrals, this identity holds for all simple functions $g$ as well. Now, by the monotone convergence theorem, we obtain the identity for all non-negative measurable functions $g$.

\subsection{Change of Variable Theorem}

Let $(\Omega,\mathscr{F},\mu)$ and $(\Omega',\mathscr{F}')$ be a measure space and a measurable space, respectively. Let $F:(\Omega,\mathscr{F})\to(\Omega',\mathscr{F}')$ be a measurable function. Then, $F$ induces a measure on $(\Omega',\mathscr{F}')$ as follows:
\[\mu_F:\mathscr{F}'\to[0,\infty],\quad \mu_F(E')=\mu(F^{-1}(E'))\]
for any $E'\in\mathscr{F}'$. Now, if $g:(\Omega',\mathscr{F}')\to[0,\infty]$ is  measurable, then we have that
\[\int_{\Omega}g\circ F(\omega)\,\mathrm{d}\mu(\omega)=\int_{\Omega'}g(\omega')\,\mathrm{d}\mu_F(\omega').\]

Naturally, the above theorem holds for integrable functions $g:(\Omega',\mathscr{F}')\to\overline{\mathbb{R}}^n$.
\section{Probability Theory (Measure Theoretic Point of View)}\label{Probability Theory (Measure Theoretical Point of View)}
In this appendix,  we provide a very brief description of the basics of the measure theoretical  Probability Theory. To deal with conditional probabilities, we provide two different points of view, and explain how these two points of view are consistent with each other. To study more in the measure theoretical Probability Theory, we recommend the following textbooks \cite{ash2000probability}, \cite{capinski2004measure} and \cite{ccinlar2011probability}.
 
\subsection{Probability Space}

A probability space is a measure space $(\Omega,\mathscr{F},P)$ such that $P(\Omega)=1$. If $(\Omega,\mathscr{F},P)$ is a probability space, then $\Omega$, $\mathscr{F}$ and $P$ are called the sample space, the event space and the probability measure, respectively. Further, each element of the event space is called an event. 

\subsection{Discrete Probability Spaces}

Let $(\Omega,\mathscr{F},P)$ be a probability space. If $\Omega$ is countable and $\mathscr{F}$ is the powerset of $\Omega$, then $(\Omega,\mathscr{F},P)$ is called a discrete probability space. Note that for any $E\in\mathscr{F}$, we have that
\[P(E)=P\left(\bigcup_{e\in E}\{e\}\right)=\sum_{e\in E}P(\{e\}).\]

\subsection{Random Variables and Random Vectors}

Let $(\Omega,\mathscr{F},P)$ be a probability space. A measurable function $X:(\Omega,\mathscr{F})\to\overline{\mathbb{R}}^n$ is called a random vector. A  random variable is obtained when in the above definition we assume that $n=1$. Let $X:(\Omega,\mathscr{F})\to\overline{\mathbb{R}}^n$ be a random vector.  We note that $X$ induces a probability measure $P_X$ on $(\overline{\mathbb{R}}^n,\mathscr{B}(\overline{\mathbb{R}}^n))$ as follows:
\[P_X(E')=P(X^{-1}(E'))\]
for any $E'\in\mathscr{B}(\overline{\mathbb{R}}^n)$. We define the support of $X$, denoted by $\mathrm{Supp}(X)$, to be the smallest closed set $V$ with $P_X(V)=1$. Especially, if $X:(\Omega,\mathscr{F})\to\overline{\mathbb{R}}$ is a random variable  which takes a finitely number of values (i.e., a simple random variable), we have that
\[\mathrm{Supp}(X)=\{r\in\overline{\mathbb{R}}:P_X(r)>0\}.\]
 If the set of values of a random vector $X:(\Omega,\mathscr{F})\to\overline{\mathbb{R}}^n$ is countable, then $X$ is called a discrete random vector. We note that the support of a discrete random variable can be obtained similar to the support of a simple random variable.  
\subsection{Joint Probability Distributions}

Let $(\Omega,\mathscr{F},P)$ be a probability space and $X:(\Omega,\mathscr{F})\to(\overline{\mathbb{R}}^n,\mathscr{B}(\overline{\mathbb{R}}^n))$ and $Y:(\Omega,\mathscr{F})\to(\overline{\mathbb{R}}^m,\mathscr{B}(\overline{\mathbb{R}}^m))$ be two random vectors, respectively. Also, let $\mathscr{B}(\overline{\mathbb{R}}^n)\otimes\mathscr{B}(\overline{\mathbb{R}}^m)$ be the $\sigma$-algebra generated by the set $\mathscr{B}(\overline{\mathbb{R}}^n)\times\mathscr{B}(\overline{\mathbb{R}}^m)$ on $\overline{\mathbb{R}}^n\times\overline{\mathbb{R}}^m$. One could see that $\overline{\mathbb{R}}^n\times\overline{\mathbb{R}}^m=\overline{\mathbb{R}}^{n+m}$ and $\mathscr{B}(\overline{\mathbb{R}}^n)\otimes\mathscr{B}(\overline{\mathbb{R}}^m)=\mathscr{B}(\overline{\mathbb{R}}^{n+m})$. Then, $(X,Y):(\Omega,\mathscr{F})\to(\overline{\mathbb{R}}^{n+m},\mathscr{B}(\overline{\mathbb{R}}^{n+m}))$ is a random vector defined by setting $(X,Y)(\omega)=(X(\omega),Y(\omega))$ for any $\omega\in\Omega$. The joint probability distribution of $X$ and $Y$ is determined by the probability measure induced by $(X,Y)$ (i.e., $P_{X,Y}(E)=P(\{\omega\in\Omega:(X(\omega),Y(\omega))\in E\})$ for any $E\in\mathscr{B}(\overline{\mathbb{R}^{n+m}})$). Two Random vectors $X$ and $Y$ defined on $(\Omega,\mathscr{F})$ are said to be independent if $P_{X,Y}(E'\times E'')=P_X(E')P_Y(E'')$ for any $E'\in\mathscr{B}(\overline{\mathbb{R}^{n}})$ and $E''\in\mathscr{B}(\overline{\mathbb{R}^{m}})$. Note that
\begin{align*}
	P_X(E')=P(\{\omega\in\Omega: X(\omega)\in E'\})&=P(\{\omega\in\Omega: (X(\omega),Y(\omega))\in E'\times\overline{\mathbb{R}}^m\})\\
	&=P_{X,Y}(E'\times\overline{\mathbb{R}}^m).
\end{align*}
Similarly, we have that $P_Y(E'')=P_{X,Y}(\overline{\mathbb{R}}^n\times E'')$.
One could generalize the above procedure to define the joint distribution of any finite number of random variables defined on $(\Omega,\mathscr{F},P)$. Further, as we explained above, the joint distribution of any non-empty subset of a finite set of random vectors defined on $(\Omega,\mathscr{F},P)$ could be obtained from the joint distribution of all  random vectors.

\subsection{Conditional Probability (Conditions on the Values of a Random Variable)}\label{con}

Let $(\Omega,\mathscr{F},P)$ be a probability space, $E\in\mathscr{F}$, and $X:(\Omega,\mathscr{F})\to\overline{\mathbb{R}}^n$ be a random vector. We define the measure $\mu_E:\mathscr{B}(\overline{\mathbb{R}}^n)\to[0,\infty]$ as follows:
\[\mu_E(E')=P(X^{-1}(E')\cap E)\]
for any $E'\in\mathscr{B}(\overline{\mathbb{R}}^n)$. Note that $P_X(E')=0$ implies that $\mu_E(E')=0$ for any $E'\in\mathscr{B}(\overline{\mathbb{R}}^n)$. Hence, $\mu_E$ is absolutely continuous with respect to $P_X$. Thus, it follows from the Radon-Nikodym theorem that there exists a unique measurable function $f_E:\overline{\mathbb{R}}^n\to[0,\infty]$ such that
\[\mu_E(E')=\int_{E'}f_E(x)\,\mathrm{d}P_X(x)\]
for any $E'\in\mathscr{B}(\overline{\mathbb{R}}^n)$. It is a common notation to denote $f_E(x)$ by $P(E|X=x)$ for any possible or impossible value $X=x$. Therefore, we have that
\begin{equation}\label{conddef}
	P((X\in E')\cap E)=\int_{E'}P(E|X=x)\,\mathrm{d}P_X(x)
\end{equation}
for any $E\in\mathscr{F}$ and $E'\in\mathscr{B}(\overline{\mathbb{R}}^n)$. The function $P(\,\cdot\,|X=x):\mathscr{F}\to[0,\infty]$ is called a conditional probability function conditioned on $X=x$. As a result, we have the following equation for any $x_0\in\overline{\mathbb{R}}^n$:
\begin{equation}\label{conditionalformula}
P(E\cap(X=x_0))=\int_{\{x_0\}}P(E|X=x)\,\mathrm{d}P_X(x)=P(E|X=x_0)P_X(x_0).
\end{equation}

Let's  see an example as follows. Assume that $X$ is a discrete random variable and $E\in\mathscr{F}$. Then, it follows from Equation (\ref{conditionalformula}) that we must have
\[P(E|X=x)=
	\dfrac{P((X=x)\cap E)}{P(X=x)},\quad x\in\mathrm{Supp}(X).\]
 For any $x\in\overline{\mathbb{R}}\backslash\mathrm{Supp}(X)$, we can assign any value to $P(E|X=x)$ (for instance, one could define $P(E|X=x)=0$). Now, we show that $P(\,\cdot\,|X=x)$ with the above definition satisfies Equation (\ref{conddef}). To do so, let  $E'\in\mathscr{B}(\overline{\mathbb{R}}^n)$. Then, we have that
\begin{align*}
	\int_{E'}P(E|X=x)\,\mathrm{d}P_X(x)&=\sum_{x\in E'}P(E|X=x)P_X(x)=\sum_{x\in E'}P(E\cap(X=x))\\
&=P\left(E\cap\left(\bigcup_{x\in E'}\{x\}\right)\right)=P(E\cap (X\in E')).
\end{align*}
Hence, $P(\,\cdot\,|X=x)$ satisfies Equation (\ref{conddef}).
One could see that for any $x\in\mathrm{Supp}(X)$, $P(\,\cdot\,|X=x)$ is a probability measure as well. 

If $X_i:(\Omega,\mathscr{F})\to\overline{\mathbb{R}}^{n_i}$ is a random vector for any $1\le i\le l$, then we know that $(X_1,\ldots,X_l)$ is a random vector. Hence, we can define
\[P(E|X_1=x_1,\ldots,X_n=x_n)=P(E|(X_1,\ldots,X_n)=(x_1,\ldots,x_n)).\]
Note that the concept of being independent is naturally defined here. Indeed, let $X:(\Omega,\mathscr{F})\to\overline{\mathbb{R}}^{k_1}$, $Y:(\Omega, \mathscr{F})\to\overline{\mathbb{R}}^{k_2}$ and $Z:(\Omega,\mathscr{F})\to\overline{\mathbb{R}}^{k_3}$ be three random vectors defined on $(\Omega,\mathscr{F})$. Then, $X$ and $Y$ are said to be independent given $Z=z$, whenever 
\[P((X,Y)\in E_1\times E_2|Z=z)=P(X\in E_1|Z=z)P(Y\in E_2|Z=z)\]
for any Borel sets $E_1$ and $E_2$ in $\overline{\mathbb{R}}^{k_1}$ and $\overline{\mathbb{R}}^{k_2}$, respectively.
\subsection{Expected Values of Random Variables}

Let $(\Omega,\mathscr{F},P)$ be a probability space and $X:(\Omega,\mathscr{F})\to\overline{\mathbb{R}}$ be a random variable. Then, the expected value of $X$ is defined as follows:
\[\mathbb{E}(X)=\int_{\Omega}X(\omega)\,\mathrm{d}P.\]
Thus, if $X$ is a discrete random variable, and $x_1,\ldots,x_n,\ldots,$ are all distinct values of $X$, then we have that
\begin{align*}
	\mathbb{E}(X)&=\int_{\bigcup_{i=1}^{\infty}X^{-1}(x_i)}X(\omega)\,\mathrm{d}P=\sum_{i=1}^{\infty}\int_{X^{-1}(x_i)}X(\omega)\,\mathrm{d}P\\
	&=\sum_{i=1}^{\infty}x_i\int_{X^{-1}(x_i)}\,\mathrm{d}P=\sum_{i=1}^{\infty}x_iP(X^{-1}(x_i))=\sum_{i=1}^{\infty}x_iP_X(x_i).
\end{align*}

\subsection{Expected Values of Random Vectors}

Let $(\Omega,\mathscr{F},P)$ be a probability space and $X:(\Omega,\mathscr{F})\to\overline{\mathbb{R}}^n$ be a random vector. Then, the expected value of $X$ is defined as follows:
\[\mathbb{E}(X)=\left(\mathbb{E}(X_1),\ldots,\mathbb{E}(X_n)\right),\]
where we assume that $X(\omega)=(X_1(\omega),\ldots,X_n(\omega))$ for any $\omega\in\Omega$. Note that $\pi_i:\overline{\mathbb{R}}^n\to\overline{\mathbb{R}}$ is measurable, and hence it is a random variable for any $1\le i\le n$. It follows that $X_i=\pi_i\circ X$ is a random variable as well for any $1\le i\le n$.
\subsection{Conditional Expectation and Conditional Probability (Conditions on a Sub $\sigma$-Algebra)}\label{consigma}

Let $(\Omega,\mathscr{F},P)$ be a probability space and $\mathscr{G}$ be a sub $\sigma$-algebra of $\mathscr{F}$. Also, let $X:(\Omega,\mathscr{F})\to\overline{\mathbb{R}}^n$ be a random vector. The expectation of $X$ with respect to $\mathscr{G}$, denoted by $\mathbb{E}(X|\mathscr{G})$, is a random vector from $(\Omega,\mathscr{G})$ to $\overline{\mathbb{R}}^n$ in such a way that 
\[\int_{E}\mathbb{E}(X|\mathcal{G})(\omega)\,\mathrm{d}P(\omega)=\int_{E}X(\omega)\,\mathrm{d}P(\omega)\]
for any $E\in\mathscr{G}$. Indeed, $\mathbb{E}(X|\mathscr{G})$ exists and it is almost every where unique. To see this, define the measure $\nu:\mathscr{G}\to[0,1]$ by setting \[\nu(E)=\int_{E}X(\omega)\,\mathrm{d}P(\omega)\] for any $E\in\mathscr{G}$. Note that $P(E)=0$ implies that $\nu(E)=0$ for any $E\in\mathscr{G}$. It follows that $\nu$ is absolutely continuous with respect to $P_{|_{\mathscr{G}}}:\mathscr{G}\to[0,\infty]$. Hence, there exists a unique measurable function $g:(\Omega,\mathscr{G})\to\overline{\mathbb{R}}^n$ for which we have that
$\nu(E)=\int_{E}g(\omega)\,\mathrm{d}P(\omega)$ for any $E\in\mathscr{G}$. Finally, $g$ is denoted by $\mathbb{E}(X|\mathscr{G})$.

Now, let $E\in \mathscr{G}$. The conditional probability of $E$ given $\mathscr{G}$ is defined as follows:
\[\mathbb{P}(E|\mathcal{G})=\mathbb{E}(\mathbbm{1}_{E}|\mathcal{G}).\]
Hence, $\mathbb{P}(E|\mathcal{G})$ is a  random variable defined on $(\Omega,\mathscr{G})$ such that 
\[\int_{E'}\mathbb{P}(E|\mathcal{G})\,\mathrm{d}P=\int_{E'}\mathbbm{1}_E\,\mathrm{d}P=P(E\cap E')\]
for any $E'\in\mathscr{G}$. We also define a probability measure  $P(\,\cdot\,|E'):\mathscr{G}\to[0,1]$ in such a way that
\[P(E|E')P(E')=P(E\cap E')\]
for any $E'\in\mathscr{G}$. It follows that 
\[P(E|E')P(E')=\int_{E'}\mathbb{P}(E|\mathcal{G})\,\mathrm{d}P\]
for any $E'\in\mathscr{G}$. 

Considering $P(\,\cdot\,|E')$ (defined above), being independent  given an event for two random variables is defined similar to the definition given in Appendix \ref{con}.
\subsection{ The Relationship Between the  Conditional Probability Points of View introduced in Appendices \ref{con} and \ref{consigma} }
Let $(\Omega,\mathscr{F},P)$ be a probability space and $X:(\Omega,\mathscr{F})\to\overline{\mathbb{R}}^n$ be a random vector. Then, it is easy to see that $\mathscr{G}=\{X^{-1}(E''):E''\in\mathscr{B}(\overline{\mathbb{R}}^n)\}$ is a sub $\sigma$-algebra of $\mathscr{F}$. We show that
\[\mathbb{P}(E|\mathscr{G})(\omega)=P(E|X=X(\omega))\]
 for almost every $\omega\in\Omega$, where $P(E|X=X(\omega))$ and $\mathbb{P}(E|\mathscr{G})$ are defined with respect to the points of view introduced in Appendices \ref{con} and \ref{consigma}, respectively.
To do so, first we note that $P(E|X=X(\omega))=P(E|X=\,\cdot\,)\circ X$ as the composition of two measurable functions is measurable. Hence, it follows from the Change of Variable theorem that
\[\int_{X^{-1}(E'')}P(E|X=X(\omega))\,\mathrm{d}P(\omega)=\int_{E''}P(E|X=x)\,\mathrm{d}P_X(x)=P(E\cap X^{-1}(E'')).\]
One could see that when $E$ is a specific value of a random variable on $(\Omega,\mathscr{F})$, then two concepts of  independence discussed in Appendices \ref{con} and \ref{consigma} coincide with each other. 
\subsection{Conditional Probabilities of Conditional Probabilities}\label{conofcon}

Let $(\Omega,\mathscr{F},P)$ be a probability space, $E, E_0\in\mathscr{F}$ and $X:(\Omega,\mathscr{F})\to\overline{\mathbb{R}}^n$ be a random vector. Let $P^E:=P(\,\cdot\,|E)$ be a conditional probability measure, and $P_X^E$ be the conditional probability measure induced by $X$ on $\overline{\mathbb{R}}^n$ (i.e., $P_X^E(E')=P(X^{-1}(E')|E)$ for any $E'\in\mathscr{B}(\overline{\mathbb{R}}^n)$). Then, by the same argument discussed in Appendix \ref{con}, there exists a unique measurable function $f^E_{E_0}:\overline{\mathbb{R}}^n\to[0,\infty]$ that 
\[P^E((X\in E')\cap E_0)=\int_{E'}f_{E_0}^E(x)\,\mathrm{d}P^E_X(x)\]
for any $E'\in\mathscr{B}(\overline{\mathbb{R}}^n)$. Now, 
we denote $f_{E_0}^E(x)$ by $P(E_0|E,X=x)$ for any $x\in\overline{\mathbb{R}}^n$. It follows that
\[P((X\in E')\cap E_0|E)=\int_{E'}P(E_0|E,X=x)\,\mathrm{d}P^E_X(x)\]
for any $E'\in\mathscr{B}(\overline{\mathbb{R}}^n)$. 

\subsection{Total law of Probability for Conditional Probabilities}
Similar to the assumptions in Appendix \ref{conofcon}, let $(\Omega,\mathscr{F},P)$ be a probability space, $E, E_0\in\mathscr{F}$ and $X:(\Omega,\mathscr{F})\to\overline{\mathbb{R}}^n$ be a random vector. Also, assume that $P^E:=P(\,\cdot\,|E)$ be a conditional probability measure, and $P_X^E$ be the conditional probability measure induced by $X$ on $\overline{\mathbb{R}}^n$ with respect to $P^E$. Then, by Appendix \ref{conofcon}, we have that
\[P(E_0|E)=P((X\in\overline{\mathbb{R}}^n)\cap E_0|E)=\int_{\overline{\mathbb{R}}^n}P(E_0|E,X=x)\,\mathrm{d}P_X^E(x).\]

\subsection{Relationship Between Conditional Probabilities of Conditional Probabilities and Conditional Probabilities Given the Joint Distribution of two Random Vectors}
Once again, let us assume that $(\Omega,\mathscr{F},P)$ is a probability space, $ E_0\in\mathscr{F}$, and $X:(\Omega,\mathscr{F})\to\overline{\mathbb{R}}^n$ and $Y:(\Omega,\mathscr{F})\to\overline{\mathbb{R}}^m$ are two random vectors. Also, assume that $P^E:=P(\,\cdot\,|E)$ is a conditional probability measure, and $P_X^E$ is the conditional probability measure induced by $X$ on $\overline{\mathbb{R}}^n$ with respect to $P^E$.
Note that if  $E=Y^{-1}(y)$ for some value $y\in\mathrm{Supp}(Y)$, then $P(E_0|Y=y,X=x)$ (defined in Appendix \ref{con}) coincides with $f_{E_0}^E(x)=P(E_0|E,X=x)$  (defined in Appendix~\ref{conofcon}). To see this, let $E=Y^{-1}(y)$ with $y\in\mathrm{Supp}(Y)$. For simplicity, we denote $P_X^E$ and $f_{E_0}^E$ by $P_X^{y}$ and $f_{E_0}^y$, respectively. Let $P(E_0|X=\,\cdot\,\,,\,Y=\cdot\,):\overline{\mathbb{R}}^{n}\times \overline{\mathbb{R}}^{m}\to[0,1]$ be a conditional probability function in the sense of Appendix \ref{con}. Assume that $P(E_0|X=\,\cdot\,\,,\,Y=\cdot\,)$ is integrable. It is enough to show that 
\[P(E_0\cap(X\in E')|Y=y)=\int_{E'}P(E_0|X=x,Y=y)\,\mathrm{d}P_X^y(x)\]
for any $E'\in\mathscr{B}(\overline{\mathbb{R}}^n)$. To prove the above equality, it is enough to show that for any $E'\in\mathscr{B}(\overline{\mathbb{R}}^n)$ and $E''\in\mathscr{B}(\overline{\mathbb{R}}^m)$, we have that
\[P(E_0\cap((X,Y)\in E'\times E''))=\int_{E''}\int_{E'}P(E_0|X=x,Y=y)\,\mathrm{d}P_X^y(x)\mathrm{d}P_Y(y),\]
since $P(E_0\cap(X\in E')|Y=\,\cdot\,)$ is  a unique function satisfying the following equality for any $E'\in\mathscr{B}(\overline{\mathbb{R}}^n)$ and $E''\in\mathscr{B}(\overline{\mathbb{R}}^m)$ (see Appendic \ref{con}):
\[P(E_0\cap((X,Y)\in E'\times E''))=\int_{E''}P(E_0\cap(X\in E')|Y=y)\,\mathrm{d}P_Y(y).\]
To do so, let $E'\in\mathscr{B}(\overline{\mathbb{R}}^n)$ and $E''\in\mathscr{B}(\overline{\mathbb{R}}^m)$. Note that it follows from Appendix~\ref{con} that
\[P(E_0\cap((X,Y)\in E'\times E''))=\int_{E'\times E''}P(E_0|X=x,Y=y)\,\mathrm{d}P_{X,Y}(x,y).\]
Thus, it is enough to show that 
\[\int_{E'\times E''}f(x,y)\,\mathrm{d}P_{X,Y}(x,y)=\int_{E''}\int_{E'}f(x,y)\,\mathrm{d}P_X^y(x)\mathrm{d}P_Y(y),\]
where $f(x,y)=P(E_0|X=x,Y=y)$.
To prove this, first assume that $f=\mathbbm{1}_{E_1\times E_2}$ for some $E_1\in\mathscr{B}(\overline{\mathbb{R}}^n)$ and $E_2\in\mathscr{B}(\overline{\mathbb{R}}^ m)$. Then, we have that
\begin{align*}
	\int_{E'\times E''}\mathbbm{1}_{E_1\times E_2}(x,y)\,\mathrm{d}P_{X,Y}(x,y)&=\int_{\overline{\mathbb{R}}^n\times\overline{\mathbb{R}}^m}\mathbbm{1}_{(E_1\cap E')\times (E_2\cap E'')}(x,y)\,\mathrm{d}P_{X,Y}(x,y)\\
	&=P_{X,Y}((E_1\cap E')\times (E_2\cap E'')),\\
	\int_{E''}\int_{E'}\mathbbm{1}_{E_1\times E_2}(x,y)\,\mathrm{d}P_X^y(x)\mathrm{d}P_Y(y)&=\int_{E''\cap E_2}\int_{\overline{\mathbb{R}}^n}\mathbbm{1}_{E_1\cap E'\times E_2}(x,y)\,\mathrm{d}P_X^y(x)\mathrm{d}P_Y(y)\\
	&=\int_{E''\cap E_2}\int_{\overline{\mathbb{R}}^n}\mathbbm{1}_{E_1\cap E'}(x)\,\mathrm{d}P_X^y(x)\mathrm{d}P_Y(y)\\
	&=\int_{E''\cap E_2}P(X\in E_1\cap E'|Y=y)\,\mathrm{d}P_Y(y)\\
	&=P((X,Y)\in(E_1\cap E')\times (E_2\cap E''))\\
	&=P_{X,Y}((E_1\cap E')\times (E_2\cap E'')).
\end{align*}
Hence, the equality holds for indicator functions. By the linearity of integrals, one could see that the equality holds for simple functions as well. Finally, by the monotone convergence theorem and dominating convergence theorem, the equality holds for any non-negative function. 

\subsection{Expected Value Given an Event}

Let $(\Omega,\mathscr{F},P)$ be a probability space and $E\in\mathscr{F}$. Also, let $X:(\Omega,\mathscr{F})\to\overline{\mathbb{R}}^n$ be a random vector. 
Assume that $\{E_i\}_{i=1}^{\infty}$ is a family of pairwise disjoint elements of $\mathscr{F}$ with $\Omega=\bigcup_{i=1}^{\infty}E_i$. Also, assume that $\mathscr{G}$ is the sub $\sigma$-algebra of $\mathscr{F}$ generated by $\{E_i\}_{i=1}^{\infty}$. One could see that 
\[\mathscr{G}=\left\{\bigcup_{i \in I}E_i:I\subseteq\mathbb{N}\right\}.\]
Thus, $\mathbb{E}(X|\mathscr{G})$ is constant on $E_i$, since $\mathbb{E}(X|\mathscr{G})^{-1}(r)$ must be in $\mathscr{G}$ for any $i\in\mathbb{N}$ and $r\in\overline{\mathbb{R}}$. It follows that 
\[\mathbb{E}(X|\mathscr{G})=\sum_{i=1}^{\infty}r_i\mathbbm{1}_{E_i}\]
for some real numbers $r_i$. In the aforementioned expression,  $r_i$ is denoted by $\mathbb{E}(X|E_i)$ for any $i\in\mathbb{N}$. It follows that
\[\mathbb{E}(X|E_i)P(E_i)=\int_{E_i}\mathbb{E}(X|\mathscr{G})\,\mathrm{d}P(\omega)=\int_{E_i}X(\omega)\,\mathrm{d}P(\omega)\]
for any $i$.
 Hence, we showed that for any $E\in\mathscr{F}$,  $\mathbb{E}(X|E)P(E)$ is well-defined (i.e., it does not depend on a family of pairwise disjoint elements of $\mathscr{F}$ which includes $E$). Now, let $E\in\mathscr{F}$ and $P^E=P(\,\cdot\,|E)$ be a conditional probability measure. Then, we claim that the following definition for $\mathbb{E}(X|E)$ is acceptable:
 \[\mathbb{E}(X|E):=\int_{\Omega}X(\omega)\,\mathrm{d}P^E(\omega).\]
 To show this, it is enough to show that
 \[\left(\int_{\Omega}X(\omega)\,\mathrm{d}P^E(\omega)\right)P(E)=\int_{\Omega}X(\omega)\,\mathrm{d}P(\omega).\]
Let $E'\in\mathscr{F}$ be arbitrary and $X=\mathbbm{1}_{E'}$. Then, we have that
\begin{align*}
	\left(\int_{\Omega}\mathbbm{1}_{E'}(\omega)\,\mathrm{d}P^E(\omega)\right)P(E)&=P^E( E')P(E)=P(E'|E)P(E)\\
	&=P(E'\cap E)=\int_E\mathbbm{1}_{E'}(\omega)\,\mathrm{d}P(\omega).
\end{align*}
One could see that the equality holds when $X$ is a simple function as well. Finally, by the monotone convergence theorem, the equality holds for any random vector $X:(\Omega,\mathscr{F})\to\overline{\mathbb{R}}^n$.

\subsection{Continuous Random Variables }\label{continuousrandomvariables}

In this subsection, we briefly explain continuous random variables, their conditional distributions, expectations, and  conditional expectations. In the sequel, we assume that $(\Omega,\mathscr{F},P)$ is a probability space, $X,Y:(\Omega,\mathscr{F})\to\overline{\mathbb{R}}$ are  random variables, and $Z:(\Omega,\mathscr{F})\to\overline{\mathbb{R}}^n$ is a random vector. We also denote the Lebesgue measure on $\overline{\mathbb{R}}^n$ by $\lambda$.

\subsubsection{Absolutely Continuous Random Variables}
The random vector $Z$ is called absolutely continuous if $P_Z$ is absolutely continuous with respect to the Lebesgue measure.

\subsubsection{Probability Density Functions}

Assume that $Z$ is absolutely continuous. It follows from the Radon-Nicodym  theorem that there exists a measurable function $f:\overline{\mathbb{R}}\to[0,\infty]$ for which $P_Z(E)=\int_Ef(z)\,\mathrm{d}\lambda(z)$ for any Borel set $E$ in $\overline{\mathbb{R}}$.  Here, $f$ is called the probability density function of $X$. Now, if the random variable $X$ is absolutely continuous,  $\mathrm{Supp}(X)\subseteq\mathbb{R}$, and $f$ is Riemann integrable, then 
\[P_X((-\infty,t])=\int_{-\infty}^tf(x)\,\mathrm{d}x\]
for any $t\in\mathbb{R}$.

\subsubsection{Distribution Function of a Random Variable} 

Define $F_X:\overline{\mathbb{R}}\to[0,1]$ by setting $F_X(t)=P_X([-\infty,t])$ for any $t\in\overline{\mathbb{R}}$. Then, $F_X$ is called the distribution function of $X$. Now, let $X$ be absolutely continuous and $\mathrm{Supp}(X)\subseteq\mathbb{R}$. Then 
\[F_X(t)=\int_{-\infty}^tf(x)\,\mathrm{d}x\]
for any $t\in\mathbb{R}$, where $f$ is the probability density function of $X$ which is assumed to be Riemann integrable. Note that if $f$ is continuous, then $F_X$ is differentiable and we have that
\[\frac{\mathrm{d}F_X}{\mathrm{d}x}(t)=f(t).\]
\subsubsection{Definition of a Continuous Random Variable}

Let $(\Omega,\mathscr{F},P)$ be a probability space. Then, the random variable $X$ is called continuous if $F_X$ is continuous. Now, we show that $X$ is continuous if and only if $P_X(t)=0$ for any $t\in\overline{\mathbb{R}}$. First, assume that $X$ is continuous and $t_0\in\overline{\mathbb{R}}$. Let $\{t_n\}_{n=1}^{\infty}$ be an increasing sequence in $\overline{\mathbb{R}}$ converging to $t_0$. Now, if $E_n=(t_n, t_0]$, then we have that $E_n\downarrow \{t_0\}$, and hence we have that
\[0=\lim_{n\to\infty}(F_X(t_0)-F_X(t_n))=\lim_{n\to\infty}P_X((t_n,t_0])=\lim_{n\to\infty}P_X(E_n)=P_X(t_0).\]
One could see that the converse holds as well. 

Let $X$ be absolutely continuous. We show that $X$ is continuous. To do so, let $t\in\overline{\mathbb{R}}$. Then, the Lebesgue measure of $\{t\}$ is 0, and hence $P_X(t)=0$. It follows that $X$ is a continuous random variable.

Continuous random variables are often absolutely continuous. Hence, from now on, by a continuous random variable, we mean an absolutely continuous random variable. 

\subsubsection{Examples of Continuous Random Variables}

As we discussed, a continuous random variable is uniquely determined by its probability density function (pdf). Thus, we introduce the pdfs of some famous random variables in the following:

\begin{align*}
	f(x)&:=\left\{\begin{array}{ll}
		\frac{1}{b-a}&,a\le x\le b\\
		0&,\text{otherwise}
		\end{array}\right.& (\text{Unifrom Random Variable on } [a,b]),\\
	g(x)&:=\left\{\begin{array}{ll}
	\lambda e^{-\lambda x}&,x\ge 0\\
		0&,\text{otherwise}
	\end{array}\right.& \left(\begin{array}{l}\text{Exponential Random Variable}\\ \text{with the parameter $\lambda>0$ }\end{array}\right),\\
h(x):&=\frac{1}{\sqrt{2\pi}\sigma}e^{-\dfrac{(x-\mu)^2}{2\sigma^2}}& \left(\begin{array}{l}\text{Normal Random Variable}\\ \text{with the variance $\sigma^2$ ($\sigma>0$)}\\\text{and the mean $\mu$ }\end{array}\right).
\end{align*}

\subsubsection{Continuous Joint Distributions of  Random Variables}

 Assume that the random vector $(X,Y)$ is absolutely continuous, and $f_{X,Y}$ is the pdf of $(X,Y)$. Then, the pdfs of $X$ and $Y$ could be obtained as follows:
\[f_X(x)=\int_{\overline{\mathbb{R}}}f_{X,Y}(x,y)\,\mathrm{d}\lambda(y),\qquad f_Y(y)=\int_{\overline{\mathbb{R}}}f_{X,Y}(x,y)\,\mathrm{d}\lambda(x),\]
since for instance we have that
\begin{align*}
	\int_{E}\int_{\overline{\mathbb{R}}}f_{X,Y}(x,y)\,\mathrm{d}\lambda(y)\mathrm{d}\lambda(x)=\int_{E\times\overline{\mathbb{R}}}f_{X,Y}(x,y)\,\mathrm{d}\lambda(x,y)=P_{X,Y}(E\times\overline{\mathbb{R}})=P_X(E)
\end{align*}
for any Borel set $E$ in $\overline{\mathbb{R}}$. Indeed,
we validate the first equality of the above expression  for simple functions $f_{X,Y}(x,y)$, and then we generalize it for arbitrary functions $f_{X,Y}(x,y)$ using the dominated convergence theorem (another way to justify this equality  is the Fubini theorem in measure theory).

Now, assume that $X$ and $Y$ are independent. Then, for any two Borel sets $E_1$ and $E_2$ in $\overline{\mathbb{R}}$, we have that $P_{X,Y}(E_1\times E_2)=P_X(E_1)P_Y(E_2)$. Now, we claim that $f_{X,Y}=f_Xf_Y$. To do so, we note that
\begin{align*}
	\int_{E_1\times E_2}f_X(x)f_Y(y)\,\mathrm{d}\lambda(x,y)&=\int_{E_1}\int_{E_2}f_X(x)f_Y(y)\,\mathrm{d}\lambda(x)\mathrm{d}\lambda(y)\\
	&=\left(\int_{E_1}f_X(x)\,\mathrm{d}\lambda(x)\right)\left(\int_{E_2}f_Y(y)\,\mathrm{d}\lambda(y)\right)\\
	&=P_X(E_1)P_Y(E_2)=P_{X,Y}(E_1\times E_2).
	\end{align*}
It follows that for any Borel set $E'$ in $\overline{\mathbb{R}}^n$, we have that 
\[\int_{E'}f_X(x)f_Y(y)\,\mathrm{d}\lambda(x,y)=P_{X,Y}(E'),\]
which implies that $f_{X,Y}=f_Xf_Y$. 

\subsubsection{Conditional Continuous Random Variables}

Let $E\in\mathscr{F}$ and $X$ be continuous (as we mentioned before, here we mean being absolutely continuous). If $P(E)>0$ and $P^E$ be a conditional probability measure (for its definition look at Appendix  \ref{conofcon}), then $P_X^E$  is absolutely continuous, since for any Borel set $E'$ in $\overline{\mathbb{R}}$, we have that
\[P^E_X(E')P(E)=P((X\in E')\cap E)\le P_X(E').\]
Otherwise, if $P(E)=0$, then we define $P_X^E$ in such a way that it is absolutely continuous. In general, we can define $P_X^E$ as follows:
\[P_X^E(E')=\int_{E'}f_{X|E}(x|E)\,\mathrm{d}\lambda(x),\]
where $f_{X|E}(\,\cdot\,|E):\overline{\mathbb{R}}\to\overline{\mathbb{R}}$ is a non-negative measurable function satisfying the following identity:
\[\int_{\overline{\mathbb{R}}}f_{X|E}(x|E)\,\mathrm{d}\lambda(x)=1.\] 
In the case that $E$ is the event $Y^{-1}(y)$, we denote $f_{X|E}$ by $f_X^y$. Also,  we assume that $P(X\in E_1|Y=\,\cdot\,):\overline{\mathbb{R}}\to\overline{\mathbb{R}}$  is the probability function defined in Appendix \ref{con}. Indeed, $P(X\in E_1|Y=\,\cdot\,)$  is measurable and satisfies the following equality:
\[P((X\in E_1)\cap (Y\in E_2))=\int_{E_2}P(X\in E_1|Y=y)\,\mathrm{d}P_Y(y)\]
for any Borel sets $E_1$ and $E_2$ in $\overline{\mathbb{R}}$.
Now, we show that 
\[f_{X}^y(x|y)f_Y(y)=f_{X,Y}(x,y)\]
for any $x,y\in\overline{\mathbb{R}}$. To do so, it is enough to show that for any Borel sets $E_1$ and $E_2$ in $\overline{\mathbb{R}}$, we have that
\[P_{X,Y}(E_1\times E_2)=\int_{E_1\times E_2} f_{X}^y(x|y)f_Y(y)\,\mathrm{d}\lambda(x,y).\]
Denote the right side of the above identity by $I$. Then, we have that
\begin{align*}
	I&=\int_{E_2}\left(\int_{E_1}f_X^y(x|y)\,\mathrm{d}\lambda(x)\right)f_Y(y)\,\mathrm{d}\lambda(y)=\int_{E_2}P_X^y(E_1)f_Y(y)\mathrm{d}\lambda(y)
	\\&=\int_{E_2}P_X^y(E_1)\,\mathrm{d}P_Y(y)=\int_{E_2}P(X\in E_1|Y=y)\,\mathrm{d}P_Y(y)=P((X\in E_1)\cap (Y\in E_2)).
\end{align*}
Hence, we have the following:
\[f_{X}^y(x|y)f_Y(y)=f_{X,Y}(x,y)\]
for any $x,y\in\overline{\mathbb{R}}$.
\subsubsection{Expectation of Continuous Random Variables}

Let $X$ be continuous and $f_X$ be the pdf of $X$. We already know that 
\[\mathbb{E}(X)=\int_{\Omega}X(\omega)\,\mathrm{d}P.\]
By the change of variable theorem, we have that
\[\mathbb{E}(X)=\int_{\overline{\mathbb{R}}}x\,\mathrm{d}P_X(x).\]
Therefore, we have that
\[\mathbb{E}(X)=\int_{\overline{\mathbb{R}}}xf_X(x)\,\mathrm{d}\lambda(x).\]

\subsubsection{Conditional Expectation of Continuous Random Variables}

Let $X$ be continuous and $E\in\mathscr{F}$. We already now that if $P^E$ is a conditional probability measure, then
\[\mathbb{E}(X|E)=\int_{\Omega}X(\omega)\,\mathrm{d}P^E(\omega).\]
Now, by the change of variable theorem, we obtain 
\[\mathbb{E}(X|E)=\int_{\overline{\mathbb{R}}}x\,\mathrm{d}P^E_X(x).\]
 Assume that $P^E_X$ is absolutely continuous. Then, there exists a unique measurable function $g^E:\overline{\mathbb{R}}\to\overline{\mathbb{R}}$ for which \[P^E_X(E')=\int_{E'}g^E(x)\,\mathrm{d}\lambda(x)\] 
 for any Borel set $E'$ in $\overline{\mathbb{R}}$. We denote $g^E(x)$ by $f_{X|E}(x|E)$ for any $x\in\overline{\mathbb{R}}$. It follows that
 \[\mathbb{E}(X|E):=\int_{\overline{\mathbb{R}}}xf_{X|E}(x|E)\,\mathrm{d}\lambda(x).\]

\subsubsection{Total Law of Probability for Continuous Random Variables}

Let $E\in\mathscr{F}$ and $X$ be continuous. Assume that $f_X$ is the pdf of $X$. We already know the following version of the total law of probability:
\[P(E)=\int_{\overline{\mathbb{R}}}P(E|X=x)\,\mathrm{d}P_X(x)\]
which implies that
\[P(E)=\int_{\overline{\mathbb{R}}}P(E|X=x)f_X(x)\,\mathrm{d}\lambda(x).\]

\subsubsection{Total Law of Probability for Conditional Continuous Random Variables}

Let $E, E_0\in\mathscr{F}$ and $X$ be continuous. Assume that $P^E$ is a conditional probability measure, and  $P_X^E$  is  absolutely continuous. Further, assume that $f_{X|E}(x|E)$ is the pdf of $P_X^E$. We already know that 
\[P(E_0|E)=\int_{\overline{\mathbb{R}}}P(E_0|E,X=x)\,\mathrm{d}P_X^E(x).\]
It follows that
\[P(E_0|E)=\int_{\overline{\mathbb{R}}}P(E_0|E,X=x)f_{X|E}(x|E)\,\mathrm{d}\lambda(x).\]

\section{Generalized Probability Density Functions}\label{gpdf}

In this appendix, we  use the Dirac measure to unify both  discrete and continuous distributions (see \cite{capinski2004measure} and \cite{ccinlar2011probability}).

\subsubsection{ Discrete Random Variable Distributions via Dirac Measures}
Let $X:\Omega\to\overline{\mathbb{R}}$ be a discrete random variable and  $E\subseteq\Omega$. We know that 
\[P_X(E)=\sum_{x\in E}P_X(x).\]
Now, define the probability measure $\mu:\mathcal{P}(\Omega)\to[0,1]$ as follows:
\[\mu:=\sum_{x\in\mathrm{Supp}(X)}P_X(x)\delta_x.\]
One could see that $\mu=P_X$, and hence
\[P_X(E)=\int_E\,\mathrm{d}\mu.\]

\subsubsection{Definition of a Generalized Probability Density Function}\label{Definition of a Generalized Probability Density Function}

Let $(\Omega,\mathscr{F},P)$ be a probability space, and $X:\Omega\to\overline{\mathbb{R}}$ be a  random variable. Assume that 
\[\mu=t\lambda+\sum_{n=1}^{\infty}t_n\delta_{x_n},\]
where $\lambda$ is the Lebesgue measure, and $t_n,x_n\in\overline{\mathbb{R}}$ are non-negative for any $n$. Then, $\mu$ is a measure. If $P_X$ is absolutely continuous with respect to $\mu$, then by the Radon-Nykodim theorem, there exists a unique non-negative  measurable function $f_X:\overline{\mathbb{R}}\to\overline{\mathbb{R}}$ in such a way that
\[P_X(E)=\int_{E}f_X(x)\,\mathrm{d}\mu(x)\]
for any $E\in\mathscr{B}(\overline{\mathbb{R}})$. The function $f_X$ is called the generalized density function of $X$ with respect to $\mu$. Obviously, any pdf is a generalized pdf. 
For an example, assume that $\mathrm{Im}(f_X)=\{f_X(x):x\in\overline{\mathbb{R}}\}\subseteq\mathbb{R}$, and $(f_X)_{|_{\mathbb{R}}}:\mathbb{R}\to\mathbb{R}$ is Riemann integrable, and $\mu=\lambda+\delta_a$ for some $a\in\mathbb{R}$ with $a\le t$, where $t\in\mathbb{R}$. Then, we have that

\begin{align*}
	\int_{[-\infty,t]}f_X(x)\,\mathrm{d}\mu(x)&=\int_{[-\infty,t]}f_X(x)\,\mathrm{d}\lambda(x)+\int_{[-\infty,t]}f_X(x)\,\mathrm{d}\delta_a(x)\\
	&=\int_{(-\infty,t]}f_X(x)\,\mathrm{d}\lambda(x)+f_X(a)\\
	&=\int_{-\infty}^tf_X(x)\,\mathrm{d}x+f_X(a).
\end{align*}

\subsubsection{Dirac Delta Function}\label{deltadiracfunction}

The Dirac delta function  is a function $\delta:\overline{\mathbb{R}}\to\overline{\mathbb{R}}$ which takes the value $0$ everywhere except for $x=0$ which takes $\infty$. It is assumed that $\delta$ has the following properties:
\begin{itemize}
	\item 
	$\int_{-\infty}^{\infty}\delta(x-a)\,\mathrm{d}x=1$ for any $a\in\mathbb{R}$, and
	\item
	$\int_{-\infty}^{\infty}g(x)\delta(x-a)\,\mathrm{d}x=g(a)$ for any $a\in\mathbb{R}$ and any continuous function $g:\mathbb{R}\to{\mathbb{R}}$ for which the smallest closed set of all points $x\in\mathbb{R}$ with $g(x)\neq 0$ is bounded. 
\end{itemize} 
The first property implies that $\delta(x-a)$ is a pdf. 

A main challenge here is that $\delta$ is not a function on $\mathbb{R}$ to $\mathbb{R}$, and hence the Riemann integral could not be defined for $\delta$. Hence, one could interpret the above Riemann integrals as Lebesgue integrals.  In this case, we note that in the Lebesgue integral removing a set of zero measure (such as a point) does not affect the value of an integral. Hence, we must have $\int_{\overline{\mathbb{R}}}\delta(x)\,\mathrm{d}\lambda(x)=0$, a contradiction! 
In fact, a true interpretation of the  Dirac delta function is to consider it as a certain linear functional. Here, we do not discuss the details of a true mathematical definition of $\delta$ (for details, see \cite{lax2014functional}). However, we can interpret  the integrals appeared above as symbolic integrals and define them as follows (for a more detailed and comprehensive reference, one could see \cite[Section 5]{loffler2019brownian}):
\[	\int_{-\infty}^{t}\delta(x-a)\,\mathrm{d}x:=\int_{[-\infty,t]}\,\mathrm{d}\delta_a(x),\quad \int_{-\infty}^{t}g(x)\delta(x-a)\,\mathrm{d}x:=\int_{[-\infty,t]}g(x)\,\mathrm{d}\delta_a(x),\] 
where $\delta_a$ is the Dirac measure with respect to $a$. The two main properties of the Dirac delta function are satisfies, since
\[\int_{\overline{\mathbb{R}}}\,\mathrm{d}\delta_a(x)=1,\qquad \int_{\overline{\mathbb{R}}}g(x)\,\mathrm{d}\delta_a(x)=g(a).\]
In general, we define
\[\int_{-\infty}^tg(x)\left(\sum_{n=1}^{\infty}t_n\delta(x-a_n)\right)\,\mathrm{d}x:=\sum_{n=1}^{\infty}t_n\int_{[-\infty,t]}g(x)\,\mathrm{d}\delta_{a_n}(x),\]
where $t_n,a_n\in\mathbb{R}$ for any $n$.
Hence, we have that
\begin{align*}
\int_{[-\infty,t]}g(x)\,\mathrm{d}\left(t_0\lambda+\sum_{n=1}^{\infty}t_n\delta_{a_n}\right)(x)&=\int_{[-\infty,t]}t_0g(x)\,\mathrm{d}\lambda(x)+\sum_{n=1}^{\infty}t_n\int_{[-\infty,t]}g(x)\,\mathrm{d}\delta_{a_n}(x)\\
&=\int_{-\infty}^tt_0g(x)\,\mathrm{d}x+\int_{-\infty}^tg(x)\left(\sum_{n=1}^{\infty}t_n\delta(x-a_n)\right)\,\mathrm{d}x\\
&=\int_{-\infty}^tg(x)\left(t_0+\sum_{n=1}^{\infty}t_n\delta(x-a_n)\right)\,\mathrm{d}x.
\end{align*}
Indeed, we have that
\[\int_{-\infty}^t\delta(x-a)\,\mathrm{d}x=\int_{[-\infty,t]}\,\mathrm{d}\delta_a(x)=U_a(t),\]
where $U_{a}$ is the Heaviside step function defined as follows:

\[U_{a}(t)=\left\{\begin{array}{ll}
	0  & ,t< a \\
	1  &  ,t\ge a
\end{array}\right..\]
It follows that $U_a(t)$ is the distribution function associated to  $\delta(x-a)$, and we have the following symbolic derivative:
\[\frac{\mathrm{d}U_a}{\mathrm{d}x}(t)=\delta(t-a).\] 
Note that $U_a$ is not really differentiable at $t=a$, that is why we used the word "symbolic" above. 
\end{document}